\documentclass{article} 
\usepackage{fullpage,amsmath,amsthm,amssymb,hyperref,cleveref,graphicx,url,rotating}
\usepackage{supertabular,algorithm2e} 

\graphicspath{{Fig/}}

\newtheorem{Theorem}{Theorem}

\newtheorem{Corollary}{Corollary}
\newtheorem{Proposition}{Proposition}
\newtheorem{Remark}{Remark}

\begin{document}

\title{On Voronoi diagrams and dual Delaunay complexes on the information-geometric Cauchy manifolds}

\author{Frank Nielsen\footnote{E-mail:{\tt \url{Frank.Nielsen@acm.org}} / Web: {\tt \url{https://franknielsen.github.io/} }}\\ Sony Computer Science Laboratories, Inc\\ Japan}

\date{}

\maketitle

\begin{abstract}
We study the Voronoi diagrams of a finite set of Cauchy distributions and their dual complexes from the viewpoint of information geometry by considering the Fisher-Rao distance, the Kullback-Leibler divergence, the chi square divergence, and a flat divergence 
derived from Tsallis' quadratic entropy related to the conformal flattening of the Fisher-Rao curved geometry. 
We prove that the Voronoi diagrams of the Fisher-Rao distance, the chi square divergence, and the Kullback-Leibler divergences all coincide with a hyperbolic Voronoi diagram on the corresponding Cauchy location-scale parameters, and that the dual Cauchy hyperbolic Delaunay complexes are Fisher orthogonal to the Cauchy hyperbolic Voronoi diagrams. 
The dual Voronoi diagrams with respect to the dual forward/reverse flat divergences amount to dual Bregman Voronoi diagrams, and their dual complexes are regular triangulations. The primal Bregman-Tsallis Voronoi diagram corresponds to the hyperbolic Voronoi diagram and the dual Bregman-Tsallis Voronoi diagram coincides with the ordinary Euclidean Voronoi diagram.
Besides, we prove that the square root of the Kullback-Leibler divergence between Cauchy distributions yields a metric distance which is Hilbertian for the Cauchy scale families.
\end{abstract}

\def\PD{\mathrm{PD}}
\def\calM{\mathcal{M}}
\def\calH{\mathcal{H}}
\def\Jac{\mathrm{Jac}}
\def\Eucl{\mathrm{Eucl}}
\def\bbD{\mathbb{D}}
\def\mattwotwo#1#2#3#4{\left[\begin{array}{ll}#1 & #2 \cr #3 & #4\end{array}\right]}
\def\arccosh{\mathrm{arccosh}}
\def\leftsup#1{{}^{#1}}
\def\supleft#1{{}^{#1}}
\def\Vor{\mathrm{Vor}}
\def\calC{\mathcal{C}}
\def\calN{\mathcal{N}}
\def\calS{\mathcal{S}}
\def\calY{\mathcal{Y}}
\def\calP{\mathcal{P}}
\def\ds{\mathrm{d}s}
\def\JB{\mathrm{JB}}
\def\Bhat{\mathrm{Bhat}}
\def\flat{\mathrm{flat}}
\def\Bi{\mathrm{Bi}}
\def\bbH{\mathbb{H}}
\def\FR{\mathrm{FR}}
\def\hcross{{h^\times}}
\def\eqdef{{:=}}
\def\hcrossent#1#2{{\hcross}\left(#1:#2\right)}
\def\arctanh{\mathrm{arctanh}}
\def\der{\mathrm{d}}
\def\dmu{\mathrm{d}\mu}
\def\dx{\mathrm{d}x}
\def\dy{\mathrm{d}y}
\def\dt{\mathrm{d}t}
\def\KL{\mathrm{KL}}
\def\kl{\mathrm{kl}}
\def\bbX{\mathbb{X}}
\def\bbQ{{\mathbb{Q}}}
\def\bbR{{\mathbb{R}}}
\def\calF{\mathcal{F}}
\def\calL{\mathcal{L}}
\def\calX{\mathcal{X}}
\def\st{\ :\ }


\section{Introduction}

Let $\calP=\{P_1,\ldots,P_n\}$ be a finite  set of points in a space $\bbX$ equipped with a measure of {\em dissimilarity}
$D(\cdot,\cdot):\bbX\times\bbX  \rightarrow\bbR_+$. The {\em Voronoi diagram}~\cite{VorOkabe-2009} of $\calP$ partitions $\bbX$ into elementary {\em Voronoi cells} $\Vor(P_1),\ldots, \Vor(P_n)$ (also called Dirichlet cells~\cite{aurenhammer1991voronoi}) such that
\begin{equation}
\Vor_D(P_i)\eqdef \left\{ X\in\bbX,\quad D(P_i,X)\leq D(P_j,X),\quad \forall j\in\{1,\ldots,n\}\right\}
\end{equation} 
denotes the {\em proximity cell} of point {\em generator} $P_i$ (also called {\em Voronoi site}), i.e., the locii of points $X\in\bbX$ closer with respect to $D$ to $P_i$ than to any other generator $P_j$.

When the dissimilarity $D$ is chosen as the Euclidean distance $\rho_E$, we recover the ordinary Voronoi diagram~\cite{VorOkabe-2009}.
The Euclidean distance $\rho_E(P,Q)$ between  two points $P$ and $Q$ is defined as
\begin{equation}
\rho_E(P,Q)=\|p-q\|_2,
\end{equation}
where $p$ and $q$ denote the Cartesian coordinates of point $P$ and $Q$, respectively, and $\|\cdot\|_2$ the $\ell_2$-norm.
Figure~\ref{fig:VorDeEucl} (left) displays the Voronoi cells of an ordinary Voronoi diagram for a given set of generators.


The Voronoi diagram and its dual {\em Delaunay complex}~\cite{cheng2012delaunay} are fundamental data structures of computational geometry~\cite{BY-1998}. These core geometric data-structures find many applications in robotics, 3D reconstruction, geographic information systems (GISs), etc. See the textbook~\cite{VorOkabe-2009} for some of their applications. 
The {\em Delaunay simplicial complex} is obtained by drawing a straight edge between two generators iff their Voronoi cells share an edge (Figure~\ref{fig:VorDeEucl}, right).
In Euclidean geometry, the {\em Delaunay simplicial complex} triangulates the convex hull of the generators, 
and is therefore called the {\em Delaunay triangulation}.
Figure~\ref{fig:VorDeEucl} depicts the dual Delaunay triangulations corresponding to ordinary Voronoi diagrams.
In general, when considering arbitrary dissimilarity $D$, the Delaunay simplicial complex  may {\em not} triangulate the convex hull of the generators (see~\cite{bogdanov2013hyperbolic} and \S\ref{sec:Vor}).

\begin{figure}
\centering
\fbox{\includegraphics[bb=0 0 512 512,width=0.3\columnwidth]{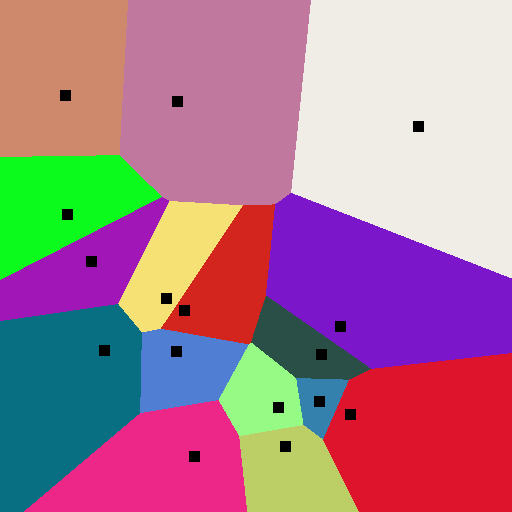}}
\fbox{\includegraphics[width=0.3\columnwidth]{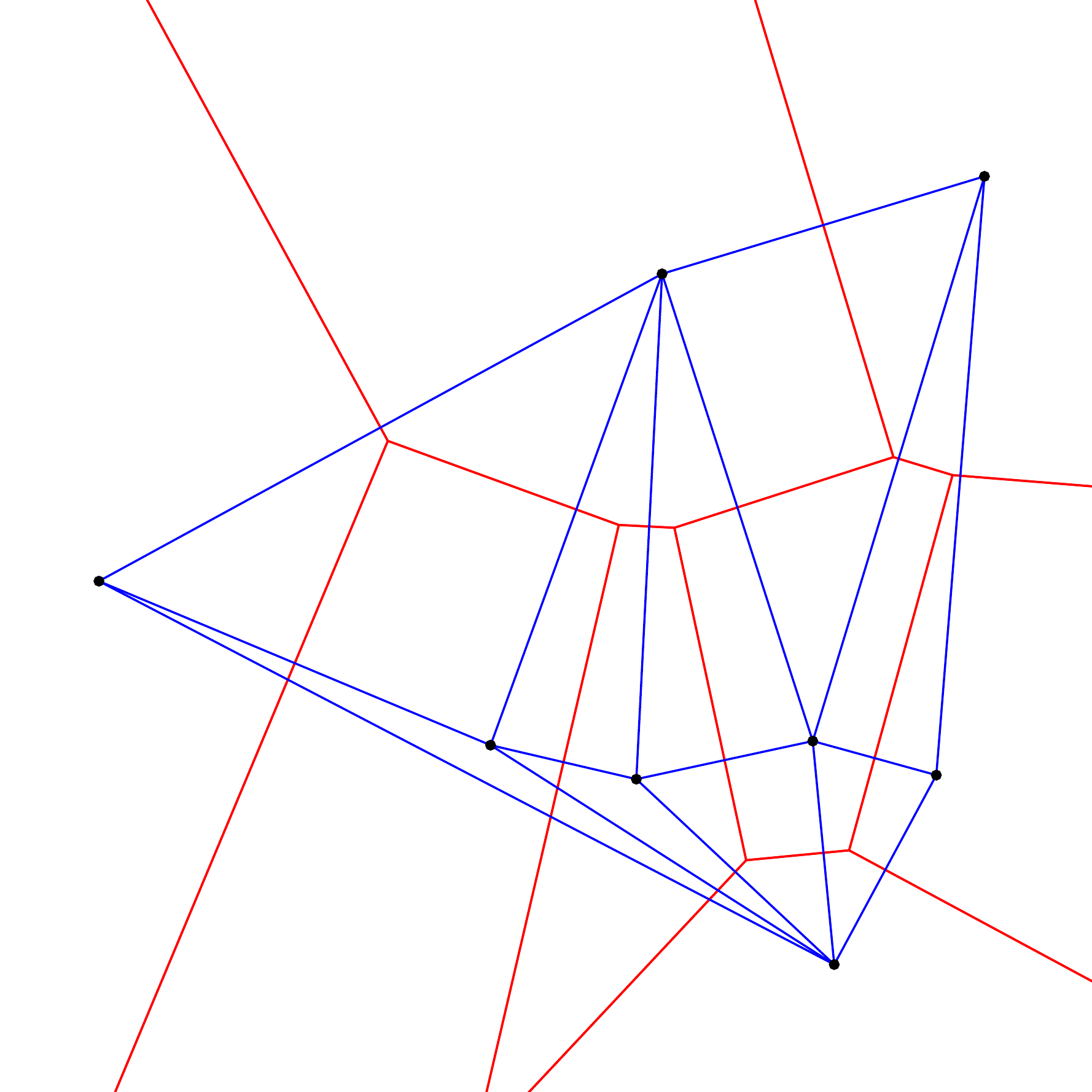}}
\fbox{\includegraphics[width=0.3\columnwidth]{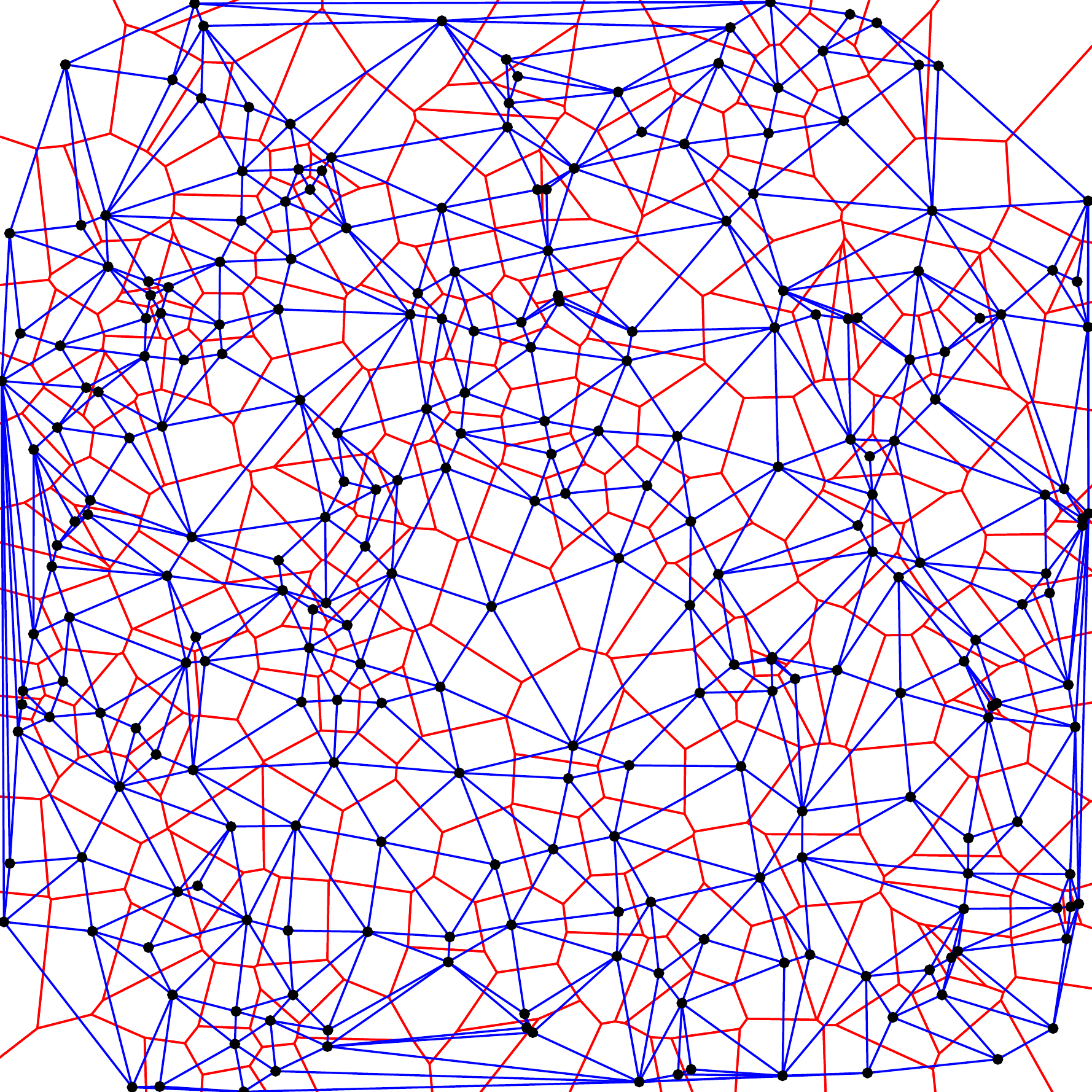}}

\caption{Euclidean Voronoi diagram of a set of generators (black square) in the plane with colored Voronoi cells (left).
Euclidean Voronoi diagrams (red) and their dual Delaunay triangulations (blue) for $n=8$ points (middle) and $n=256$ points (right).} 
\label{fig:VorDeEucl}
\end{figure}

When the dissimilarity is {\em oriented} or {\em asymmetric}, i.e., $D(P,Q)\not=D(Q,P)$, one can define the {\em reverse} or {\em dual} dissimilarity $D^*(P,Q)\eqdef D(Q,P)$.
This duality is termed {\em reference duality} in~\cite{zhang2004divergence}, and is an involution:
\begin{equation}
(D^*)^*(P,Q)=D(P,Q).
\end{equation} 
The dissimilarity $D(P:Q)$ is called the {\em forward} dissimilarity.

In the remainder, we shall use the `:' notational convention~\cite{IG-2016} between the arguments of the dissimilarity to emphasize that a dissimilarity $D$ is asymmetric: $D(P:Q)\not= D(Q:P)$.
For an oriented dissimilarity $D(\cdot:\cdot)$, we can define two types of {\em dual Voronoi cells}  as follows:
\begin{eqnarray}
\Vor_D(P_i) &\eqdef& \left\{ X\in\bbX,\quad D(P_i:X)\leq D(P_j:X),\quad \forall j\in\{1,\ldots,n\}\right\},
\end{eqnarray}
and
\begin{eqnarray}
\Vor_D^*(P_i) &\eqdef& \left\{ X\in\bbX\quad D(X:P_i)\leq D(X:P_j),\quad \forall j\in\{1,\ldots,n\}\right\},\\
         &=&  \left\{ X\in\bbX\quad D^*(P_i:X)\leq D^*(P_j:X),\quad \forall j\in\{1,\ldots,n\}\right\},\\
				 &=&  \Vor_D^*(P_i)=\Vor_{D^*}(P_i).
\end{eqnarray}

That is, the dual Voronoi cell $\Vor_D^*(P_i) $ with respect to a dissimilarity $D$ is the primal Voronoi cell $\Vor_{D^*}(P_i)$ for the dual (reverse) dissimilarity $D^*$.

In general, we can build a Voronoi diagram as a {\em minimization diagram}~\cite{boissonnat2006curved} by defining the $n$ functions
 $f_i(X)\eqdef D(P_i:X)$. Then $X\in \Vor_D(P_i)$ iff $f_i(X)\leq f_j(X)$ for all   $ j\in\{1,\ldots,n\}$.
Thus by building the {\em lower envelope}~\cite{boissonnat2006curved}  of the $n$ functions $f_1(X), \ldots, f_n(X)$, we can retrieve   the Voronoi diagram.

An important class of {\em smooth} asymmetric dissimilarities are the  Bregman divergences~\cite{Bregman-1967}.
A {\em Bregman divergence} $B_F$ is defined for a smooth and strictly convex functional generator $F(\theta)$ by
\begin{equation}
B_F(\theta_1:\theta_2) \eqdef F(\theta_1)-F(\theta_2)-(\theta_1-\theta_2)^\top\nabla F(\theta_2),
\end{equation}
where $\nabla F$ denotes the gradient of $F$.
In information geometry~\cite{IG-CalinUdriste-2014,IG-2016,EIG-2018}, Bregman divergences are the canonical divergences of dually flat spaces~\cite{IG-2016}. 
Dually flat spaces generalize the (self-dual) Euclidean geometry obtained for the generator $F_{\mathrm{Eucl}}(\theta)=\frac{1}{2}\theta^\top\theta$.
In information sciences, dually flat spaces can be obtained, for example,  
as the induced information geometry of the Kullback-Leibler divergence~\cite{CT-2012} of an exponential family manifold~\cite{eguchi1992geometry,IG-2016} or a mixture manifold~\cite{MCIG-2018}.
The dual Bregman Voronoi diagrams and their dual regular complexes have been studied in~\cite{BVD-2010}.

In this paper, we study the Voronoi diagrams induced by the Fisher-Rao distance~\cite{Rao-1945,atkinson1981rao,pinele2020fisher}, the Kullback-Leibler (KL) divergence~\cite{CT-2012} and the chi square distance~\cite{NN-fdiv-2013} for the family $\calC$ of Cauchy distributions. Cauchy distributions also called Lorentzian distributions in the literature~\cite{naudts2009,matsuzoe2014hessian}.

The paper is organized with our main contributions as follows:

In Section~\ref{sec:IGCauchy}, we concisely review the information geometry of the Cauchy family: 
We first describe the hyperbolic Fisher-Rao geometry in \S\ref{sec:FR} and make a connection between the Fisher-Rao distance and the chi square divergence, then we point out the remarkable fact that any $\alpha$-geometry coincides with the Fisher-Rao geometry (\S\ref{sec:alphageo}), and we finally present  dually flat geometric structures on the Cauchy manifold related to Tsallis' quadratic entropy~\cite{Tsallis-1988,tsallis2009introduction}
 which amount to a conformal flattening of the Fisher-Rao geometry (\S\ref{sec:dfs}).
Section~\ref{sec:KLmetrization}  proves that the square root of the KL divergence between any two Cauchy distributions yields a metric distance (Theorem~\ref{thm:sqrtKL}), and that this metric distance can be isometrically embedded in a Hilbert space for the case of   Cauchy scale families (Theorem~\ref{thm:sqrtKLscale}). 
Section~\ref{sec:Vor} shows that the Cauchy Voronoi diagrams induced either by the Fisher-Rao distance, the chi-square divergence, or the Kullback-Leibler divergence (and its square root metrization) all coincide with a hyperbolic Voronoi diagram~\cite{HVD-2010} calculated on the Cauchy 2D location-scale parameters.
This result yields a practical and efficient construction algorithm of hyperbolic Cauchy Voronoi diagrams~\cite{HVD-2010,HVD-2014} (Theorem~\ref{thm:CVD}) and their dual hyperbolic Cauchy Delaunay complexes (explained in details in Appendix~\ref{sec:hpd}). 
We prove that the hyperbolic Cauchy Voronoi diagrams are Fisher orthogonal to the dual Cauchy Delaunay complexes (Theorem~\ref{thm:CVDortho}).
In \S\ref{sec:dualVor}, we show that the primal Voronoi diagram with respect to the flat divergence coincides with the hyperbolic Voronoi diagram, and that the Voronoi diagram with respect to the reverse flat divergence matches the ordinary Euclidean Voronoi diagram.
Finally, we conclude this work in \S\ref{sec:Concl}.

\section{Information geometry of the Cauchy family}\label{sec:IGCauchy}

We start by reporting the Fisher-Rao geometry of the Cauchy manifold (\S\ref{sec:FR}), 
then show that all $\alpha$-geometries coincide with the Fisher-Rao geometry (\S\ref{sec:alphageo}).
Then we recall that we can associate an information-geometric structure to any  parametric divergence (\S\ref{sec:divgeo}),
and finally dually flatten this Fisher-Rao curved geometry using Tsallis's quadratic entropy~\cite{Tsallis-1988,tsallis2009introduction} (\S\ref{sec:dfs}) and a conformal Fisher metric.

\subsection{Fisher-Rao geometry of the Cauchy manifold\label{sec:FR}} 

{\em Information geometry}~\cite{IG-CalinUdriste-2014,IG-2016,EIG-2018} investigates the geometry of families of probability measures.
The 2D family $\calC$ of Cauchy distributions 
\begin{equation}
\calC \eqdef\left\{ 
p_\lambda(x)\eqdef \frac{s}{\pi (s^2+(x-l)^2)},\quad \lambda\eqdef(l,s)\in\bbH\eqdef\bbR\times\bbR_+
 \right\},
\end{equation}
is a {\em location-scale family}~\cite{murray1993differential} (and also a univariate {\em elliptical distribution family}~\cite{Mitchell-1988})  where $l\in\bbR$ and $s>0$ denote the {\em location} parameter and the {\em scale} parameter, respectively:
\begin{equation}
p_{l,s}(x)\eqdef\frac{1}{s}p\left(\frac{x-l}{s}\right),
\end{equation}
where 
\begin{equation}
p(x)\eqdef\frac{1}{\pi (1+x^2)}=:p_{0,1}(x)
\end{equation}
is the {\em Cauchy standard distribution}.

Let $l_\lambda(x)\eqdef \log p_\lambda(x)$ denote the {\em log density}.
The parameter space $\bbH\eqdef\bbR\times\bbR_+$ of the Cauchy family is called the {\em upper plane}.
The {\em Fisher-Rao geometry}~\cite{hotelling1930spaces,Rao-1945,pinele2020fisher} of $\calC$ consists in modeling $\calC$ as a Riemannian manifold $(\calC,g_\FR)$ by choosing the {\em Fisher Information metric}~\cite{IG-2016} (FIm)  
\begin{equation}
g_\FR(\lambda) = [g_{ij}^\FR(\lambda)],\quad  g_{ij}^\FR(\lambda) \eqdef E_{p_\lambda}\left[\partial_i l_\lambda(x)\partial_j l_\lambda(x) \right],
\end{equation}
 as the Riemannian metric tensor, where  
$\partial_m:=\frac{\partial}{\partial\lambda_m}$ for $m\in\{1,2\}$  (i.e., $\partial_1=\frac{\partial}{\partial l}$ and $\partial_2=\frac{\partial}{\partial s}$).  The matrix $[g_{ij}^\FR]$ is called the {\em Fisher Information Matrix} (FIM), and is the expression of the FIm tensor in a local coordinate system $\{e_1, e_2\}$: $g_{ij}^\FR(\lambda)=g(e_i,e_j)$ with $i,j\in\{1,2\}$.

The {\em Fisher-Rao distance} 
$\rho_{\FR}[p_{\lambda_1},p_{\lambda_2}]=\rho_{\FR}[p_{l_1,s_1},p_{l_2,s_2}]$
is then defined as the {\em Riemannian geodesic length distance} on the Cauchy manifold $(\calC,g_\FR)$:
\begin{equation}
\rho_\FR\left(p_{\lambda_1}\left(x\right), p_{\lambda_2}\left(x\right)\right)=
\min _{\substack{\lambda(s)\\ \mbox{such that}\\ \lambda(0)=\lambda_{1},
\lambda(1)=\lambda_{2}   }} 
\int_{0}^{1} \sqrt{ 
\left(\frac{\mathrm{d}\lambda(t)}{\mathrm{d} t}\right)^{T} g_\FR(\lambda(s)) \frac{\mathrm{d} \lambda(t)}{\mathrm{d} t}
} \mathrm{d}t. 
\end{equation}

The Fisher information metric tensor for the Cauchy family~\cite{Mitchell-1988} is
\begin{equation}\label{eq:mtCauchy}
g_\FR(\lambda)= g_\FR(l,s)=\frac{1}{2s^2}\left[
\begin{array}{cc}
1 & 0\\
0 & 1\\
\end{array}
\right],
\end{equation}
where $\lambda=(l,s)\in\bbH$.

A {\em generic formula} for the Fisher-Rao distance between two univariate elliptical distributions is reported in~\cite{Mitchell-1988}.
This formula when instantiated for the Cauchy distributions yields the following closed-form formula for the Fisher-Rao distance:
\begin{equation}
\rho_{\FR}[p_{l_1,s_1},p_{l_2,s_2}]=\frac{1}{\sqrt{2}} \left|\log\frac{\tan\left(\frac{\psi_1}{2}\right)}{\tan\left(\frac{\psi_2}{2}\right)}\right|,
\end{equation}
where 
\begin{eqnarray}
\psi_i&=&\mathrm{arcsin}\left(\frac{s_i}{A}\right),\quad i\in\{1,2\},\\
A^2 &=& s_1^2+\frac{\left( (l_2-l_1)^2-(s_1^2-s_2^2)\right)^2}{4(l_2-l_1)^2}.
\end{eqnarray}

However, by noticing that the metric tensor for the Cauchy family (Eq.~\ref{eq:mtCauchy}) is equal to the {\em scaled} metric tensor $g_P$ of the Poincar\'e (P) hyperbolic upper plane~\cite{anderson2006hyperbolic}:
\begin{equation}
g_P(x,y)=\frac{1}{y^2}\left[
\begin{array}{cc}
1 & 0\\
0 & 1\\
\end{array}
\right],
\end{equation}
we get  a relationship between the square infinitesimal lengths (line elements) 
$\ds_\FR^2=\frac{\mathrm{d}l^2+\mathrm{d}s^2}{2s^2}$ and 
$\ds_P^2=\frac{\mathrm{d}x^2+\mathrm{d}y^2}{y^2}$ as follows:
\begin{equation}
\ds_\FR=\frac{1}{\sqrt{2}}\ds_P.
\end{equation}

It follows that the Fisher-Rao distance between two Cauchy distributions is simply obtained by {\em rescaling} the 2D hyperbolic distance expressed in the 
Poincar\'e upper plane~\cite{anderson2006hyperbolic}:
\begin{equation}
\rho_{\FR}[p_{l_1,s_1},p_{l_2,s_2}] =  \frac{1}{\sqrt{2}} \rho_P(l_1,s_1;l_2,s_2)
\end{equation}
where
\begin{equation}
\rho_P(l_1,s_1;l_2,s_2) \eqdef    \mathrm{arccosh}\left(1+\delta(l_1,s_1,l_2,s_2)\right),
\end{equation}
with
\begin{equation}
\mathrm{arccosh}(x)\eqdef\log \left(x+\sqrt{x^{2}-1}\right),\quad x>1,
\end{equation}
and
\begin{equation}\label{eq:delta}
\delta(l_1,s_1;l_2,s_2) \eqdef \frac{(l_2-l_1)^2+(s_2-s_1)^2}{2s_1s_2}.
\end{equation}

This latter term $\delta$ shall naturally appear in \S\ref{sec:dfs} when studying the dually flat space obtained by conformal flattening the Fisher-Rao geometry.
The expression $\delta(l_1,s_1,l_2,s_2)$ of Eq.~\ref{eq:delta} can be interpreted as a {\em conformal divergence} 
for the squared Euclidean distance~\cite{nielsen2015total,conformaldiv-2015}.

We may also write the delta term using the 2D Cartesian coordinates $\lambda=(\lambda^{(1)},\lambda^{(2)})$ as:
\begin{equation}
\delta(\lambda_1,\lambda_2) \eqdef 
\frac{(\lambda_2^{(1)}-\lambda_1^{(1)})^2+(\lambda_2^{(2)}-\lambda_1^{(1)})^2}{2\lambda_1^{(2)}\lambda_2^{(2)}} = \frac{\|\lambda_1-\lambda_2\|_2^2}{2\lambda_1^{(2)}\lambda_2^{(2)}},
\end{equation}
where $\lambda\in\bbH$.

In particular, when $l_1=l_2$, we get the simplified Fisher-Rao distance for  Cauchy scale families:
\begin{equation}
\rho_{\FR}[p_{l,s_1},p_{l,s_2}]=\frac{1}{\sqrt{2}} \left|\log\left(\frac{s_1}{s_2}\right)\right|.
\end{equation}

\begin{Proposition}\label{prop:FRCauchy}
The Fisher-Rao distance between two Cauchy distributions is
\begin{equation*}
\rho_{\FR}[p_{l_1,s_1},p_{l_2,s_2}]=\left\{
\begin{array}{ll}
\frac{1}{\sqrt{2}} \left|\log\frac{s_1}{s_2}\right| & \mbox{when $l_1=l_2$},\\
\frac{1}{\sqrt{2}} \mathrm{arccosh}\left(1+\frac{(l_2-l_1)^2+(s_2-s_1)^2}{2s_1s_2}\right) & \mbox{when $l_1\not =l_2$.}
\end{array}
\right.
\end{equation*}
\end{Proposition}

The Fisher-Rao manifold of Cauchy distributions has {\em constant negative scalar curvature} $\kappa=-2$, see~\cite{Mitchell-1988} for detailed calculations.

\begin{Remark}
It is well-known that the Fisher-Rao geometry of location-scale families amount to a hyperbolic geometry~\cite{murray1993differential}.
For $d$-variate scale-isotropic Cauchy distributions 
$p_{\lambda}(x)$ with $\lambda=(l,s)\in\bbR^d\times\bbR$, the Fisher information metric is
$g_\FR(\lambda)=\frac{1}{2s^2}I$, where $I$ denotes the $(d+1)\times (d+1)$ identity matrix.
It follows that
\begin{equation}
\rho_{\FR}[p_{l_1,s_1},p_{l_2,s_2}] =  \frac{1}{\sqrt{2}} \mathrm{arccosh}\left(1+\Delta(l_1,s_1,l_2,s_2)\right),
\end{equation}
where
\begin{equation}
\Delta(l_1,s_1,l_2,s_2) \eqdef \frac{\|l_2-l_1\|^2_2+(s_2-s_1)^2}{2s_1s_2},
\end{equation}
where $\|\cdot\|_2$ is the $d$-dimensional Euclidean $\ell_2$-norm: $\|x\|=\sqrt{x^\top x}$. 
That is, $\rho_{\FR}[p_{l_1,s_1},p_{l_2,s_2}]$ is the scaled $d$-dimensional real hyperbolic distance~\cite{anderson2006hyperbolic} expressed in the Poincar\'e upper space model. 
\end{Remark}

Let us mention that recently the Riemannian geometry of location-scale models has also been studied from the complementary viewpoint of 
warped metrics~\cite{chen2017differential,Said2019}.

\subsection{The dualistic $\alpha$-geometry of the statistical Cauchy manifold}\label{sec:alphageo}

A {\em statistical manifold}~\cite{Lauritzen-1987} is a triplet $(M,g,T)$ where $g$ is a Riemannian metric tensor and $T$ is a cubic totally symmetric tensor (i.e., $T_{\sigma(i)\sigma(j)\sigma(k)}=T_{ijk}$ for any permutation $\sigma$). 
For a {\em parametric} family of probability densities $M=\{p_\lambda(x)\}$, the cubic tensor is called the {\em skewness tensor}~\cite{IG-2016}, and defined by:
\begin{equation}
T_{ijk}(\theta)\eqdef E_{p_\lambda}\left[\partial_i l_\lambda(x) \partial_j l\lambda(x) \partial_k l\lambda(x)\right].
\end{equation}

A statistical manifold structure $(M,g,T)$ allows one to construct Amari's dualistic {\em $\alpha$-geometry}~\cite{IG-2016} for  any $\alpha\in\bbR$:
Namely 
a quadruplet $(M,g_\FR,\nabla^{-\alpha},\nabla^\alpha)$ where $\nabla^{-\alpha}$ and $\nabla^{\alpha}$ are {\em dual torsion-free affine connections} coupled to the Fisher metric $g_\FR$ (i.e., $\nabla^{-\alpha}=(\nabla^{\alpha})^*$). We refer the reader to the textbook~\cite{IG-2016} and the overview~\cite{EIG-2018} for further details.

The Fisher-Rao geometry $(M,g_\FR)$ corresponds to the $0$-geometry, i.e., the self-dual geometry where $\nabla^0\eqdef\supleft{g}\nabla$ is the  {\em Levi-Civita metric connection}~\cite{IG-2016} induced by the metric tensor  (with $(\supleft{g}\nabla)^*=\supleft{g}\nabla$).
That is, we have
\begin{equation}
(\calC,g_\FR)=(\calC,g_\FR,\nabla^0,\nabla^0).
\end{equation}

In information geometry, the {\em invariance principle} states that the geometry should be invariant under the transformation of a random variable $X$ to $Y$ provided that $Y=t(X)$ is a sufficient statistics~\cite{IG-2016} of $X$. 
The $\alpha$-geometry $(M,g_\FR,\nabla^{-\alpha},\nabla^\alpha)$ and its special case of Fisher-Rao geometry are invariant geometry~\cite{IG-2016,EIG-2018} for any $\alpha\in\bbR$.

A remarkable fact is that all the $\alpha$-geometries of the Cauchy family coincide with the Fisher-Rao geometry since the cubic skewness tensor $T$ vanishes everywhere~\cite{Mitchell-1988}, i.e., $T_{ijk}=0$.
The non-zero  coefficients of the Christoffel symbols of the $\alpha$-connections (including the Levi-Civita metric connection derived from the Fisher metric tensor) are:
\begin{eqnarray}
\supleft{\alpha}\Gamma_{12}^1 &=& \supleft{\alpha}\Gamma_{21}^1 = \supleft{\alpha}\Gamma_{22}^2=-\frac{1}{s},\\
\supleft{\alpha}\Gamma_{11}^2 &=& \frac{1}{s}.
\end{eqnarray}

Thus all $\alpha$-geometries coincide and have constant negative scalar curvature $\kappa=-2$.
In other words, we cannot choose a value for $\alpha$ to make the Cauchy manifold dually flat~\cite{IG-2016}.
To contrast with this result, Mitchell~\cite{Mitchell-1988} reported values of $\alpha$ for which the $\alpha$-geometry is dually flat for some parametric location-scale families of distributions:
For example, it is well known that the manifold $\calN$ of univariate Gaussian distributions is $\pm 1$-flat~\cite{IG-2016}.
The manifold $\calS_k$ of $t$-Student's distributions with $k$ degrees of freedom is proven dually flat when $\alpha=\pm\frac{k+5}{k-1}$~\cite{Mitchell-1988}.
Dually flat manifolds are Hessian manifolds~\cite{shima2007geometry} with dual geodesics being straight lines in one of the two dual global affine coordinate systems. On a global Hessian manifold,  the canonical divergences are Bregman divergences. 
Thus these dually flat Bregman manifolds are computationally friendly~\cite{BVD-2010} as many techniques of computational geometry~\cite{BY-1998} can be naturally extended to these Hessian spaces (e.g., the smallest enclosing balls~\cite{nielsen2008smallest}). 

\subsection{Dualistic structures induced by a divergence}\label{sec:divgeo}
A {\em divergence} or {\em contrast function}~\cite{eguchi1992geometry} is a smooth parametric dissimilarity.
Let $\calM$ denote the manifold of its parameter space.
Eguchi~\cite{eguchi1992geometry} showed how to associate to any divergence $D$ a canonical 
information-geometric structure $(\calM,\supleft{D}g,\supleft{D}\nabla,\supleft{D}\nabla^*)$. 
Moreover, the construction allows to prove that $\supleft{D}\nabla^*=\supleft{D^*}\nabla$.
That is the dual connection $\supleft{D}\nabla^*$ for the divergence $D$ corresponds to the primal connection for the reverse divergence $D^*$
 (see~\cite{IG-2016,EIG-2018} for details).

Conversely, Matsumoto~\cite{matumoto-1993} proved that given an information-geometric structure $(\calM,g,\nabla,\nabla^*)$, one can build a divergence $D$ such that $(\calM,g,T)=(\calM,\supleft{D}g,\supleft{D}T)$ from which we can derive the structure 
$(\calM,\supleft{D}g,\supleft{D}\nabla,\supleft{D}\nabla^*)$.
Thus when calculating the Voronoi diagram $\Vor_D$ for an arbitrary divergence $D$, 
we may use the induced information-geometric structure $(\calM,\supleft{D}g,\supleft{D}\nabla,\supleft{D}\nabla^*)$ to investigate some of  the properties of the Voronoi diagram:
For example, is the bisector $\mathrm{Bi}_D$ $\supleft{D}\nabla$-autoparallel?, or is the bisector $\mathrm{Bi}_D$ of two generators orthogonal with respect to the metric $\supleft{D}g$ to their $\supleft{D}\nabla$-geodesic? 
Section~\ref{sec:Vor} will study these questions in particular cases.

\subsection{Dually flat geometry of the Cauchy manifold by conformal flattening}\label{sec:dfs}

The Cauchy distributions are usually handled in information geometry  using the wider scope 
of {\em $q$-Gaussians}~\cite{Naudts-2011,matsuzoe2014hessian,IG-2016} (deformed exponential families~\cite{VG-2013}).
The $q$-Gaussians also include the Student's $t$-distributions. 
Cauchy distributions are $q$-Gaussians for $q=2$.
These $q$-Gaussians are also called $q$-normal distributions~\cite{qnormal-2012}, and they can be obtained as maximum entropy distributions with respect to {\em Tsallis'  entropy} $T_q(\cdot)$~\cite{Tsallis-1988,tsallis2009introduction}  (see Theorem 4.12 of~\cite{IG-2016}):
\begin{equation}
T_q(p)\eqdef  \frac{1}{q-1}\left( 1-\int_{-\infty}^{\infty} p^q(x)\dx \right), \quad q\not=1.
\end{equation}

When $q=2$, we have the following {\em Tsallis' quadratic entropy}:
\begin{equation}
T_2(p)\eqdef 1-\int_{-\infty}^{\infty} p^2(x)\dx.
\end{equation}
We have $\lim_{q\rightarrow 1} T_q(p)=S(p)\eqdef -\int p(x)\log p(x)\dx$, Shannon entropy.

Thus $q$-Gaussians are $q$-exponential families~\cite{naudts2009}, generalizing the MaxEnt exponential families derived from Shannon entropy~\cite{amari2012geometry}.
The integral $E(p)\eqdef\int_{-\infty}^{\infty} p^2(x)\dx$ corresponds to {\em Onicescu's informational energy}~\cite{Onicescu-1966,Onicescu-2020}. Tsallis' entropy is considered in non-extensive statistical physics~\cite{tsallis2009introduction}.

A dually flat structure construction for $q$-Gaussians is reported in~\cite{IG-2016} (Sec. 4.3, p. 84--89).
We instantiate this construction for the Cauchy distributions ($2$-Gaussians):

Let  
\begin{equation}
\exp_{\calC}(u)\eqdef\frac{1}{1-u},\quad u\not=1,
\end{equation}  denote the {\em deformed $q$-exponential} and  
\begin{equation}
\log_{\calC}(u)\eqdef 1-\frac{1}{u},\quad  u\not=0, 
\end{equation}
its compositional inverse, the {\em deformed $q$-logarithm}.

The probability density of a $2$-Gaussian can be factorized as
\begin{equation}
p_\theta(x)=\exp_\calC(\theta^\top x-F(\theta)),
\end{equation}
where $\theta$ denotes the 2D natural parameters.
We have:
\begin{eqnarray}
\log_\calC(p_\theta(x)) &=& 1-\frac{1}{s}\pi(s^2+(x-l)^2) =1-\pi\left(s+ \frac{(x-l)^2}{s}\right) ,\\
 &=:& \theta^\top t(x)-F(\theta),\\
&=& \underbrace{\left(2\pi\frac{l}{s}\right)x+\left(-\frac{\pi}{s}\right)x^2}_{\theta^\top t(x)}-\underbrace{\left(\pi s+\pi\frac{l^2}{s}-1\right)}_{F(\theta)}.
\end{eqnarray}

Therefore the natural parameter is $\theta(l,s)=(\theta_1,\theta_2)=\left(2\pi\frac{l}{s},-\frac{\pi}{s}\right)\in\Theta=\bbR\times\bbR_{-}$ (for $t(x)=(x,x^2)$) and the {\em deformed log-normalizer} is
\begin{eqnarray}
F(\theta(\lambda)) &=&\pi s+\pi\frac{l^2}{s}-1  =: F_\lambda(\lambda),\\
F(\theta) &=& -\frac{\pi^2}{\theta_2}-\frac{\theta_1^2}{4\theta_2}-1.
\end{eqnarray} 

In general, we obtain a strictly convex and $C^3$-function $F_q(\theta)$, called the {\em $q$-free energy} for a $q$-Gaussian family.
Here, we let $F(\theta)\eqdef F_2(q)$ for the Cauchy family: $F(\theta)$ is the {\em Cauchy free energy}.

We  convert back the natural parameter $\theta\in\Theta$ to the ordinary parameter $\lambda\in \bbH$ as follows:
\begin{equation}
\lambda(\theta)=(l,s)=\left(-\frac{\theta_1}{2\theta_2}, -\frac{\pi}{\theta_2}\right).
\end{equation}

The gradient of the deformed log-normalizer is: 
\begin{equation}
\nabla F(\theta)=\left[
\begin{array}{l}
-\frac{\theta_1}{2\theta_2}\\
\frac{\pi^2}{\theta_2^2}+\frac{\theta_1^2}{4\theta_2^2}
\end{array}
\right].
\end{equation}

The gradient $\nabla F(\theta)$ defines the dual global affine coordinate system $\eta\eqdef\nabla F(\theta)$ where 
$\eta\in H=\bbR\times\bbR_{+}$ is the dual parameter space.

It follows the following divergence $D_{\flat}[p_{\lambda_1}:p_{\lambda_2}]$~\cite{IG-2016} between Cauchy densities which is by construction equivalent to a Bregman divergence $B_F(\theta_1:\theta_2)$ (canonical divergence in dually flat space) between their corresponding natural parameters:

\begin{eqnarray}
D_{\flat}[p_{\lambda_1}:p_{\lambda_2}] &\eqdef& 
\frac{1}{\int p_{\lambda_2}^2(x) \dx} \left(\int \frac{p_{\lambda_2}^2(x)}{p_{\lambda_1}(x)} \dx -1\right),\\
&=&  {2\pi s_2} \left(\frac{s_1^2+s_2^2+(l_1-l_2)^2}{2s_1s_2} -1\right),\\
&=&   {2\pi s_2} \frac{(s_1-s_2)^2+(l_1-l_2)^2}{2s_1s_2},\\
&=&  {2\pi s_2}\delta(l_1,s_1,l_2,s_2),\\
&=& B_F(\theta_1:\theta_2),
\end{eqnarray}
where $\theta_1:=\theta(\lambda_1)$ and $\theta_2:=\theta(\lambda_2)$.
We term $B_F(\theta_1:\theta_2)$ the {\em Bregman-Tsallis (quadratic) divergence} ($B_{F_q}$ for general $q$-Gaussians).

We used a computer algebra system (CAS, see Appendix~\ref{sec:maxima}) to calculate  the closed-form formulas of the following definite integrals:
\begin{eqnarray}
\int p_{\lambda_2}^2(x) \dx &=& \frac{1}{2\pi s_2},\\
\int \frac{p_{\lambda_2}^2(x)}{p_{\lambda_1}} \dx &=& \frac{s_1^2+s_2^2+(l_1-l_2)^2}{2s_1s_2}.
\end{eqnarray}

Here, observe that the equivalent Bregman divergence  is {\em not} on swapped parameter order as it is the case for ordinary exponential families:
 $D_\KL[p_{\theta_1}:p_{\theta_2}]=B_F({\theta_2:\theta_1})$ where $F$ denotes the cumulant function of the exponential family, see~\cite{IG-2016,EIG-2018}.

We term the divergence $D_{\flat}$ the {\em flat divergence} because its induced affine connection~\cite{eguchi1992geometry}
 $\leftsup{D_{\flat}}\nabla$ has {\em zero curvature} (i.e., the 4D Riemann-Christofel curvature tensor induced by the connection vanishes, see \cite{IG-2016} p. 134).
 
Since $D_{\flat}[p_{\lambda_1}:p_{\lambda_2}]= {2\pi s_2}\delta(l_1,s_1,l_2,s_2)=\frac{\pi}{s_1} \left({(s_1-s_2)^2+(l_1-l_2)^2}\right)$, the flat divergence is interpreted as a conformal  squared Euclidean distance~\cite{conformaldiv-2015}, with conformal factor $\frac{\pi}{s_1}$.
In general, the Fisher-Rao geometry of $q$-Gaussians has scalar curvature~\cite{qnormal-2012} $\kappa=-\frac{q}{3-q}$.
Thus we recover the scalar curvature $\kappa=-2$ for the Fisher-Rao Cauchy manifold since $q=2$.

\begin{Theorem}\label{thm:BDqgaussian}
The flat divergence $D_{\flat}[p_{\lambda_1}:p_{\lambda_2}]$ between two Cauchy distributions is equivalent to a Bregman divergence 
$B_F(\theta_1:\theta_2)$ on the corresponding natural parameters, and yields the following closed-form formula using the ordinary location-scale
 parameterization:
\begin{equation}
D_{\flat}[p_{\lambda_1}:p_{\lambda_2}] = {2\pi s_2}\delta(l_1,s_1,l_2,s_2)
=  \frac{\pi}{s_1} \left( (s_1-s_2)^2+(l_1-l_2)^2\right)
=\frac{\pi}{s_1}\|\lambda_1-\lambda_2\|^2_2.
\end{equation}
\end{Theorem}
 
The conversion of $\eta$-coordinates  to $\theta$-coordinates are calculated as follows:
\begin{equation}
\theta(\eta)=\left[
\begin{array}{l}
\frac{2\pi\eta_1}{\sqrt{\eta_2-\eta_1^2}}\\
\frac{-\pi-}{\sqrt{\eta_2-\eta_1^2}}
\end{array}
\right] :=\nabla F^*(\eta),
\end{equation}
where 
\begin{equation}
F^*(\eta)\eqdef\theta(\eta)^\top \eta-F(\theta(\eta)),
\end{equation}
is the {\em Legendre-Fenchel convex conjugate}~\cite{IG-2016}:
\begin{equation}
F^*(\eta) = 1-2\pi\sqrt{\eta_2-\eta_1^2}.
\end{equation}

Since 
\begin{equation}
\eta(\lambda)=\eta(\theta(\lambda))=(\lambda_1,\lambda_1^2+\lambda_2^2)=(l,l^2+s^2),
\end{equation}
we have
\begin{equation}
F_\lambda^*(\lambda)\eqdef F^*(\eta(\lambda))=1-2\pi\sqrt{l^2+s^2-l^2}=1-2\pi s
\end{equation}
that is {\em independent} of the location parameter $l$.
Moreover, we have~\cite{IG-2016}
\begin{equation}
F_\lambda^*(\lambda) \eqdef 1-\frac{1}{\int p^2(x)dx}=1-\frac{1}{\frac{1}{2\pi s}}=1-2\pi s.
\end{equation}
We can convert the dual parameter $\eta$ to the ordinary parameter $\lambda\in\bbH$ as follows:
\begin{equation}
\lambda(\eta)=(l,s)=(\eta_1,\sqrt{\eta_2-\eta_1^2}).
\end{equation}

It follows that we have the following equivalent expressions for the flat divergence:
\begin{equation}
D_{\flat}[p_{\lambda_1}:p_{\lambda_2}]=B_F(\theta_1:\theta_2)=B_{F^*}(\eta_2:\eta_1)=A_F(\theta_1:\eta_2)=A_{F^*}(\eta_2:\theta_1),
\end{equation}
where 
\begin{equation}
A_F(\theta_1:\eta_2)\eqdef F(\theta_1)+F^*(\eta_2)-\theta_1^\top\eta_2,
\end{equation}
 is the {\em Legendre-Fenchel divergence} measuring the {\em inequality gap} of the Fenchel-Young inequality:
\begin{equation}
F(\theta_1)+F^*(\eta_2)\geq \theta_1^\top\eta_2.
\end{equation}
That is, $A_F(\theta_1:\eta_2)=\mathrm{rhs}(\theta_1:\eta_2)-\mathrm{lhs}(\theta_1:\eta_2)\geq 0$, 
where $\mathrm{rhs}(\theta_1:\eta_2)\eqdef F(\theta_1)+F^*(\eta_2) $ and $\mathrm{lhs}(\theta_1:\eta_2)=\theta_1^\top\eta_2$.

The Hessian metrics of the dual convex potential functions $F(\theta)$ and $F^*(\eta)$ are:
\begin{eqnarray}
\nabla^2 F(\theta) &=& \mattwotwo{-\frac{1}{2\theta_2}}{\frac{\theta_1}{2\theta_2^2}}{\frac{\theta_1}{2\theta_2^2}}{-\frac{\theta_1^2}{2\theta_2^2}-\frac{2\pi^2}{\theta_2^2}}  =: g_F(\theta) ,\\
\nabla^2 F^*(\eta) &=& \mattwotwo{\frac{2}{\sqrt{\eta_2-\eta_1^2}} + \frac{2\eta_1^2}{(\eta_2-\eta_1^2)^{\frac{3}{2}}}}{-\frac{\eta_1}{(\eta_2-\eta_1^2)^{\frac{3}{2}}}}{-\frac{\eta_1}{(\eta_2-\eta_1^2)^{\frac{3}{2}}}}{\frac{1}{2}{(\eta_2-\eta_1^2)^{\frac{3}{2}}}} =: g_F^*(\eta).\\
\end{eqnarray}

We check the Crouzeix identity~\cite{crouzeix1977relationship,EIG-2018}:
\begin{equation}
\nabla^2 F(\theta)\nabla^2 F^*(\eta(\theta))=\nabla^2 F(\theta(\eta))\nabla^2 F^*(\eta)=I,
\end{equation}
where $I$ denotes the $2\times 2$ identity matrix. 

The Hessian metric $\nabla^2 F(\theta)$ is also called the {\em $q$-Fisher metric}~\cite{qnormal-2012} (for $q=2$).
Let $g_{\FR}^\lambda(\lambda)$ and $g_{\FR}^\theta(\theta)$ denote the Fisher information metric expressed using the $\lambda$-coordinates and the $\theta$-coordinates, respectively.
Then, we have
\begin{equation}
g_\FR^\theta(\theta) = \Jac^\top_{\lambda}(\theta)\times g_\FR^\lambda(\lambda(\theta))\times \Jac_{\lambda}(\theta),
\end{equation}
where $\Jac_{\lambda}(\theta)$ denotes the Jacobian matrix:
\begin{equation}
\Jac_{\lambda}(\theta)\eqdef \left[ \frac{\partial\lambda_i}{\partial\theta_j} \right].
\end{equation}
Similarly, we can express the Hessian metric $g_F\eqdef \nabla^2 F(\theta)$ using the $\lambda$-coordinate system:
\begin{equation}
g_F^\lambda(\lambda) = \Jac^\top_{\theta}(\lambda)\times g_F^\theta(\theta(\lambda))\times \Jac_{\theta}(\lambda).
\end{equation}

We calculate explicitly the following Jacobian matrices:
\begin{equation}
\Jac_{\theta}(\lambda)=\pi \mattwotwo{\frac{2}{\lambda_2}}{-2\frac{\lambda_1}{\lambda_2^2}}{0}{\frac{1}{\lambda_2^2}}.
\end{equation}
and
\begin{equation}
\Jac_{\lambda}(\theta)=  \mattwotwo{-\frac{1}{2\theta_2}}{\frac{\theta_1}{2\theta_2^2}}{0}{\frac{\pi}{\theta_2^2}}.
\end{equation}

We check that we have
\begin{eqnarray}
g_F^\theta(\theta) &=& -\frac{2\theta_2}{\pi^2} g_\FR^\theta(\theta),\\
g_F^\lambda(\lambda) &=&\frac{2}{\pi\sigma} g_\FR^\lambda(\lambda).
\end{eqnarray}

That is, the Riemannian metric tensors $g_\FR^\lambda(\lambda)$ and $g_F^\lambda(\lambda)$ (or $g_F^\theta(\theta)$ and $g_\FR^\theta(\theta)$) are conformally equivalent.
This is, there exists a smooth function $u(\lambda)=\log \frac{2}{\pi\sigma}$ such that $g_F^\lambda(\lambda) =e^{u(\lambda)}  g_\FR^\lambda(\lambda)$.

This dually flat space construction of the Cauchy manifold
\begin{equation*}
\left(\calC,g(\theta)=\nabla^2 F(\theta),\supleft{D_\flat}\nabla,\supleft{D_\flat}\nabla^*=\supleft{D_\flat^*}\nabla\right)
\end{equation*}
 can be interpreted as a conformal flattening of the curved $\alpha$-geometry~\cite{qnormal-2012,IG-2016,ohara2019conformal}.
The relationships between the curvature tensors of dual $\pm\alpha$-connections are  studied in~\cite{zhang2007note}.

Notice that this dually flat geometry can be recovered from the divergence-based structure of \S\ref{sec:divgeo} by considering the Bregman-Tsallis divergence.
Figure~\ref{fig:CauchyManifold} illustrates the relationships between the  invariant $\alpha$-geometry and the dually flat geometry of the Cauchy manifold.
The $q$-Gaussians can further be generalized by {\em $\chi$-family} with corresponding deformed logarithm and exponential functions~\cite{IG-2016,amari2012geometry}. The {\em $\chi$-family} unifies both the dually flat exponential family with the dually flat mixture family~\cite{amari2012geometry}. 

A statistical dissimilarity $D[p_{\lambda_1}:p_{\lambda_2}]$ between two parametric distributions
 $p_{\lambda_1}$ and $p_{\lambda_2}$ amounts to an equivalent dissimilarity $D(\theta_1:\theta_2)$ between their parameters:
 $D(\theta_1:\theta_2) \eqdef D[p_{\lambda_1}:p_{\lambda_2}]$.
When the parametric dissimilarity is smooth, one can construct the divergence-based $\alpha$-geometry~\cite{amari2010information,EIG-2018}.
Thus the dually flat space structure of the Cauchy manifold can also be obtained from the {\em divergence-based $\pm\alpha$-geometry} obtained from the flat divergence $D_\flat$ (see Figure~\ref{fig:CauchyManifold}). 
It can be shown that the  dually flat space $q$-geometry is the unique geometry in the intersection of the conformal Fisher-Rao geometry with the deformed $\chi$-geometry (Theorem~13 of ~\cite{amari2012geometry}) when the manifold is the positive orthant $\bbR^{d+1}$.
Note that a dually flat space in information geometry is usually not Riemannian flat (with respect to the Levi-Civita connection, e.g., the Gaussian manifold). In particular, Matsuzoe proved in~\cite{Matsuzoe-2014} that the Riemannian manifold $(\calC,\nabla^2 F(\theta))$ induced by the $q$-Fisher metric is of constant curvature $-1$ when $q=2$.

There are many alternative possible ways to build a dually flat space from a $q$-Gaussian family once a convex Bregman generator $F(\theta)$ has been built from the density $p_q(\theta)$ of a $q$-Gaussian. The method presented above is a natural generalization of the  dually flat space construction  for exponential families.
To give another approach, let us mention that Matsuzoe~\cite{Matsuzoe-2014} also introduced another Hessian metric $g^M(\theta)=[g_{ij}^M(\theta)]$ defined by:
\begin{equation}
g_{ij}^M(\theta)\eqdef \int \partial_i p_\theta(x)\partial_j\log_q p_\theta(x) \dx.
\end{equation} 
This metric is conformal to both the Fisher metric and the $q$-Fisher metric, and is obtained by generalizing  equivalent representations of the Fisher information matrix (see $\alpha$-representations in~\cite{IG-2016}).

\begin{figure}
\centering
 
\includegraphics[width=\columnwidth]{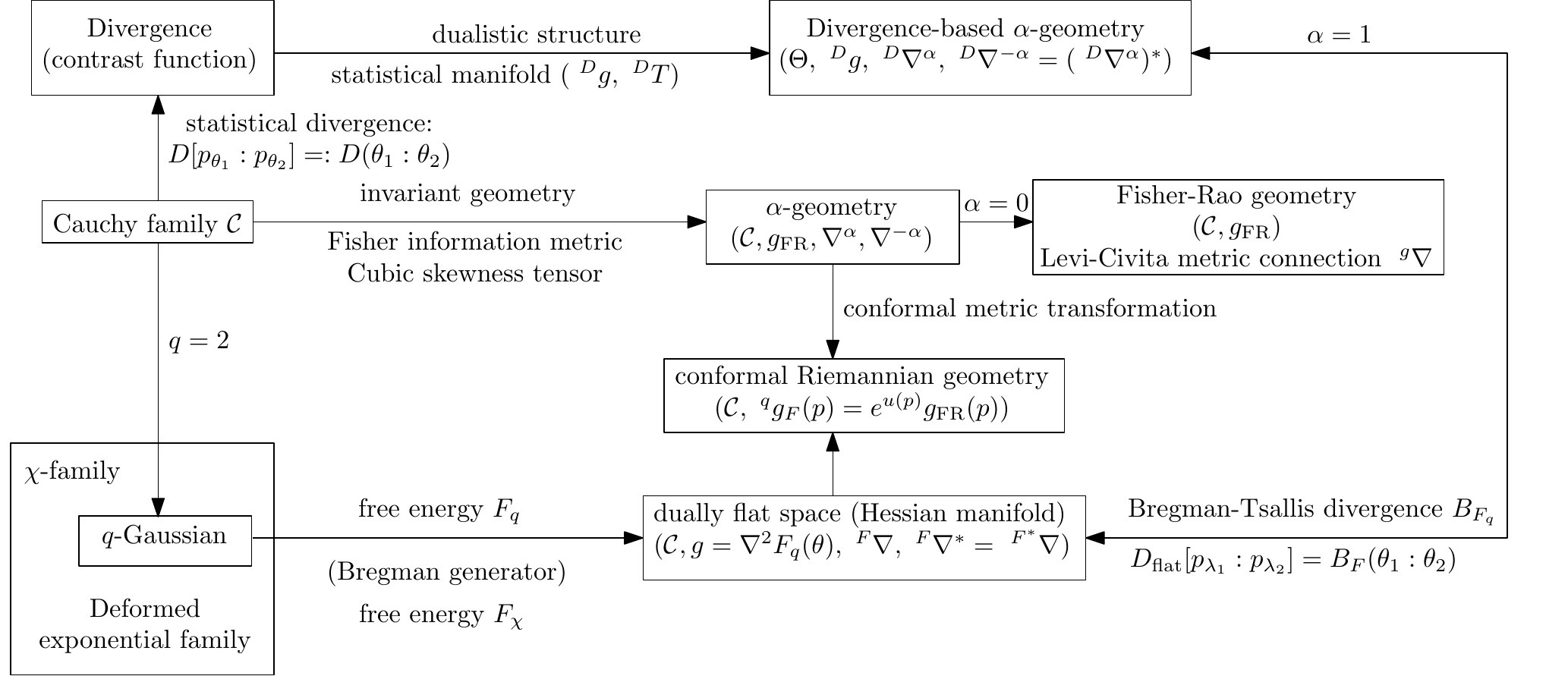}

\caption{Information-geometric structures on the Cauchy manifold and their relationships.} \label{fig:CauchyManifold}
\end{figure}

\section{Invariant divergences: $f$-divergences and $\alpha$-divergences }\label{sec:alphadiv}

\subsection{Invariant divergences in information geometry}

The {\em $f$-divergences}~\cite{csiszar1967information,NN-fdiv-2013} between two densities $p(x)$ and $q(x)$ is defined for a positive convex function $f$, strictly convex at $1$, with $f(1)=0$ as: 
\begin{equation}\label{eq:fdiv}
I_f[p:q]\eqdef\int_{\calX} p(x) f\left(\frac{q(x)}{p(x)}\right)\dx,
\end{equation}
The KL divergence is a $f$-divergence obtained for the generator $f(u)=-\log(u)$.

An {\em invariant divergence} is a divergence $D$ is a divergence which satisfies the {\em information monotonicity}~\cite{IG-2016}:
$D[p_X:p_Y]\geq D[p_{t(X)}:p_{t(Y)}]$ with equality iff $t(X)$ is a sufficient statistic.
The invariant divergences are the $f$-divergences for the simplex sample space~\cite{IG-2016}.
Moreover, the {\em standard $f$-divergences} (calibrated with $f(1)=0$ and $f'(1)=f''(1)=1$) induce the Fisher information metric (FIm) for its metric tensor $\supleft{I_f}g$ when the sample space is the probability simplex: $\supleft{I_f}g=g_\FR$, see~\cite{IG-2016}. 

\subsection{$\alpha$-Divergences between location-scale densities}
Let $I_\alpha[p:q]$ denote the {\em $\alpha$-divergence}~\cite{IG-2016} between $p$ and $q$:
\begin{equation}
I_\alpha[p:q] \eqdef \frac{1}{\alpha(1-\alpha)}(1-C_\alpha[p:q]),\quad \alpha\not\in\{0,1\}
\end{equation}
where $C_\alpha[p:q]$ is {\em Chernoff $\alpha$-coefficient}~\cite{Chernoff-1952,ChernoffNielsen-2013}:
\begin{eqnarray}
C_\alpha[p:q] &\eqdef& \int p^{\alpha}(x)q^{1-\alpha}(x)\dx,\\
 &=& \int q(x)\left(\frac{p(x)}{q(x)}\right)^{\alpha},\\
 &=& C_{1-\alpha}[q:p].
\end{eqnarray}
We have $I_\alpha[p:q]=I_{-\alpha}[q:p]={I_\alpha}^*[p:q]$.

The $\alpha$-divergences include the {\em chi square divergence} ($\alpha=2$), the {\em squared Hellinger divergence} ($\alpha=0$, symmetric) and in the limit cases the Kullback-Leibler (KL) divergence ($\alpha\rightarrow 1$) and the reverse KL divergence ($\alpha\rightarrow 0$).
The $\alpha$-divergences are $f$-divergences for the generator:
\begin{equation}
f_\alpha(u)=
\left\{\begin{array}{ll}
\frac{u^{1-\alpha}-u}{\alpha(\alpha-1)}, & \text { if } \alpha \neq 0, \alpha \neq 1 \\
u \log(u), & \text { if } \alpha=0 \quad\text{(reverse Kullback-Leibler divergence),} \\
-\log(u), & \text { if } \alpha=1 \quad\text{(Kullback-Leibler divergence).}
\end{array}\right.
\end{equation}

For location scale families, let
\begin{equation}
C_\alpha(l_1,s_1;l_2,s_2)\eqdef C_\alpha\left[p_{l_1,s_1}:p_{l_2,s_2}\right].
\end{equation}

Using change of variables in the integrals, one can show the following identities:
\begin{eqnarray}
C_\alpha(l_1,s_1;l_2,s_2) &=& C_\alpha\left(0,1;\frac{l_2-l_1}{s_1},\frac{s_2}{s_1}\right),\\
&=& C_\alpha\left( \frac{l_1-l_2}{s_2},\frac{s_1}{s_2};0,1\right),\\
&=& C_{1-\alpha}\left(0,1 ; \frac{l_1-l_2}{s_2},\frac{s_1}{s_2}\right),\\
&=& C_{1-\alpha}(l_2,s_2;l_1,s_1). 
\end{eqnarray}

For the location-scale families which include the normal family $\calN$, the Cauchy family $\calC$ and the $t$-Student families $\calS_k$ with fixed degree of freedom $k$,
  the $\alpha$-divergences are {\em not} symmetric in general (e.g., $\alpha$-divergences between two normal distributions).
However, we have shown that the chi square divergences and the KL divergence are symmetric when densities belong
to the Cauchy family.
Thus it is of interest to prove whether the $\alpha$-divergences between Cauchy densities are symmetric or not, and report their closed-form formula for all $\alpha\in\bbR$. 

Using symbolic integration described in Appendix~\ref{sec:maxima}, we found that
{\small 
\begin{equation}
C_3(p_{\lambda_1};p_{\lambda_2}) = \frac{3s_2^4+(2s_1^2+6l_2^2-12l_1l_2+6l_1^2)s_2^2+3s_1^4+(6l_2^2-12l_1l_2+6l_1^2)s_1^2+3l_2^4-12
l_1l_2^3+18l_1^2l_2^2-12l_1^3l_2+3l_1^4)}{8s_1^2s_2^2},
\end{equation}
}
and checked that this Chernoff similarity coefficient is symmetric:

\begin{equation}
C_3(p_{\lambda_1};p_{\lambda_2}) = C_3(p_{\lambda_2};p_{\lambda_1}).
\end{equation}

Therefore the $3$-divergence $I_3$ between two Cauchy distributions is symmetric.
In particular, when $l_1=l_2=l$, we find that
\begin{eqnarray}
C_3(p_{l,s_1};p_{l,s_2}) &=& \frac{3(s_1^4+s_2^4)+2s_1^2 s_2^2}{8s_1^2s_2^2},\\
 &=& 1+ \frac{3}{4} \frac{(s_1^2-s_2^2)^2}{2s_1^2s_2^2},\\
&=& 1+ \frac{3}{4} \delta(l^2,s_1^2,l_2^2,s_2^2).
\end{eqnarray}
 
In the Appendix, we proved by symbolic calculations that the $\alpha$-divergences are symmetric for $\alpha\in\{0,1,2,3,4\}$.

\begin{Remark}
The Cauchy family can also be interpreted as a family of univariate elliptical distributions~\cite{Mitchell-1988}.
A {\em univariate elliptical distribution} has canonical parametric density:
\begin{equation}
q_{\mu,\sigma}(x) \eqdef \frac{1}{\sigma} h\left( \left( \frac{x-\mu}{\sigma}\right)^2\right),
\end{equation}
for some function $h(u)$.
For example, the Gaussian distributions are elliptical distributions obtained for $h(u)=\frac{1}{\sqrt{2\pi}}\exp\left(-\frac{1}{2}u\right)$.
Location-scale densities $p_{l,s}$ with standard density $p_{0,1}$ can be interpreted as univariate elliptical distributions $q_{\mu,\sigma}$
with $h(u)=p_{0,1}(u^2)$ and $(\mu,\sigma)=(l,s)$: $p_{l,s}=q_{\mu,\sigma}$.
It follows that the Cauchy densities are elliptical distributions for $h(u)=\frac{1}{\pi(1+u)}$.
By doing a change of variable in the KL divergence integral, we find again the following identity:
\begin{equation}
D_\KL\left[q_{\mu_1,\sigma_1} : q_{\mu_2,\sigma_2} \right] =  D_\KL\left[q_{0,1} : q_{\frac{\mu_2-\mu_1}{\sigma_2},\frac{\sigma_1}{\sigma_2}} \right]. 
\end{equation}
\end{Remark}

\subsection{Metrization of the Kullback-Leibler divergence}\label{sec:KLmetrization}

The {\em Kullback-Leibler divergence}~\cite{CT-2012} $D_\KL[p:q]$ between two continuous probability densities $p$ and $q$ defined over the real line support is an oriented dissimilarity measure defined by:
\begin{equation}
D_\KL[p:q] \eqdef  \int_{-\infty}^\infty p(x)\log\left(\frac{p(x)}{q(x)}\right)\dx. \label{eq:KLDiv}
\end{equation}

The closed-form formula for the KL divergence between two Cauchy distributions requires to perform a (non-trivial) integration task.
The following closed-form expression  has been reported in~\cite{KLCauchy-2019} using advanced symbolic integration:
\begin{equation}\label{eq:KLCauchycft}
D_\KL[p_{l_1,s_1}:p_{l_2,s_2}] =  \log\left(1+\frac{(s_1-s_2)^2+(l_1-l_2)^2}{4s_1s_2}\right).  
\end{equation}

Although the KL divergence is usually asymmetric, it is a remarkable fact that it is symmetric between any two Cauchy densities. 
However, the KL divergence of Eq.~\ref{eq:KLDiv} and Eq.~\ref{eq:KLCauchycft} does {\em not} satisfy the triangle inequality,
and therefore although symmetric, it is not a metric distance.

The KL divergence between two Cauchy distributions is related to the Pearson $D_{\chi^2_P}[p:q]$ and Neyman $D_{\chi^2_N}[p:q] $ chi square divergences~\cite{NN-fdiv-2013}:
\begin{eqnarray}
D_{\chi^2_P}[p:q] &\eqdef& \int \frac{(q(x)-p(x))^2}{p(x)} \dx,\\
D_{\chi^2_N}[p:q] &\eqdef& \int \frac{(q(x)-p(x))^2}{q(x)} \dx =  D_{\chi^2_P}^*[p:q] = D_{\chi^2_P}[q:p].
\end{eqnarray}
Indeed, the formula for the Pearson and Neyman chi square divergences between two Cauchy distributions coincide, and (surprisingly) amount to the $\delta$ distance:
\begin{eqnarray}
D_{\chi^2_P}[p_{l_1,s_1}:p_{l_2,s_2}] &=& D_{\chi^2_N}[p_{l_1,s_1}:p_{l_2,s_2}],\\
&=& \frac{(s_1-s_2)^2+(l_2-l_1)^2}{2s_1s_2},\\
&=:& \delta(l_1,s_1;l_2,s_2).
\end{eqnarray}
Since the Pearson and Neyman chi square divergences are symmetric, let us write $D_{\chi^2}[p:q]=D_{\chi^2_P}[p:q]$ in the remainder.
We can rewrite the Fisher-Rao distance between two Cauchy distributions using the $D_{\chi^2}$ divergence as follows:
\begin{equation}
\rho_\FR[p_{l_1,s_1},p_{l_2,s_2}]=\frac{1}{\sqrt{2}} \arccosh\left(1+D_{\chi^2}[p_{l_1,s_1}:p_{l_2,s_2}]\right).
\end{equation}

Figure~\ref{fig:convertChiFR} plots the strictly increasing chi-to-Fisher-Rao conversion function:
 \begin{equation}
t_{\chi\rightarrow\FR}(u)\eqdef\frac{1}{\sqrt{2}} \arccosh\left(1+u\right).
\end{equation}

\begin{figure}
\centering
\includegraphics[width=0.65\columnwidth]{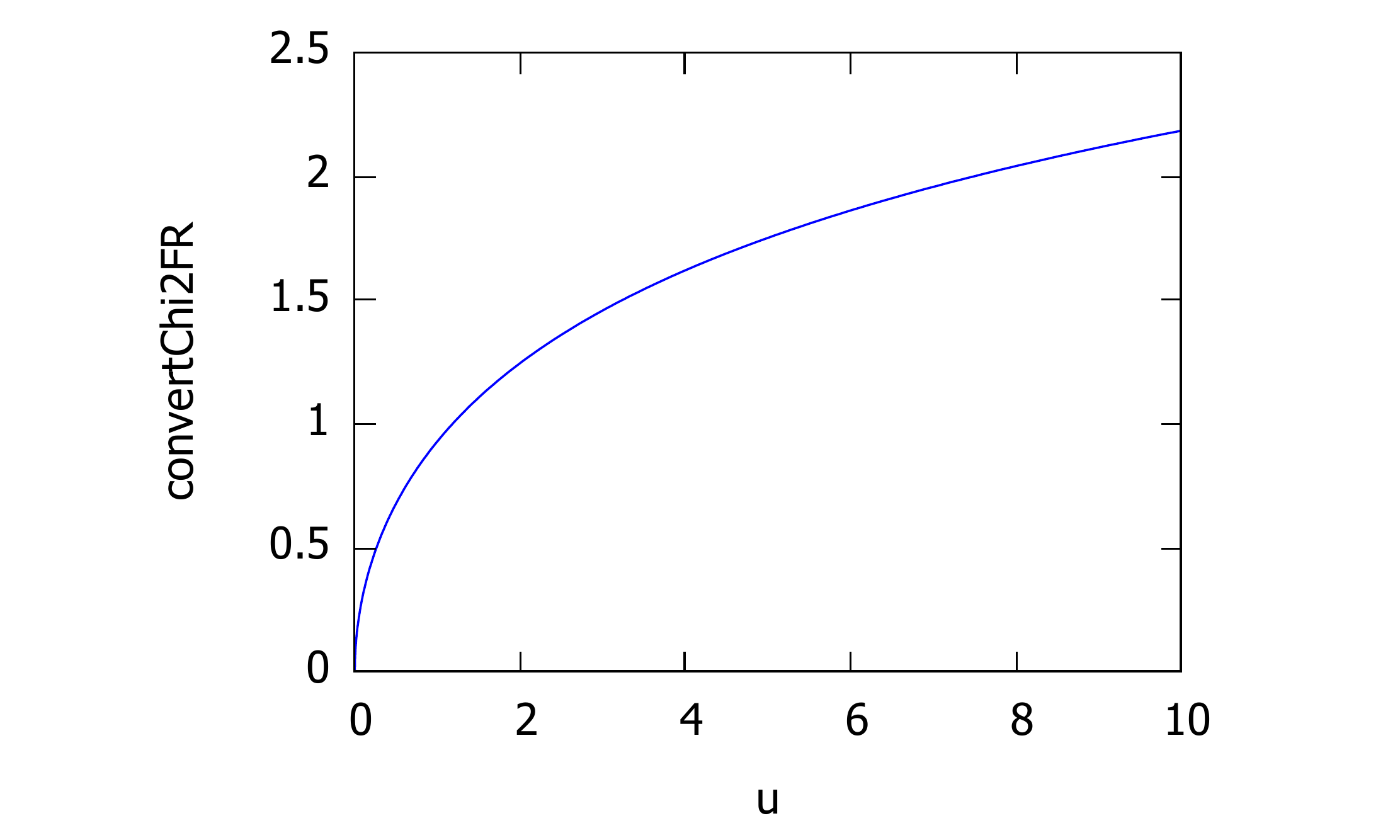}
\caption{Plot of the chi-to-Fisher-Rao conversion function: A strictly increasing function.} \label{fig:convertChiFR}
\end{figure}

Since the Cauchy family is a location-scale family, we have the following general invariance property of $f$-divergences:

\begin{Theorem}
The $f$-divergence~\cite{csiszar1967information} between two location-scale densities $p_{l_1,s_1}$ and $p_{l_2,s_2}$ can be reduced to the calculation of the $f$-divergence between one standard density with another location-scale density:
\begin{equation}
I_f[p_{l_1,s_1}:p_{l_2,s_2}] = I_f\left[p:p_{\frac{l_2-l_1}{s_1},\frac{s_2}{s_1}}\right] 
= I_f\left[p_{\frac{l_1-l_2}{s_2},\frac{s_1}{s_2}}:p\right].
\end{equation}
\end{Theorem}

\begin{proof}
The proof follows from changes of the variable $x$ in the definite integral of Eq~\ref{eq:fdiv}:
Consider $y=\frac{x-l_1}{s_1}$ with $\dx=s_1\dy$, $x=s_1y+l_1$ and $\frac{x-l_2}{s_2}=\frac{s_1y+l_1-l_2}{s_2}=\frac{y-\frac{l_2-l_1}{s_1}}{\frac{s_2}{s_1}}$. We have

\begin{eqnarray}
I_f[p_{l_1,s_1}:p_{l_2,s_2}] &:=& \int_{\calX} p_{l_1,s_1}(x) f\left(\frac{p_{l_2,s_2}(x)}{p_{l_1,s_1}(x)}  \right)\dx,\\
&=& \int_{\calY} \frac{1}{s_1} p(y) f\left( \frac{\frac{1}{s_2} p\left(\frac{y-\frac{l_2-l_1}{s_1}}{\frac{s_2}{s_1}}\right)}{\frac{1}{s_1}p(y)}
\right) s_1\dy,\\
&=& \int p(y) f\left(\frac{p_{\frac{l_2-l_1}{s_1},\frac{s_2}{s_1}}(y)}{p(y)}\right) \dy,\\
&=&  I_f\left[p:p_{\frac{l_2-l_1}{s_1},\frac{s_2}{s_1}}\right].
\end{eqnarray}
The proof for  $I_f[p_{l_1,s_1}:p_{l_2,s_2}]=I_f(p_{\frac{l_1-l_2}{s_2},\frac{s_1}{s_2}}:p)$ is similar.
One can also use the {\em conjugate generator} 
$f^*(u)\eqdef uf(\frac{1}{u})$ which yields the {\em reverse $f$-divergence}: $I_{f^*}[p:q]=I_f[q:p]={I_f}^*[p:q]$.
\end{proof}

Since the KL divergence is expressed by
$D_\KL[p_{l_1,s_1}:p_{l_2,s_2}] =  \log\left(1+\frac{1}{2}\delta(l_1,s_1;l_2,s_2)\right)$, we also check that
\begin{eqnarray}
\delta(l_1,s_1;l_2,s_2) &=& \delta\left(0,1;\frac{l_1-l_2}{s_2},\frac{s_1}{s_2}\right),\\
 &=& \delta\left(\frac{l_2-l_1}{s_1},\frac{s_2}{s_1};0,1\right),\\
 &=:& \delta(a,b), 
\end{eqnarray} 
where
\begin{equation}
\delta(a,b)\eqdef \frac{a^2+(b-1)^2}{4b}.
\end{equation}

It follows the following corollary for scale families:
\begin{Corollary}
The $f$-divergences between scale densities is scale-invariant and amount to a scalar scale-invariant divergence $D_f(s_1:s_2) := I_f[p_{s_1}:p_{s_2}]$.
\end{Corollary}

\begin{proof}
\begin{eqnarray}
D_f(s_1:s_2) := I_f[p_{s_1}:p_{s_2}] &=&   I_f\left(p:p_{\frac{s_2}{s_1}}\right) =: D_f\left(1:\frac{s_2}{s_1}\right),\\
 &=& I_f\left[p_{\frac{s_1}{s_2}}:q\right]=: 
D_f\left(\frac{s_1}{s_2}:1\right).
\end{eqnarray}
\end{proof}

Many algorithms and data-structures can be designed efficiently when dealing with metric distances:
For example, the   metric ball tree~\cite{Uhlmann-1991} or the vantage point tree~\cite{yianilos1993data,bvp-2009} are two such data structures for querying efficiently nearest neighbors in metric spaces.
Thus it is of interest to consider {\em statistical dissimilarities} which are metric distances.
The total variation distance~\cite{CT-2012} and the square-root of the Jensen-Shannon divergence~\cite{JS-2004} are two common examples of {\em statistical metric distances} often met in the literature.
In general, the metrization of $f$-divergences was investigated in~\cite{kafka1991powers,Vajda-MetricDivergence-2009}.

We shall prove the following theorem:

\begin{Theorem}\label{thm:sqrtKL}
The square root of the Kullback-Leibler divergence  
between two Cauchy density $p_{l_1,s_1}$ and $p_{l_2,s_2}$  is a metric distance:
\begin{equation}
\rho_\KL[p_{l_1,s_1},p_{l_2,s_2}]:=\sqrt{D_\KL[p_{l_1,s_1}:p_{l_2,s_2}]} = \sqrt{\log\left(1+\frac{(s_1-s_2)^2+(l_1-l_2)^2}{4s_1s_2}\right)}. \label{eq:sqrtKLCauchy}
\end{equation}
\end{Theorem}

\begin{proof}
The proof consists in showing that the square root of the conversion function of the Fisher-Rao distance to the KL divergence is a metric transform~\cite{Dissimilarity-2005}.
A {\em metric transform} $t(u):\bbR_+\rightarrow \bbR_+$ is a transform which preserves the metric distance $\rho$, i.e., 
$(t\circ\rho)(p,q) = t(\rho(p,q))$ is a metric distance.
The following are sufficient conditions for function $t(u)$ to be a metric transform: 
\begin{enumerate}
\item $t$ is a strictly increasing function,  
\item $t(0)=0$, 
\item  $t$ satisfies that {\em subadditive property}: $t(a+b)\leq t(a)+t(b)$ for all $a,b\geq 0$. 
\end{enumerate}
For example, strictly concave functions $t(u)$ with $t(0)=0$ are metric transforms.
In general, one can check that $t(u)$ is subadditive by verifying that the ratio of functions $\frac{t(u)}{u}$ is non-decreasing.

The following transform $\sqrt{t_{\FR\rightarrow\KL}(u)}$ converts the Fisher-Rao distance $\rho_\FR$ to the Kullback-Leibler divergence $D_\KL$:
\begin{equation}
t_{\FR\rightarrow\KL}(u):=\log\left(\frac{1}{2}+\frac{1}{2}\mathrm{cosh}(\sqrt{2}u)\right),
\end{equation}
where
\begin{equation}
\mathrm{cosh}(x)\eqdef\frac{e^x+e^{-x}}{2}.
\end{equation}

The square root of that conversion function is a subadditive function  since $\frac{\sqrt{t_{\FR\rightarrow\KL}(u)}}{u}$ is non-decreasing (see Figure~\ref{fig:metrictransform})
 and $\sqrt{t_{\FR\rightarrow\KL}(0)}=0$. 

\begin{figure}
\centering
\includegraphics[width=0.65\columnwidth]{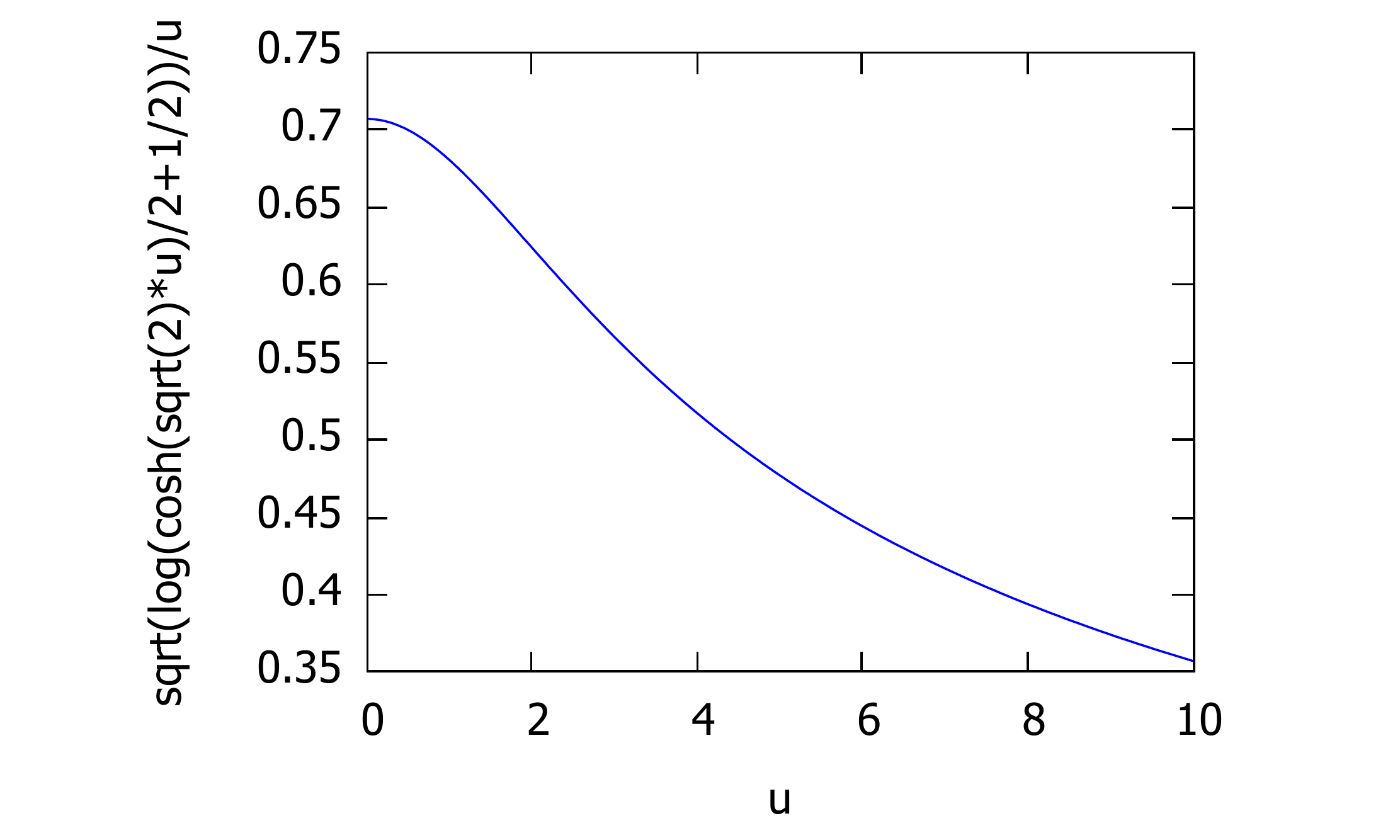}
\caption{Plot of the function  $\frac{\sqrt{t_{\FR\rightarrow\KL}(u)}}{u}$.} \label{fig:metrictransform}
\end{figure}

Since the Fisher-Rao distance is a metric distance and since $\sqrt{t_{\FR\rightarrow\KL}(u)}$ is a metric transform, we conclude that
\begin{equation}
\rho_\KL[p_{l_1,s_1}:p_{l_2,s_2}]:=\sqrt{D_\KL[p_{l_1,s_1}:p_{l_2,s_2}]}=\sqrt{t_{\FR\rightarrow\KL}(\rho_\FR[p_{l_1,s_1}:p_{l_2,s_2}])}
\end{equation} 
is a metric distance.
\end{proof}

A metric distance $\rho(p,q)$ is said {\em Hilbertian} if there exists an {\em embedding} $\phi(\cdot)$ into a Hilbert space such that 
$\rho(p,q)=\|\phi(p)-\phi(q)\|_H$, where $\|\cdot\|_H$ is a norm. A metric is said Euclidean if there exists an {\em embedding} with associated norm $\ell_2$, the Euclidean norm.
For example, the square root of the celebrated Jensen-Shannon divergence is a Hilbertian distance~\cite{JS-2004}.

Let us prove the following:

\begin{Theorem}\label{thm:sqrtKLscale}
The square root of the KL divergence between to Cauchy densities of the same scale family is a Hilbertian distance.
\end{Theorem}

\begin{proof}
For Cauchy distributions with  fixed location parameter $l$, the KL divergence of  Eq.~\ref{eq:KLCauchycft} simplifies to:
\begin{equation}
D_\KL[p_{l,s_1}:p_{l,s_2}] = \log\left(\frac{(s_1+s_2)^2}{4s_1s_2}\right). 
\end{equation}
We can rewrite this KL divergence as 
\begin{equation}
D_\KL[p_{l,s_1}:p_{l,s_2}]=2\log\left(\frac{A(s_1,s_2)}{G(s_1,s_2)}\right),
\end{equation}
where $A(s_1,s_2)=\frac{s_1+s_2}{2}$ and $G(s_1,s_2)=\sqrt{s_1s_2}$ are the {\em arithmetic mean} and the {\em geometric mean} of $s_1$ and $s_2$, respectively.
Then we use Lemma~3 of~\cite{BregmanTriangleIneq-2013} to conclude that  $\sqrt{D_\KL[p_{l,s_1}:p_{l,s_2}]}$ is a Hilbertian metric distance.

Another proof consists in rewriting the KL divergence as a scaled {\em Jensen-Bregman divergence}~\cite{BR-2011,BregmanTriangleIneq-2013}:
\begin{equation}
D_\KL[p_{l,s_1}:p_{l,s_2}]=2\ \JB_{F}(s_1,s_2),
\end{equation}
where 
\begin{equation}
\JB_{F}(\theta_1,\theta_2) \eqdef \frac{F(\theta_1)+F(\theta_2)}{2}-F\left(\frac{\theta_1+\theta_2}{2}\right),
\end{equation}
for a strictly convex generator $F$.
We use $F(\theta)=-\log(u)$, i.e., the {\em Burg information} yielding the {\em Jensen-Burg divergence} $\JB_{F}$.
Then we use Corollary~1 of~\cite{BregmanTriangleIneq-2013} (i.e., $F$ is the cumulant
of an infinitely divisible distribution) to conclude that $\sqrt{\JB_{F}(\theta_1,\theta_2)}$ is a metric distance
 (and hence, $\rho_\KL(l,s_1,l,s_2)=\sqrt{D_\KL[p_{l,s_1}:p_{l,s_2}]}=\sqrt{2} \sqrt{\JB_{F}(s_1,s_2)}$ is a Hilbertian metric distance).
\end{proof}

The {\em $\alpha$-skewed Jensen-Bregman divergence} is defined by
\begin{equation}
\JB_{F}^\alpha(\theta_1:\theta_2) \eqdef \alpha F(\theta_1)+(1-\alpha)F(\theta_2)-F\left(\alpha\theta_1+(1-\alpha)\theta_2\right),
\end{equation}
and the maximal $\alpha$-skewed Jensen-Bregman divergence is called the {\em Jensen-Chernoff divergence}:
\begin{equation}
\JB_{F}^{\alpha^*}(\theta_1:\theta_2) \eqdef \max_{\alpha\in(0,1)} \JB_{F}^\alpha(\theta_1:\theta_2).
\end{equation}
The maximal exponent $\alpha^*$ corresponds to the {\em error exponent} in Bayesian hypothesis testing on exponential family manifolds~\cite{ChernoffNielsen-2013}. 
In general, the metrization of Jensen-Bregman divergence (and Jensen-Chernoff) was studied in~\cite{Chen-BD2-2008}.

Furthermore, by combining Corollary~1~of~\cite{BregmanTriangleIneq-2013} with Theorem~3~of~\cite{BR-2011}, we get the following proposition:

\begin{Proposition}\label{prop:sqrtBhat}
The square root of the Bhattacharyya divergence between two densities of an exponential family is a metric distance when the exponential family is infinitely divisible.
\end{Proposition}

This proposition holds because the Bhattacharyya divergence
\begin{equation}
D_\Bhat[p,q]=-\log\left(\int \sqrt{p(x)q(x)}\dx\right),
\end{equation}
 between two  parametric densities $p(x)=p_{\theta_1}(x)$ and $q(x)=p_{\theta_2}(x)$ of an exponential family with cumulant function $F$ amounts to a Jensen-Bregman divergence~\cite{BR-2011} (Theorem~3~of~\cite{BR-2011}):
\begin{equation}
D_\Bhat[p_{\theta_1}(x),p_{\theta_2}(x)]=\JB_F(\theta_1,\theta_2).
\end{equation}

Notice that Proposition~\ref{prop:sqrtBhat} recovers the fact that the square root of the Bhattacharyya divergence between two zero-centered normal distributions 
is a metric (proved differently in~\cite{Sra-2016}) since the set of normal distributions form an infinitely divisible exponential family.

\section{Cauchy Voronoi diagrams and dual Cauchy Delaunay complexes}\label{sec:Vor}

Let us consider the Voronoi diagram~\cite{VorOkabe-2009} of a finite set  $\calP=\{p_{\lambda_1},\ldots p_{\lambda_n}\}$ of $n$ Cauchy distributions with the location-scale parameters $\lambda_i=(l_i,s_i)\in\bbH$ for $i\in\{1,\ldots, n\}$.
We shall consider the Fisher-Rao distance $\rho_\FR$, the KL divergence $D_\KL$ and its square root metrization $\rho_\KL$, the chi square  divergence $D_{\chi^2}$, and the flat divergence $D_\flat$.

\subsection{The hyperbolic Cauchy Voronoi diagrams}
Observe that the Voronoi diagram does not change under any strictly increasing function $t$ of the dissimilarity measure (e.g., square root function): $\Vor_{D\circ t}(\calP)=\Vor_{D}(\calP)$.
Thus we get the following theorem:

\begin{Theorem}\label{thm:CVD}
The Cauchy Voronoi diagrams under the Fisher-Rao distance, the the chi-square  divergence and the Kullback-Leibler divergence all coincide, and amount to a hyperbolic Voronoi diagram on the corresponding location-scale parameters.
\end{Theorem}

\begin{proof}
The KL divergence can be expressed as
\begin{eqnarray}
D_\KL[p_{l_1,s_1}:p_{l_2,s_2}] &=& \log  \left(1+\frac{1}{2}\delta(l_1,s_1,l_2,s_2) \right).
\end{eqnarray}
Thus both the $D_\KL$ and $\rho_\FR$ dissimilarities are expressed as strictly increasing functions of $\delta$ (a synonym for the $D_{\chi^2}$ divergence).
Therefore the {\em Voronoi bisectors} between two Cauchy distributions $p_{l_1,s_1}$ and $p_{l_2,s_2}$ for $D\in\{\rho_\FR, D_{\KL}, \sqrt{D_{\KL}}, D_{\chi^2}\}$  amounts to the same expression:
\begin{eqnarray}
\mathrm{Bi}_D(p_{\lambda_1}:p_{\lambda_2})&=&\left\{\lambda\in \bbH \ :\  \delta(\lambda,\lambda_1)=\delta(\lambda,\lambda_2) \right\},\\
\mathrm{Bi}_D(p_{l_1,s_1}:p_{l_2,s_2})&=&\left\{(l,s)\in \bbH \ :\  \delta(l,s,l_1,s_1)=\delta(l,s,l_2,s_2) \right\}.
\end{eqnarray}
\end{proof}

It follows that we can calculate the Cauchy Voronoi diagram of $n$  Cauchy distributions in optimal $\Theta(n\log n)$ time by calculating the 2D hyperbolic Voronoi diagram~\cite{HVD-2010,HVD-2014} on the location-scale parameters (see Appendix~\ref{sec:hpd} for details).
Figure~\ref{fig:hvd} displays the Voronoi diagram of a set of Cauchy distributions by its equivalent parameter hyperbolic Voronoi diagram in the Poincar\'e upper plane model, the Poincar\'e disk model, and the Klein disk model.
A model of hyperbolic geometry is said {\em conformal} if it preserves angles, i.e., its underlying Riemannian metric tensor is a scalar  positive function of the Euclidean metric tensor.
The Poincar\'e disk model and the Poincar\'e upper plane model are both conformal models~\cite{anderson2006hyperbolic}.
The Klein model is {\em not} conformal, except at the disk origin. 
Let $\bbD=\{ p\ :\ \|p\|<1 \}$ denote the open unit disk domain for the Poincar\'e and Klein disk models.
Indeed, the Riemannian metric corresponding to the Klein disk model is
\begin{equation}
\ds_{\mathrm{Klein}}^2(p) =  
\frac{\mathrm{d} s_{\Eucl}^{2}}{1-\|p\|^{2}}
+\frac{\langle p, \mathrm{d} p\rangle}{\left(1-\|p\|^{2}\right)^{2}},
\end{equation}
where $\mathrm{d}p=\dx+\dy$ and $\ds_\Eucl=\sqrt{\dx^2+\dy^2}$ denotes the Euclidean line element.
Since $\ds_{\mathrm{Klein}}^2(0)=\mathrm{d} s_{\Eucl}^{2}$, we deduce that Klein model is conformal at the origin (when measuring the angles between two vectors $v_1$ and $v_2$ of the tangent plane $T_0$).

\begin{figure}%
\centering
\fbox{\includegraphics[width=0.8\columnwidth]{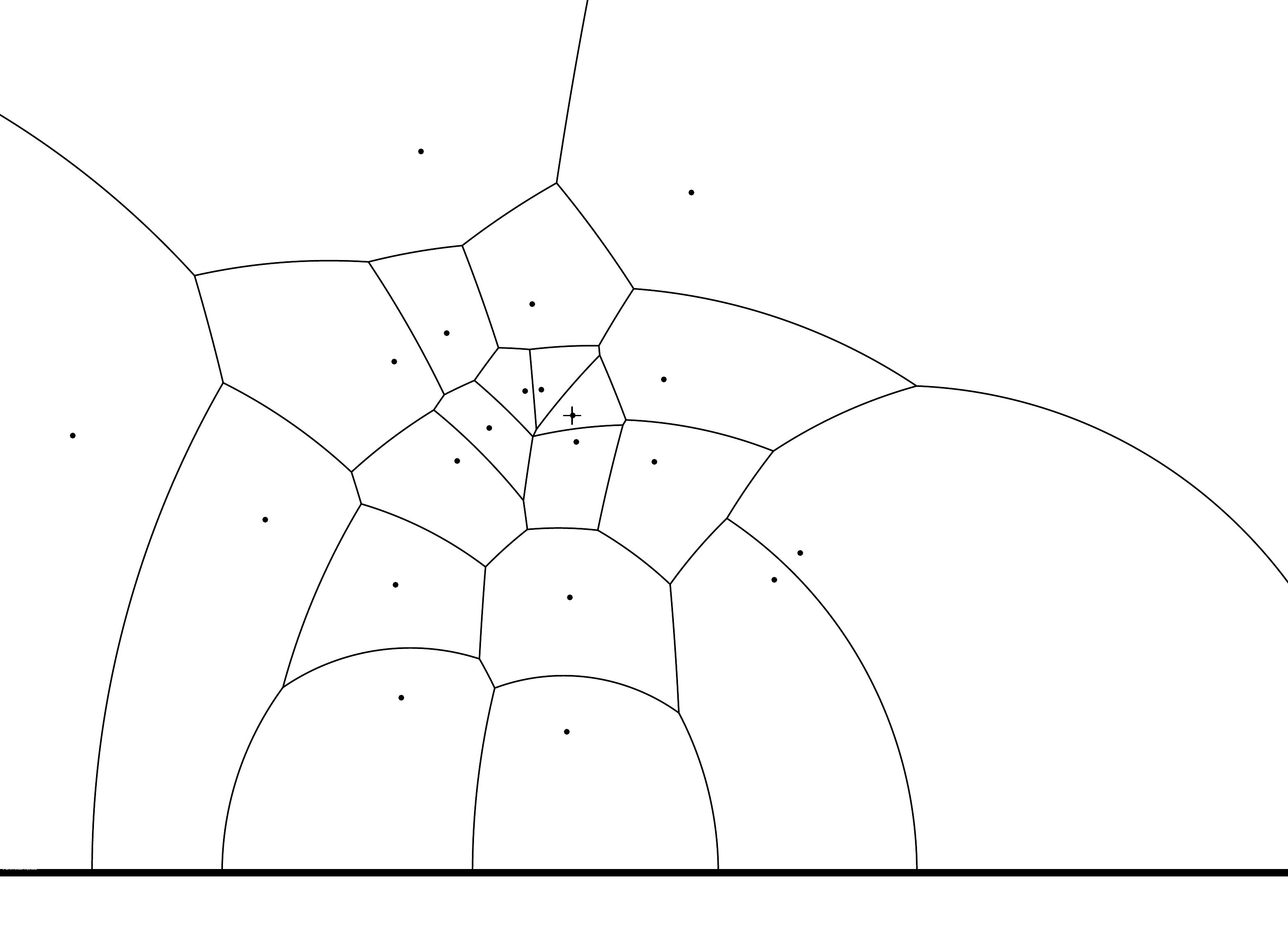}}\\
\fbox{\includegraphics[width=0.4\columnwidth]{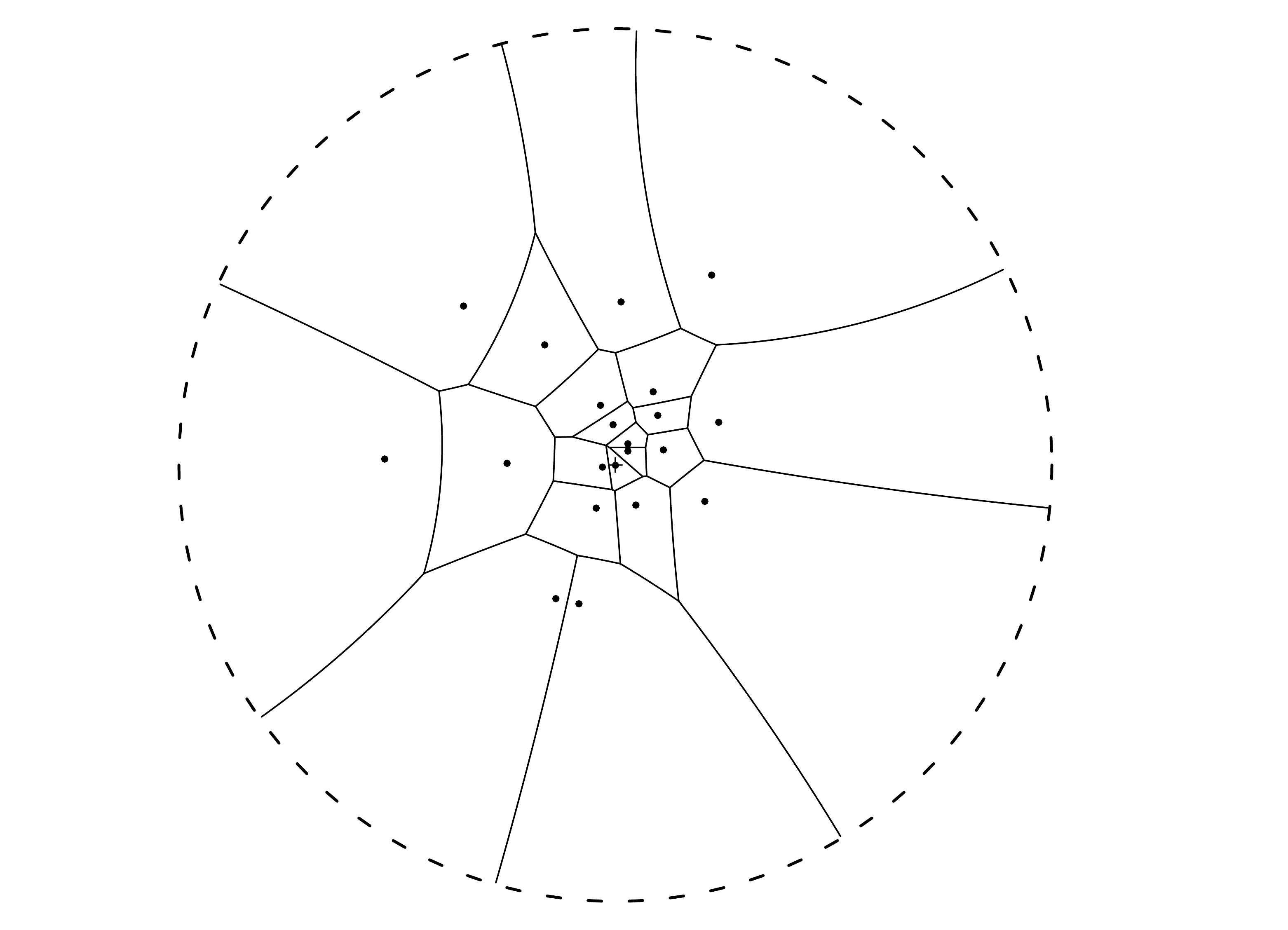}}
\fbox{\includegraphics[width=0.4\columnwidth]{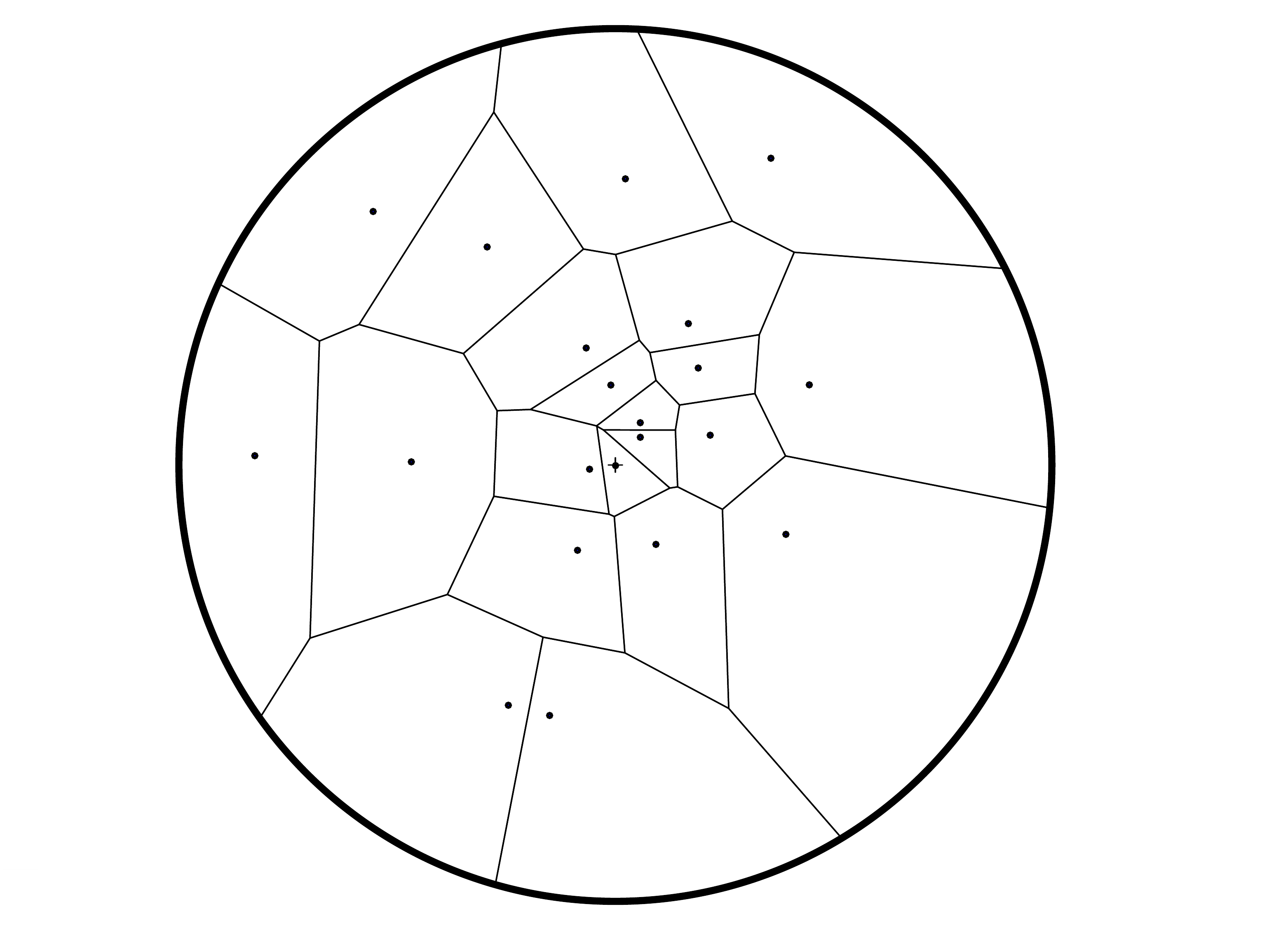}}

\caption{Hyperbolic Voronoi diagram of  a set of Cauchy distributions in the Poincar\'e upper plane (top), the Poincar\'e disk model (bottom left), and the Klein disk model (bottom right).}%
\label{fig:hvd}%
\end{figure}

\begin{figure}
\centering

\includegraphics[bb=0 0 512 512,width=0.4\columnwidth]{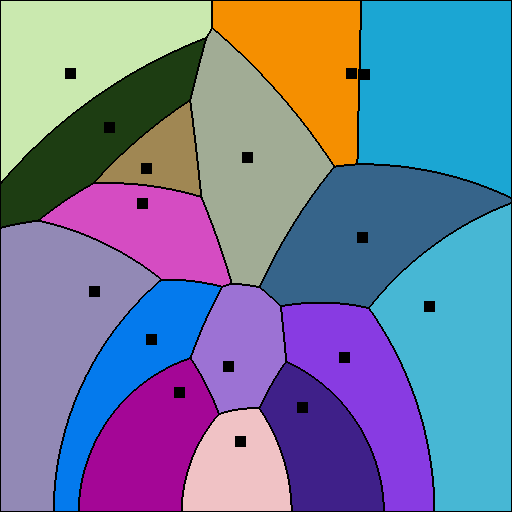}

\caption{A hyperbolic Cauchy Voronoi diagram of a finite set of Cauchy distributions (black square generators, colored Voronoi cells, and black cell borders).} \label{fig:VoronoiCauchy}
\end{figure}

\begin{figure}
\centering
\includegraphics[width=0.4\columnwidth]{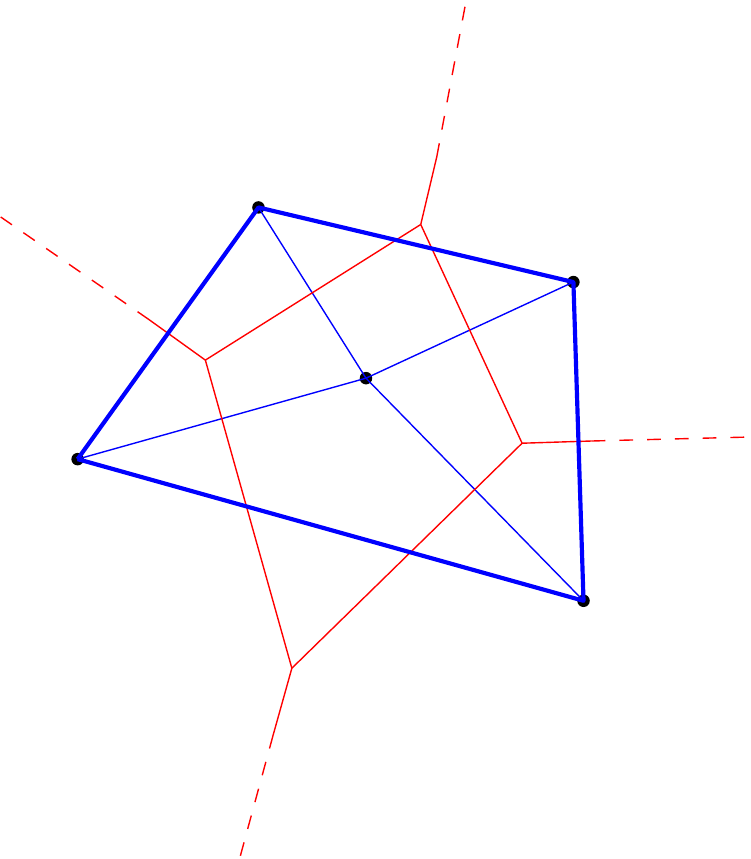}

\caption{Duality between the ordinary Euclidean Voronoi diagram and the Delaunay structures: The Voronoi diagram partitions the space into Voronoi  proximity cells. The Delaunay complex triangulates the convex hull of the generators. A Delaunay edge is drawn between the generators of adjacent Voronoi cells. Observe that the Delaunay edges cuts orthogonally the corresponding Voronoi bisectors in Euclidean geometry.} \label{fig:VoronoiDuality}
\end{figure}

The dual of the Voronoi diagram is called the {\em Delaunay (simplicial) complex}~\cite{BY-1998,bogdanov2013hyperbolic}:
We build the Delaunay complex by drawing an edge between generators whose  Voronoi cells are adjacent.
For the ordinary Euclidean Delaunay complex with  points in general position (i.e., no $d+2$ cospherical points in dimension $d$), the Delaunay complex triangulates the convex hull of the points~\cite{boissonnat2006curved,nielsen1998output}.
 Therefore it is called the {\em Delaunay triangulation}~\cite{VorOkabe-2009,boissonnat2006curved,cheng2012delaunay}.
Similarly, for the hyperbolic Voronoi diagram, we construct the {\em hyperbolic Delaunay complex} by drawing a hyperbolic geodesic edge between any two generators whose Voronoi cells are adjacent. However, we do not necessarily obtain anymore a geodesic triangulation of the hyperbolic geodesic convex hull but rather a simplicial complex, hence the name {\em hyperbolic Delaunay complex}~\cite{bogdanov2013hyperbolic,tanuma2011revisiting,deblois2018delaunay}.
In extreme cases, the hyperbolic Delaunay complex has a tree structure. 
See Figure~\ref{fig:treeDel} for examples of a hyperbolic Delaunay triangulation and a  hyperbolic Delaunay complex  which is not a triangulation
In fact, hyperbolic geometry is very well-suited for embedding isometrically with low distortion weighted tree graphs~\cite{Sarkar-2011}. Hyperbolic embeddings of hierarchical structures~\cite{nickel2017poincare} has become a hot topic in machine learning.

\begin{figure}%
\centering
\fbox{\includegraphics[width=0.45\columnwidth]{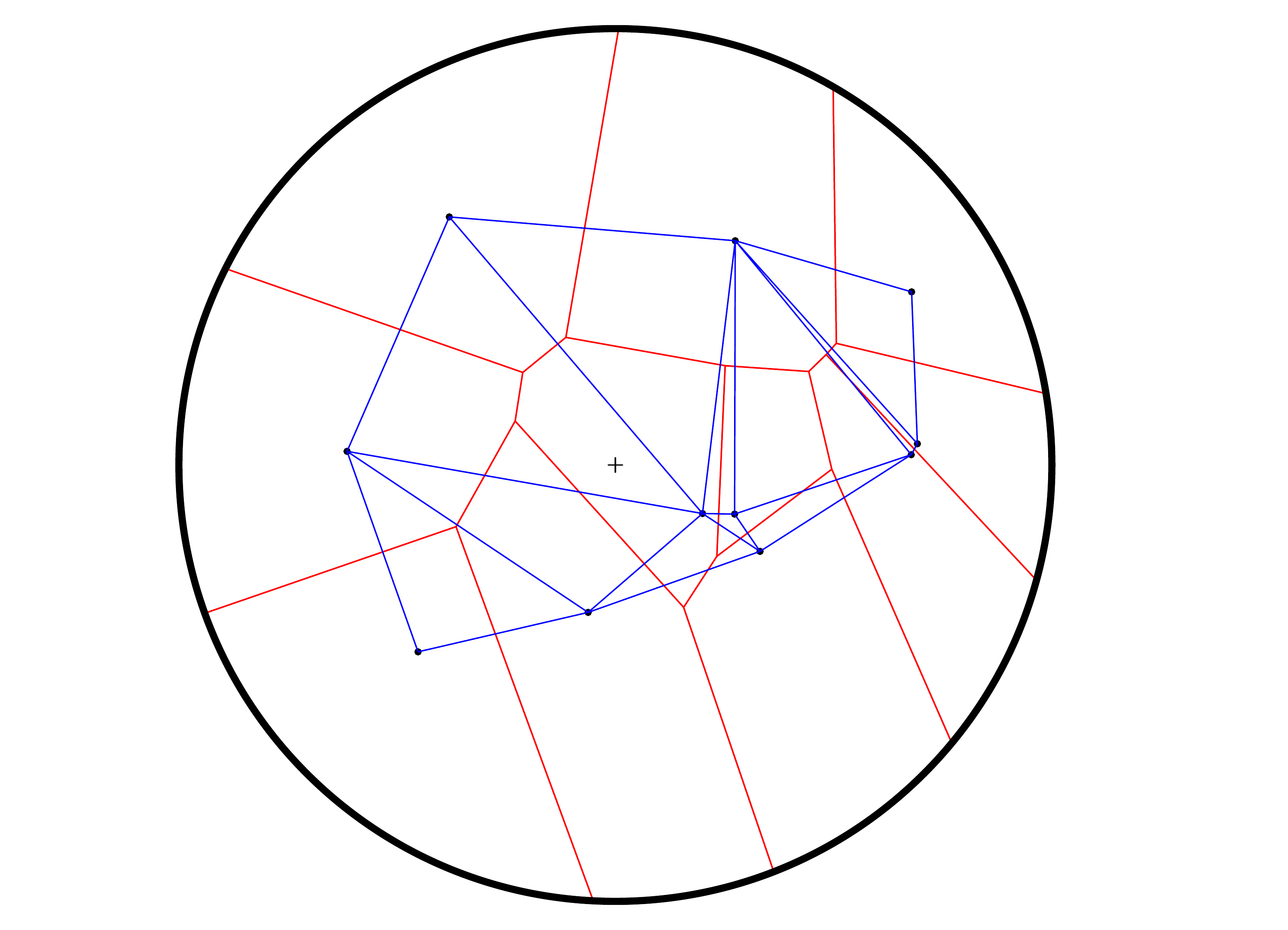}}
\fbox{\includegraphics[width=0.45\columnwidth]{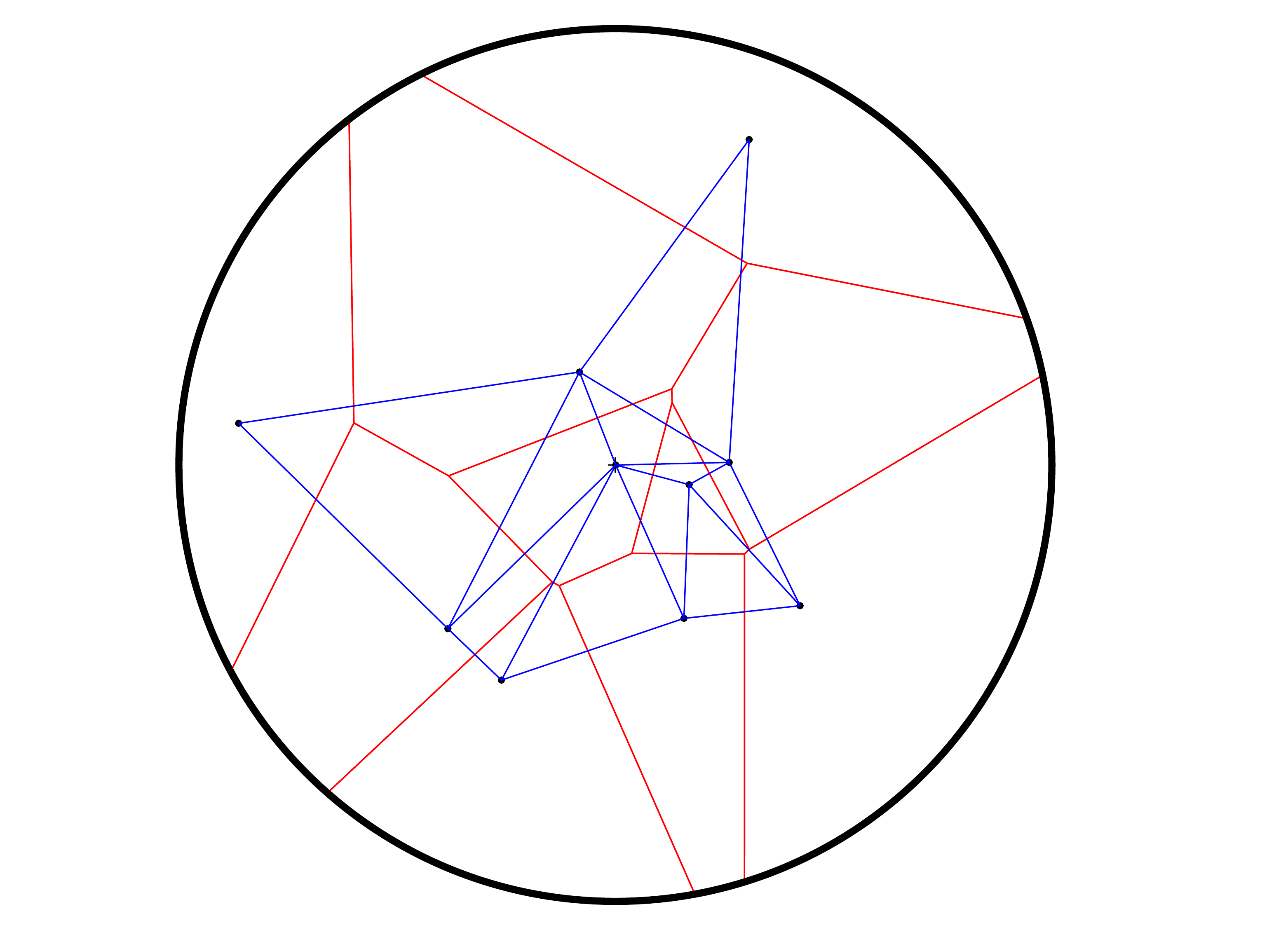}}\\

\fbox{\includegraphics[width=0.45\columnwidth]{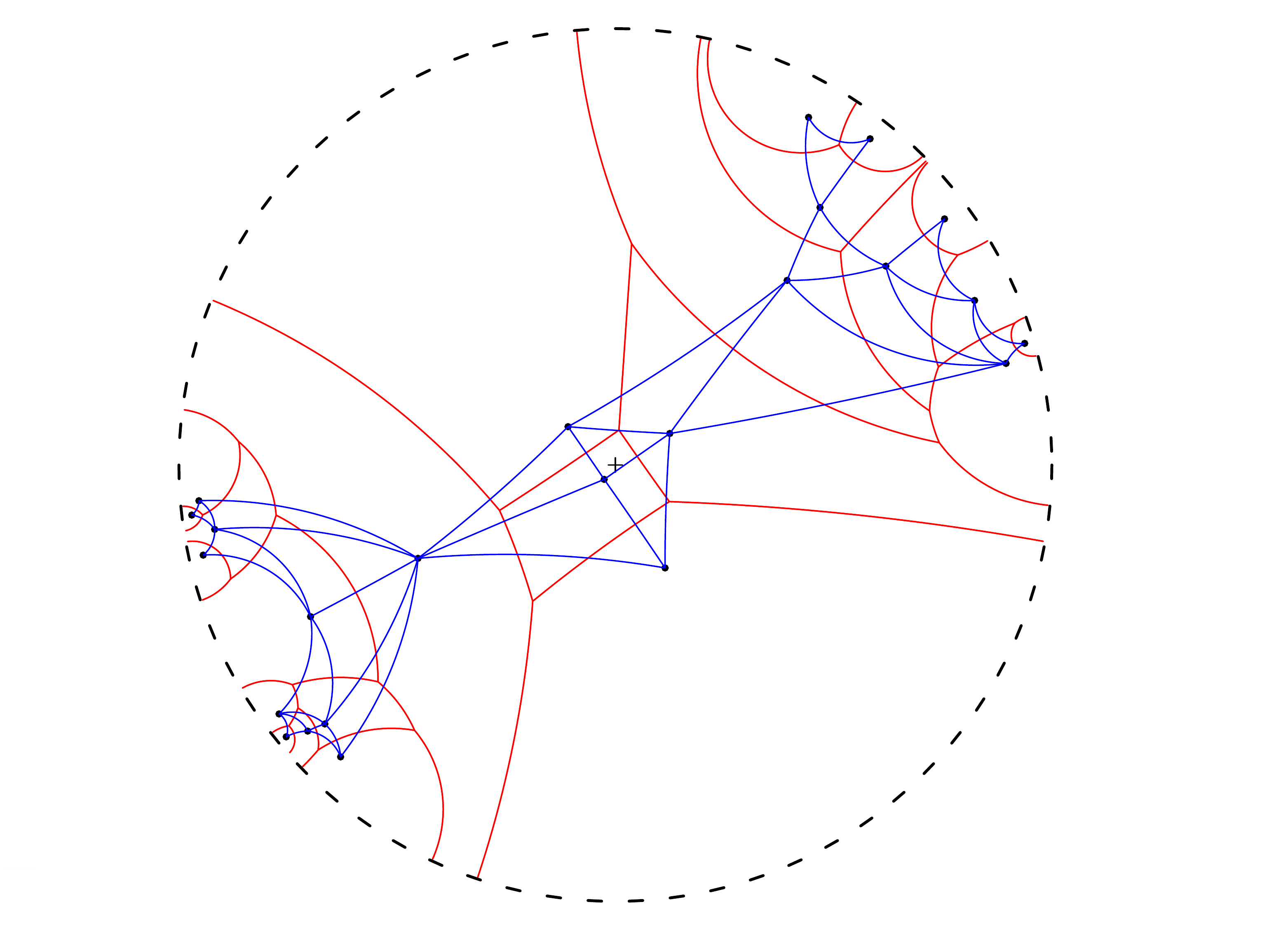}}
\fbox{\includegraphics[width=0.45\columnwidth]{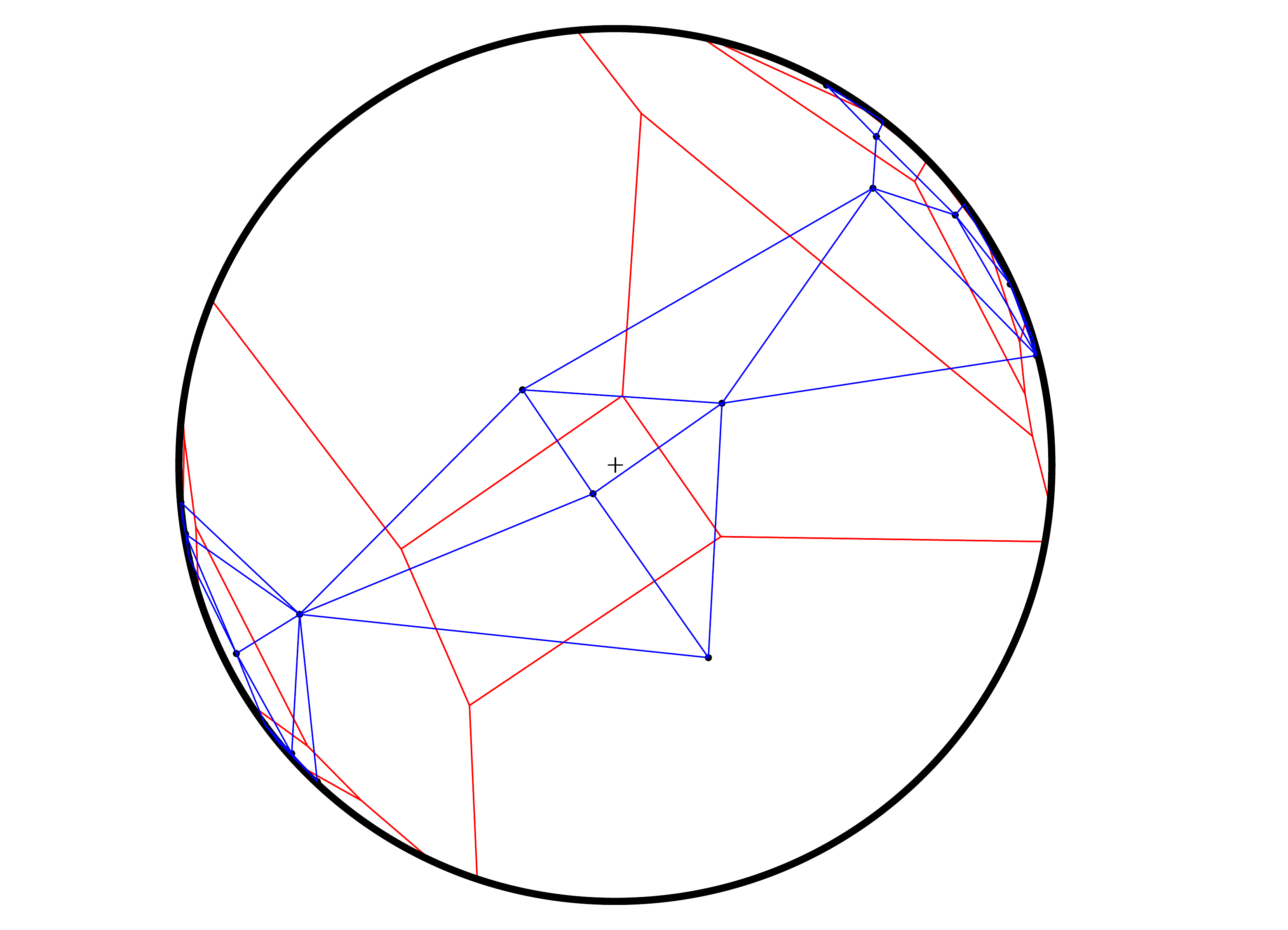}}

\caption{Examples of hyperbolic Voronoi Delaunay complexes drawn in the Klein model: Delaunay complex triangulates the convex hull yielding the Delaunay triangulation (top left), and Delaunay complex which does not triangulate the convex hull, 
 (top right). Bottom: A hyperbolic Voronoi diagram and its dual Delaunay complex displayed in the Poincar\'e disk model (left) and in the Klein disk model (right).} 
\label{fig:treeDel}%
\end{figure}

\begin{figure}%
\centering

\begin{tabular}{cc}
\fbox{\includegraphics[width=0.4\columnwidth]{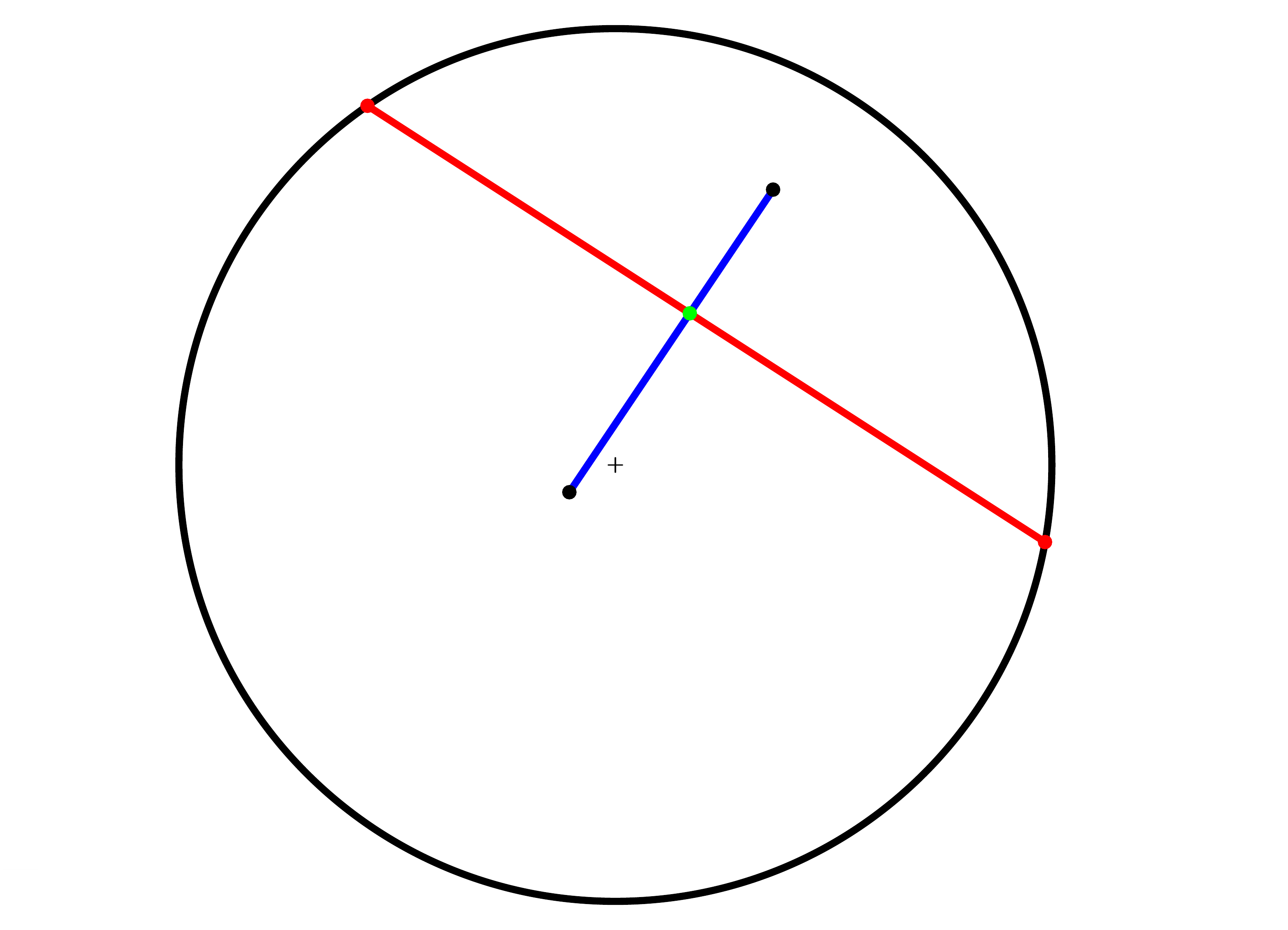}} & 
\fbox{\includegraphics[width=0.4\columnwidth]{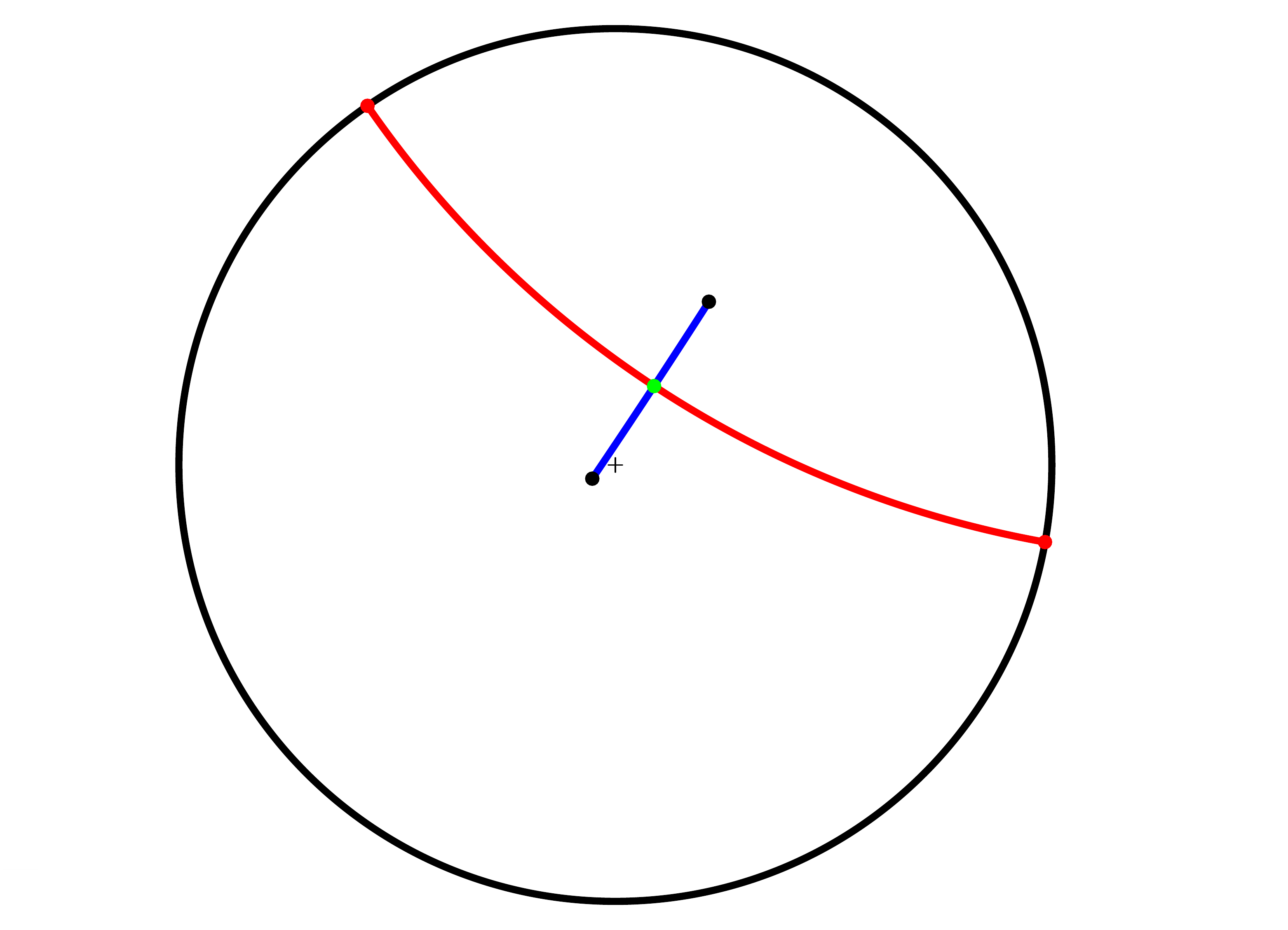}}\\
non-conformal (Klein) & conformal (Poincar\'e) \\
\fbox{\includegraphics[width=0.4\columnwidth]{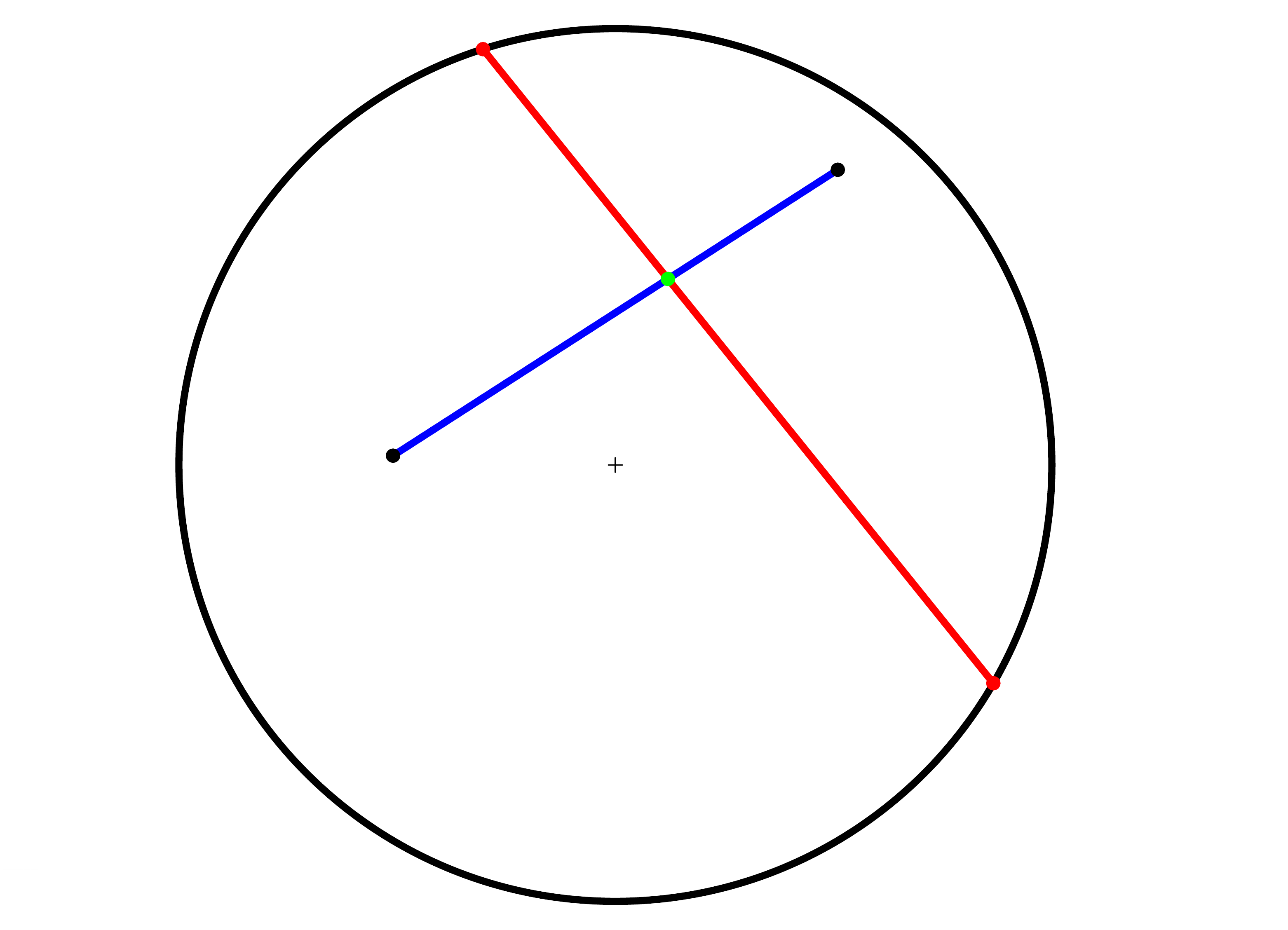}} &
\fbox{\includegraphics[width=0.4\columnwidth]{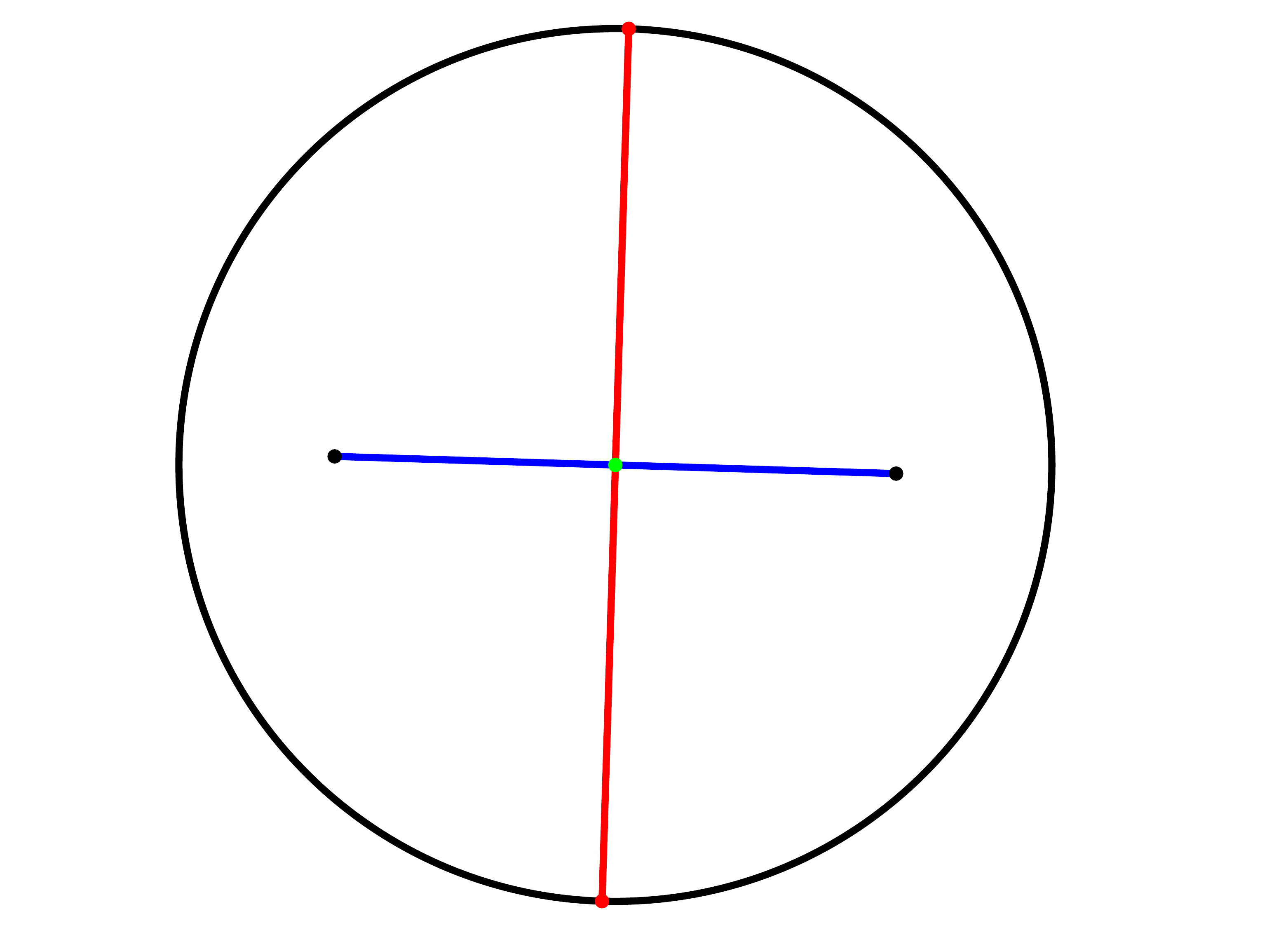}}\\
non-conformal (Klein) & conformal at the origin (Klein)
\end{tabular}

\caption{In hyperbolic geometry, the Voronoi bisector between two generators is orthogonal to the geodesic linking them. 
The top figures display a pair of (bisector,geodesic) in the Klein model (left), and the same pair in the Poincar\'e model (right).
When viewed in Klein non-conformal model, the bisector does not intersect orthogonally (with respect to the Euclidean geometry)  the geodesic (left) except when the intersection point is at the disk origin (bottom right).}%
\label{fig:hbiorthogonal}%
\end{figure}

Let us now prove that these Cauchy hyperbolic Voronoi/Delaunay structures are {\em Fisher orthogonal}:

\begin{Theorem}\label{thm:CVDortho}
The Cauchy Voronoi diagram is Fisher orthogonal to the Cauchy Delaunay complex.
\end{Theorem}

\begin{proof}
It is enough to prove that the corresponding hyperbolic geodesic $\gamma(p_{\lambda_1},p_{\lambda_2})$ is orthogonal to the bisector 
$\Bi(p_{\lambda_1}:p_{\lambda_2})$.
The distance in the Klein disk model is
\begin{equation}
\rho_{\mathrm{Klein}}(p,q) = \rho_{K}(p,q)   \eqdef 
\operatorname{arccosh}\left( \frac{1-\langle p,q\rangle}{\sqrt{\left(1-\|p\|^{2}\right)\left(1-\|q\|^{2}\right)}} \right).
\end{equation}
The equation of the hyperbolic bisector in the Klein disk model~\cite{HVD-2010} is
\begin{equation}
\mathrm{Bi}_{\rho_{\mathrm{Klein}}}({\lambda_1}:{\lambda_2}) = \left\{ \lambda\in\mathbb{D} \ :\ \lambda^\top \left(\sqrt{1-\|\lambda_1\|^{2}} \lambda_2-\sqrt{1-\|\lambda_2\|^{2}} \lambda_1\right) +\sqrt{1-\|\lambda_2\|^{2}}-\sqrt{1-\|\lambda_1\|^{2}}=0\right\} . 
\end{equation}

Using a M\"obius transformation~\cite{HVD-2010} (i.e., a hyperbolic ``rigid motion''), we may consider without loss of generality 
that $p_{\lambda_1}=-p_{\lambda_2}$.
It follows that the bisector equation writes simply as
\begin{equation}\label{eq:originbi}
\mathrm{Bi}_{{\rho_\mathrm{Klein}}} = \left\{ \lambda \st 2\sqrt{1-\|p_{\lambda_1}\|} \lambda^\top \lambda_1=0\right\}.
\end{equation}
Since the Klein disk model is conformal at the origin, we deduce from Eq.~\ref{eq:originbi} that 
we have $\gamma(p_{\lambda_1},p_{\lambda_2})\perp \Bi(p_{\lambda_1}:p_{\lambda_2})$.
\end{proof}

Figure~\ref{fig:hbiorthogonal} displays two bisectors with their corresponding geodesics in the Klein model.
We check that the Euclidean angles are deformed when the intersection point is not at the disk origin.
Appendix~\ref{sec:hpd} provides further details for the efficient construction of the hyperbolic Voronoi diagram in the Klein model.

\begin{Remark}
The hyperbolic Cauchy Voronoi diagram can be used for {\em classification tasks} in statistics as originally motivated by C.R. Rao in his celebrated paper~\cite{Rao-1945}:
Let $p_{\lambda_1},\ldots,p_{\lambda_n}$ be $n$ Cauchy distributions, and $x_1,\ldots, x_s$ be $s$ identically and independently samples drawn from a Cauchy distribution $p_\lambda$. We can estimate $\hat{\lambda}$ the location-scale parameters from the $s$ samples~\cite{haas1970inferences}, and then decide the multiple test hypothesis $H_i: p_\lambda=p_{\lambda_i}$ by choosing the hypothesis $H_i$ such that 
$\rho_\FR(p_{\lambda_i}, p_\lambda)\leq \rho_\FR(p_{\lambda_j}, p_\lambda)$ for all $j\in\{1,\ldots, n\}$.
This classification task amounts to perform a nearest neighbor query in the Fisher-Rao hyperbolic Cauchy Voronoi diagram. 
Hypothesis testing for comparing location parameters based on Rao's distance is investigated in~\cite{gimenez2016geodesic}.
\end{Remark}

Figure~\ref{fig:hvd300} displays the hyperbolic Voronoi Cauchy diagram induced by $300$ Cauchy distribution generators.

\begin{figure}%
\centering
\fbox{\includegraphics[width=0.30\columnwidth]{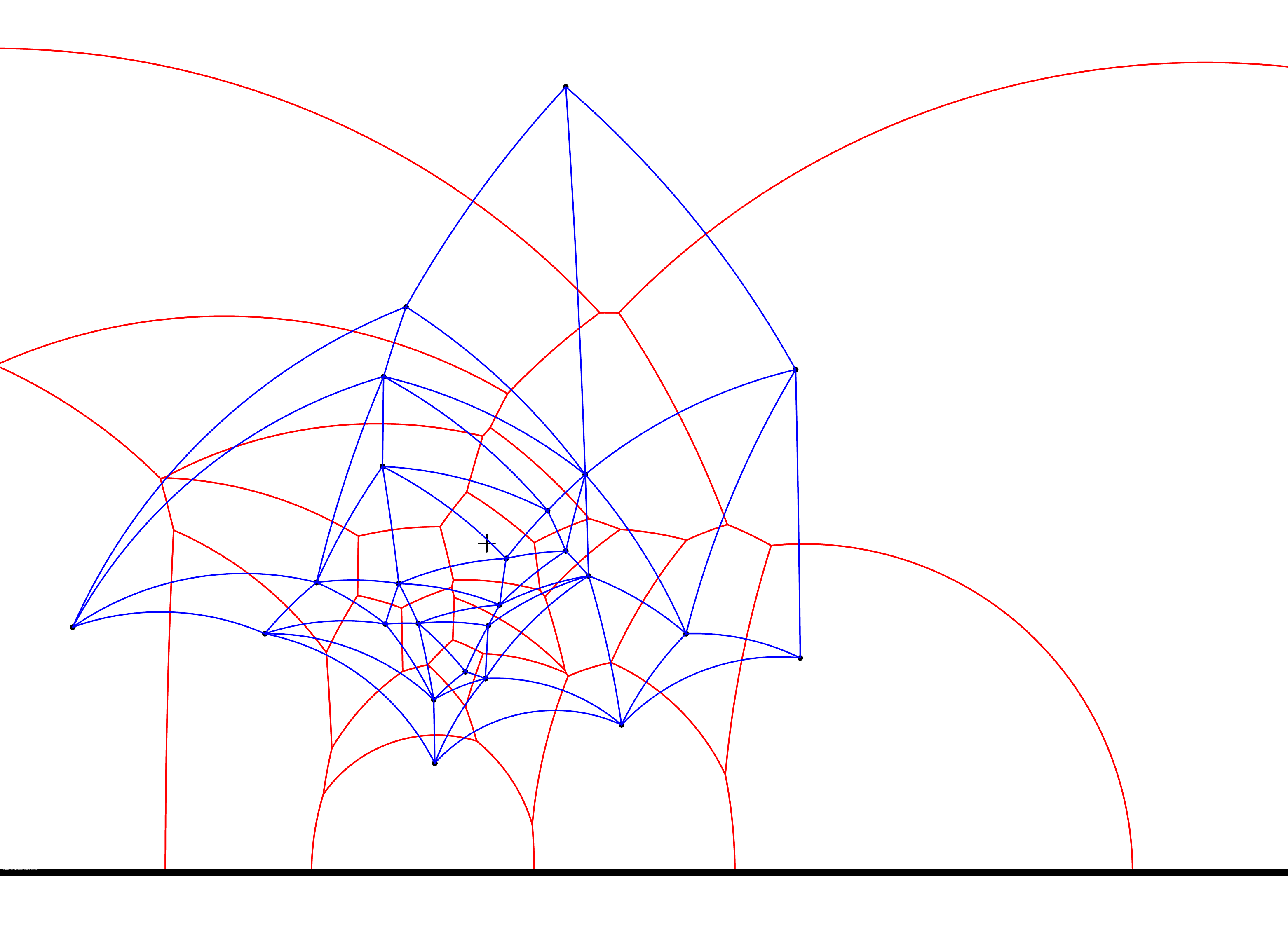}} 
\fbox{\includegraphics[width=0.30\columnwidth]{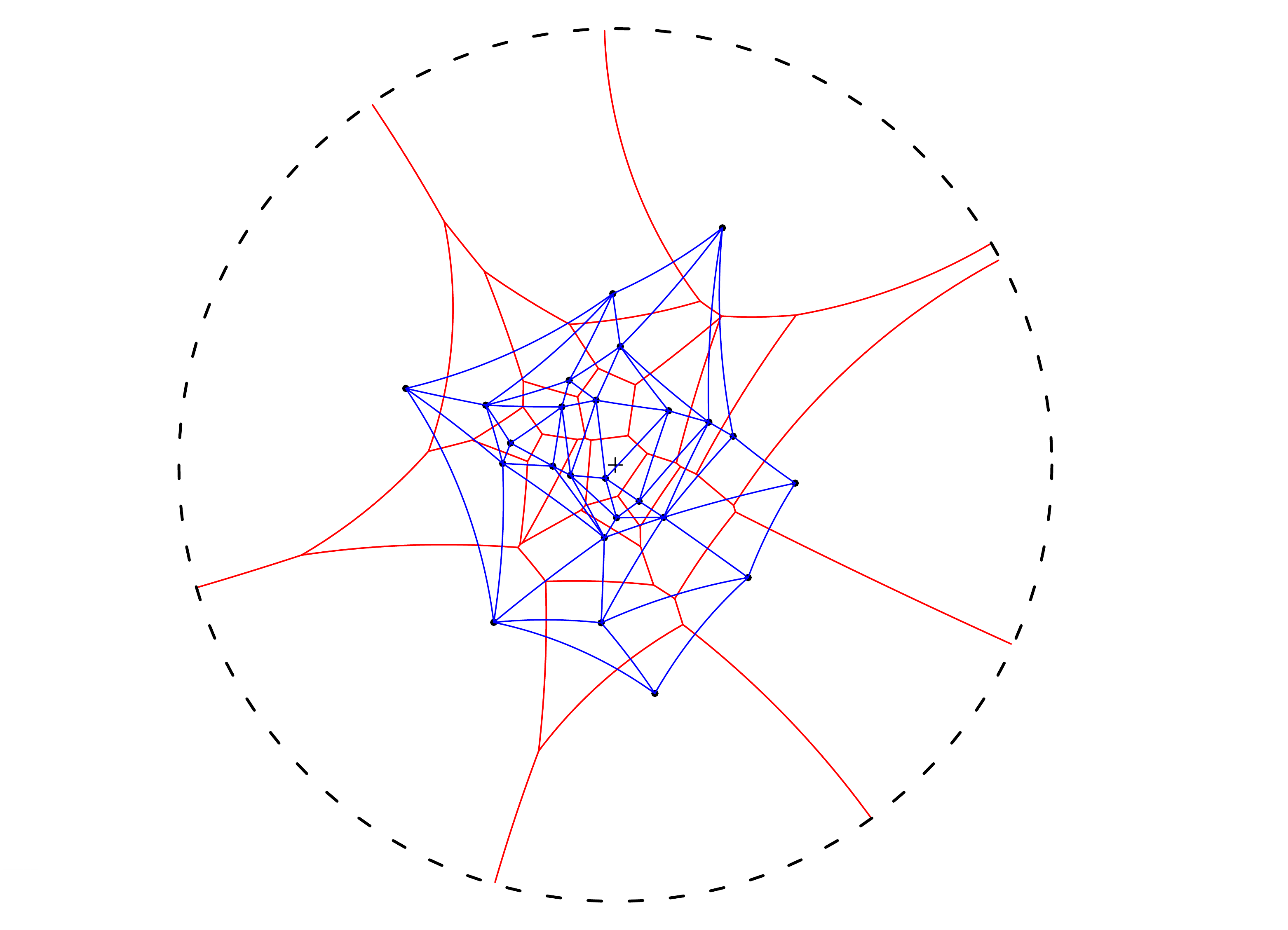}}
\fbox{\includegraphics[width=0.30\columnwidth]{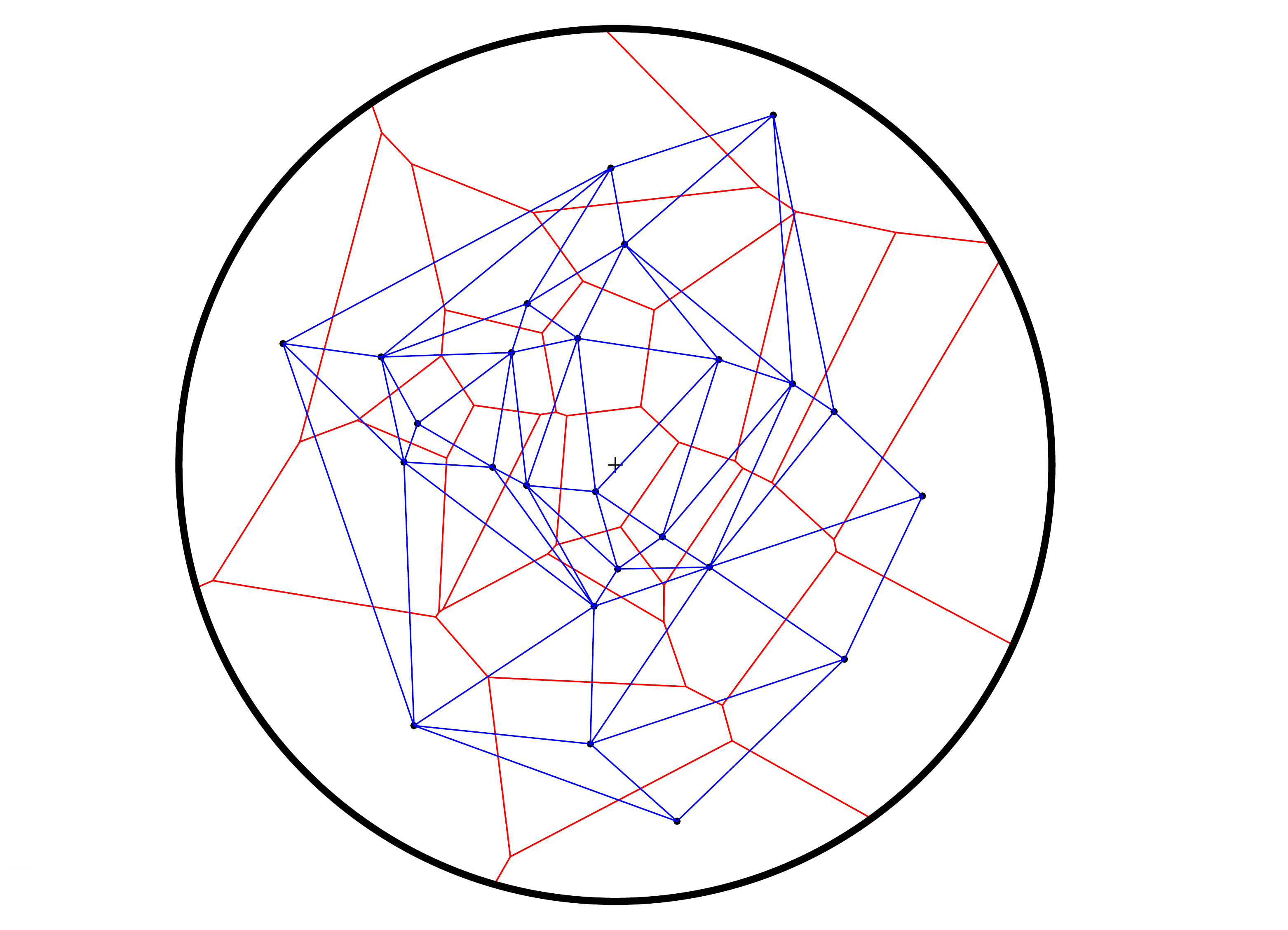}}\\
\fbox{\includegraphics[width=0.30\columnwidth]{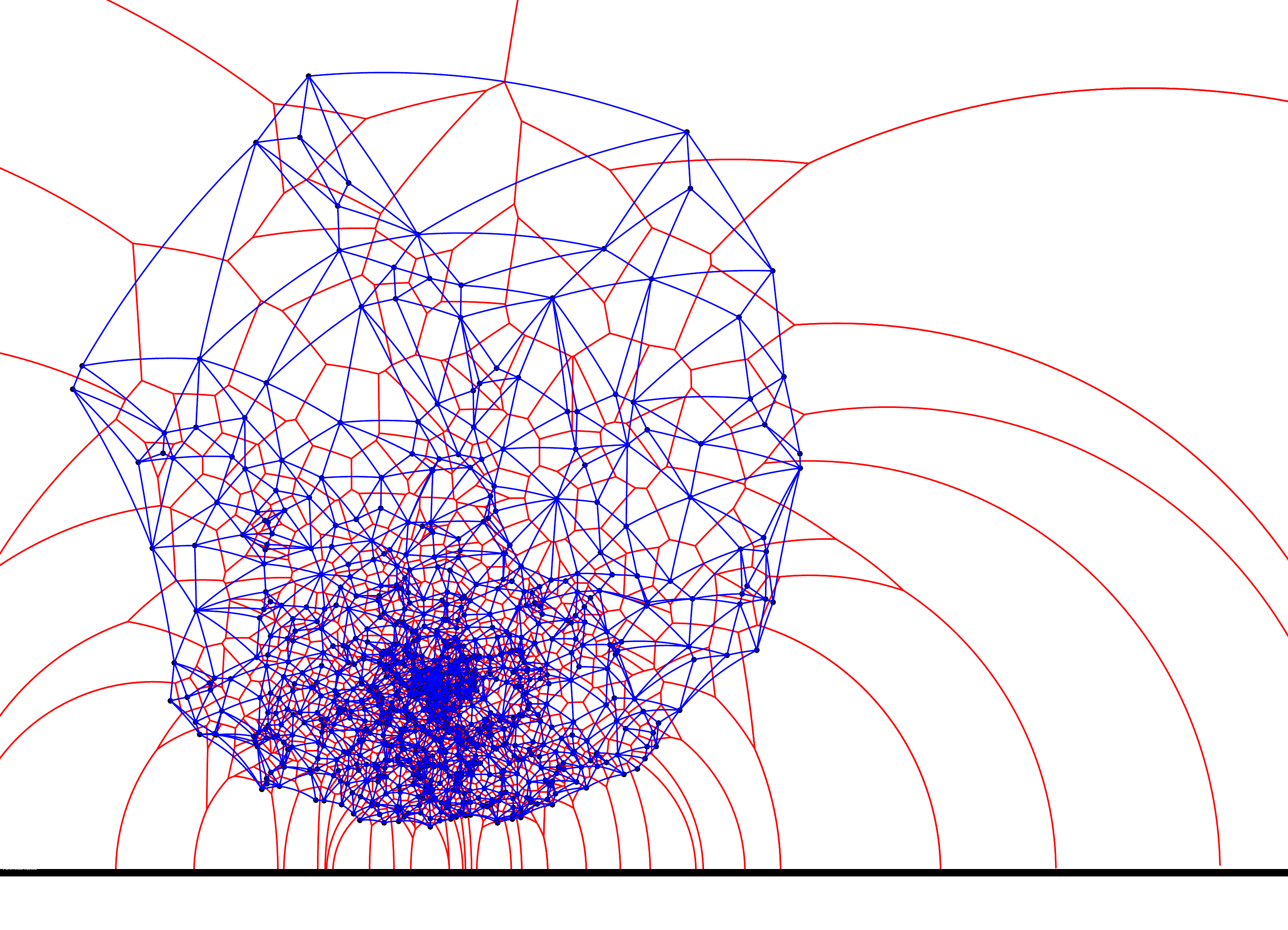}} 
\fbox{\includegraphics[width=0.30\columnwidth]{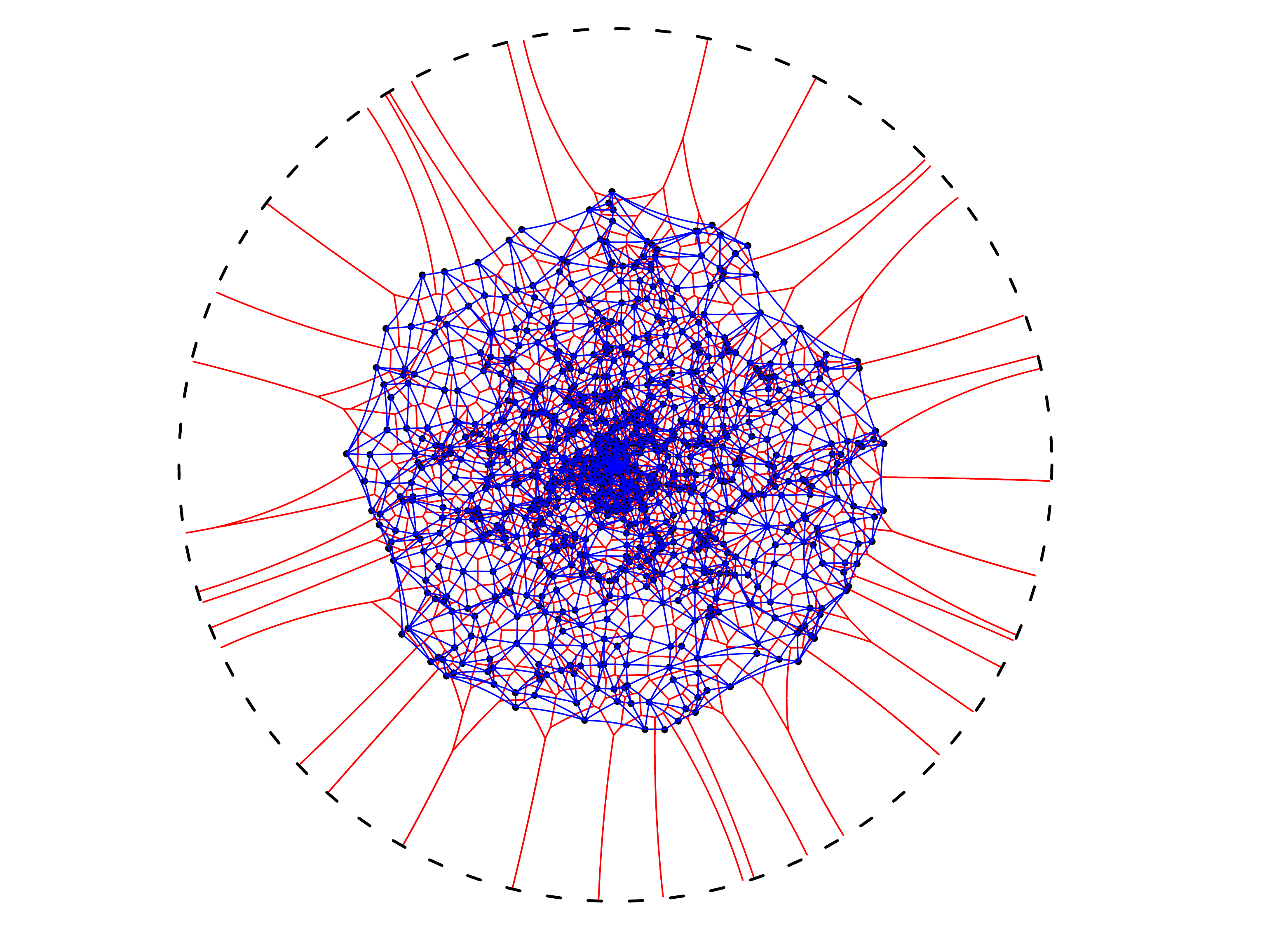}}
\fbox{\includegraphics[width=0.30\columnwidth]{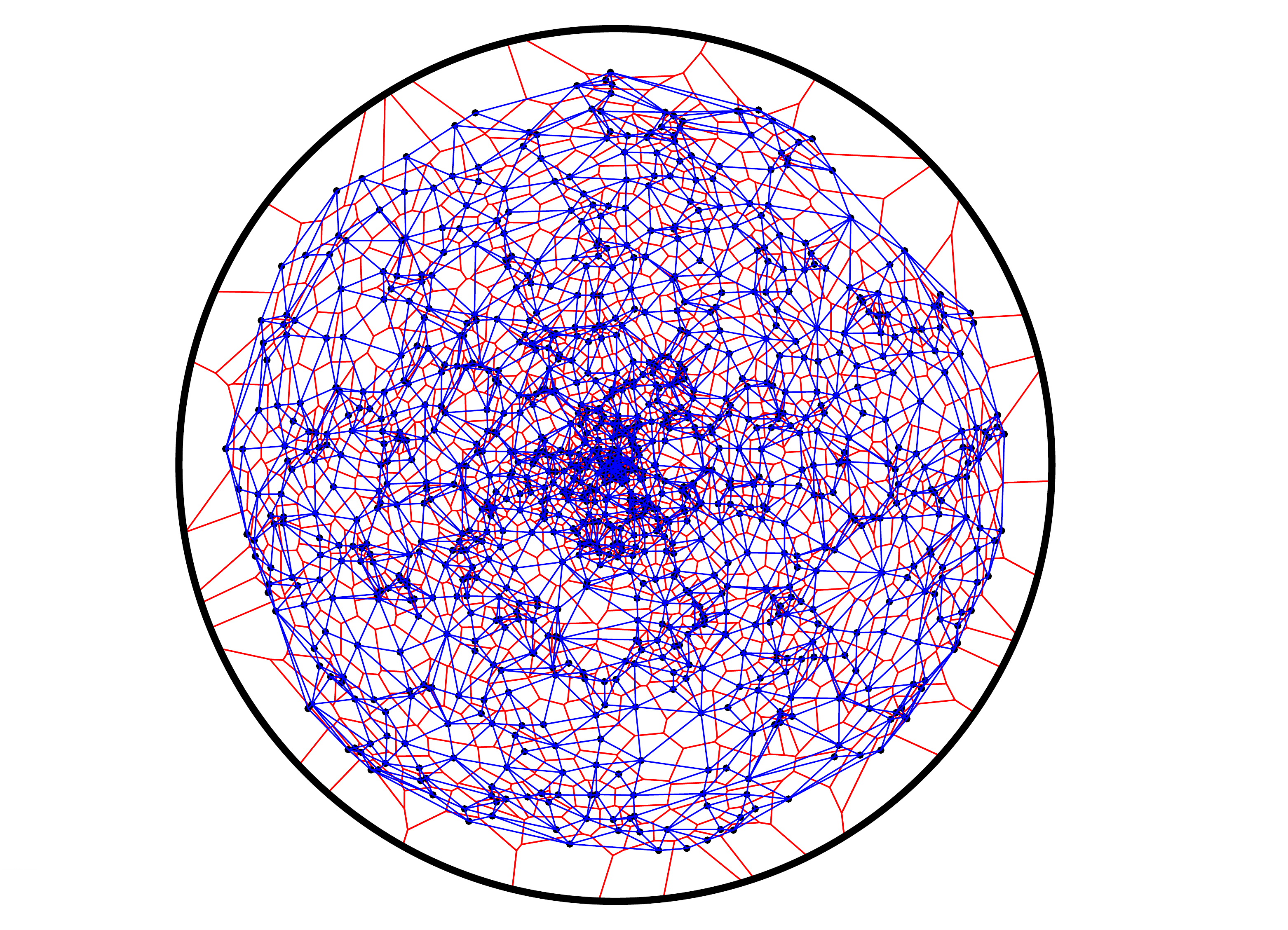}}\\
\fbox{\includegraphics[width=0.3\columnwidth]{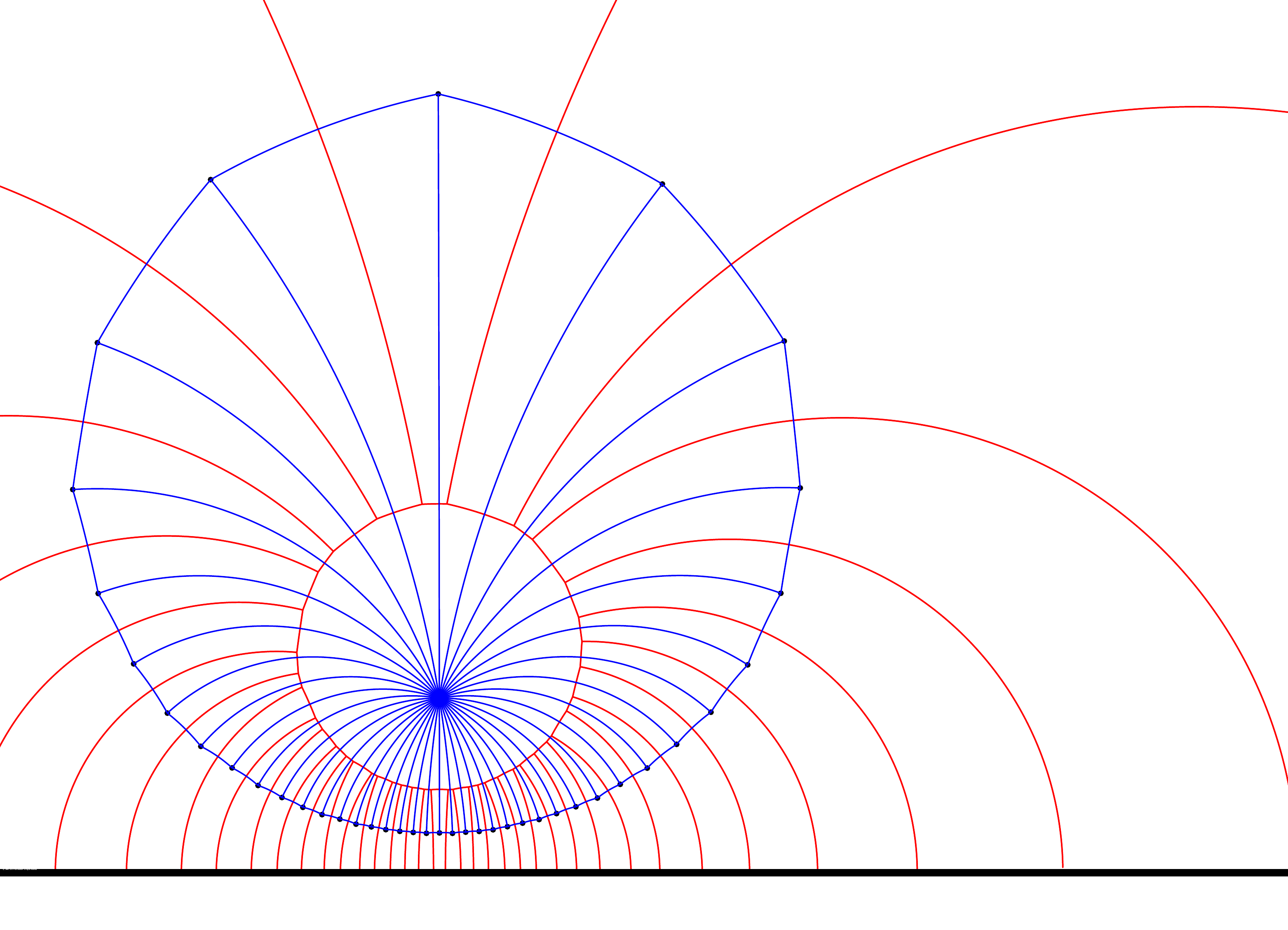}}
\fbox{\includegraphics[width=0.3\columnwidth]{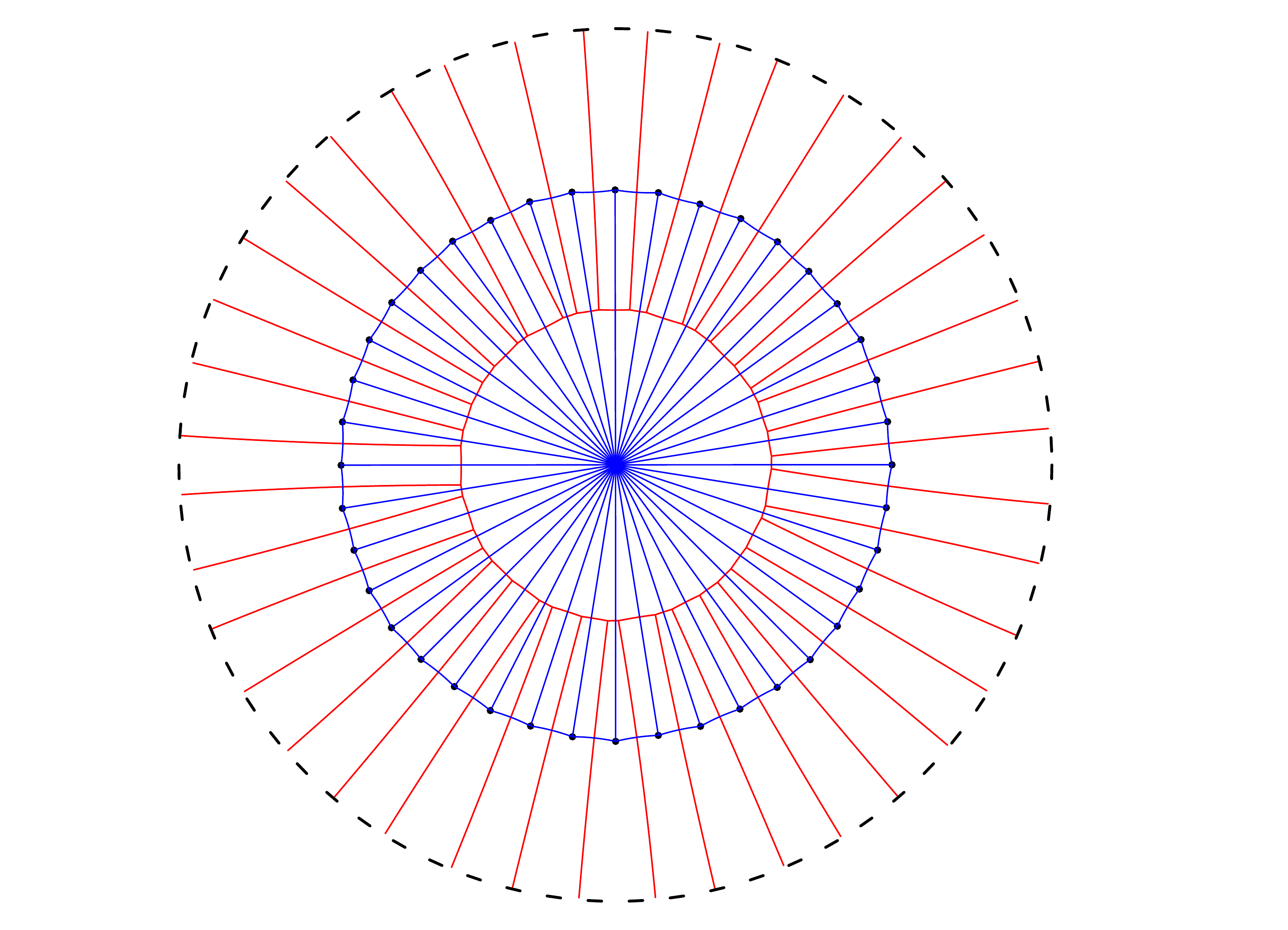}}
\fbox{\includegraphics[width=0.3\columnwidth]{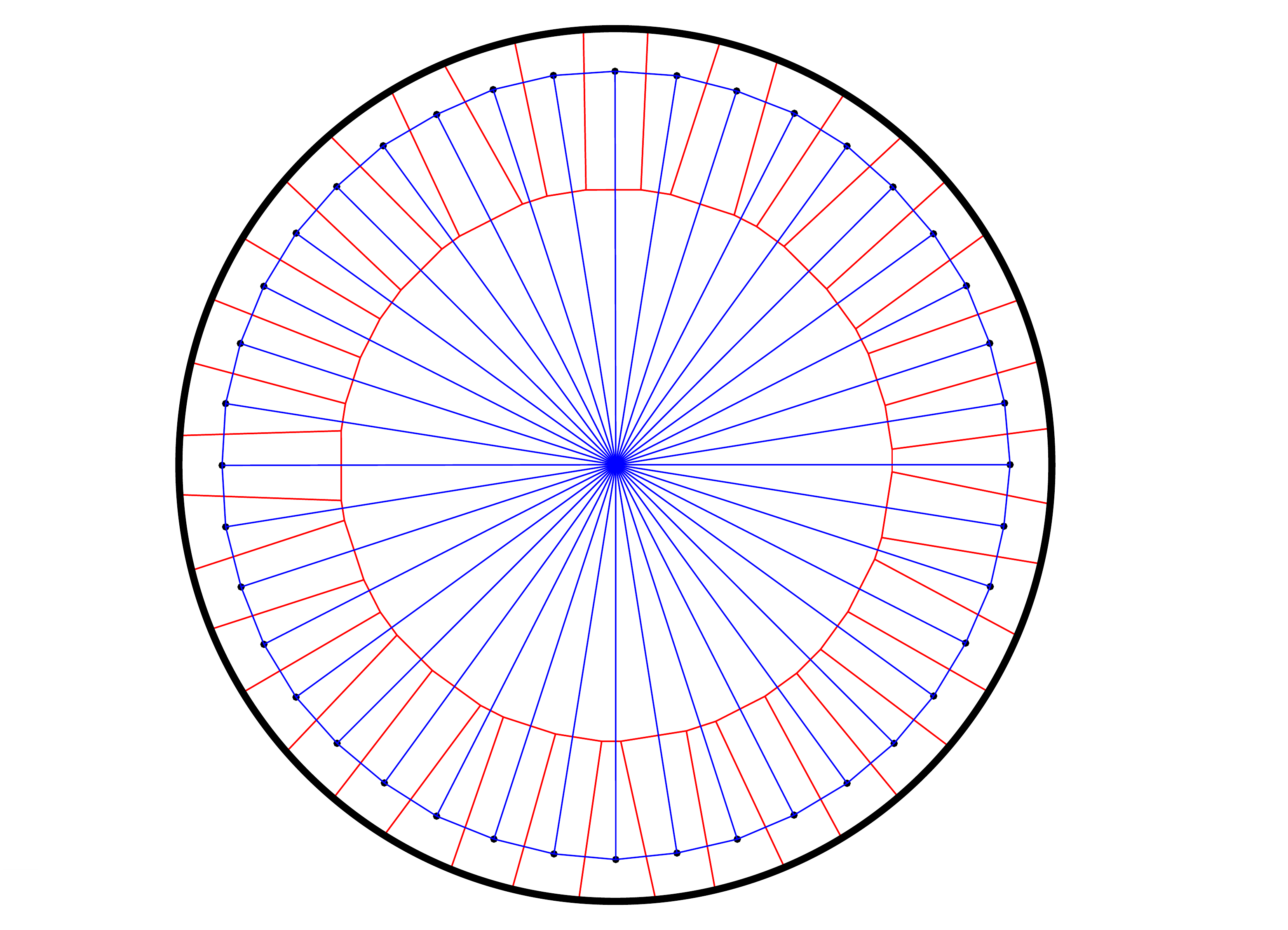}}

\caption{Equivalent hyperbolic Voronoi diagram and dual Delaunay complex of  a set of  Cauchy distributions in the Poincar\'e upper plane (left), the Poincar\'e disk model (middle), and the Klein disk model (right). Top row figures for $n=24$ Cauchy distributions, middle row figures for $n=1024$ distributions and bottom row figures for a quasi-regular set of $n=25$ Cauchy distributions.}%
\label{fig:hvd300}%
\end{figure}

Notice that it is possible to construct a set of points such that all hyperbolic Voronoi cells for that point set are unbounded.
See Figure~\ref{fig:unboundedhdv} for such an example.

\begin{figure}%
\centering
\includegraphics[width=0.4\columnwidth]{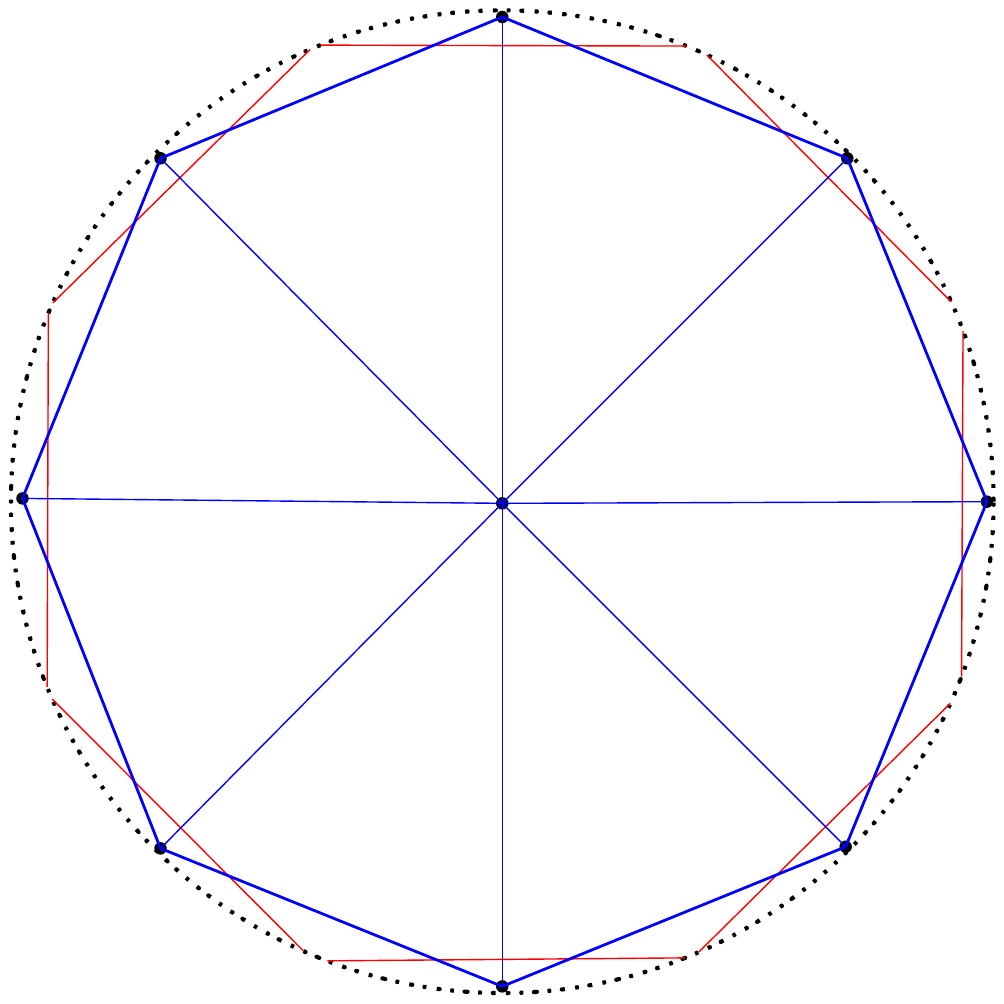}\\
\caption{A hyperbolic Voronoi diagram with all unbounded Voronoi cells.}%
\label{fig:unboundedhdv}%
\end{figure}

The ordinary Euclidean Delaunay triangulation satisfies the {\em empty sphere} property~\cite{delaunay1934sphere,BY-1998}:
That is the {\em circumscribing spheres} passing through the vertices of the Delaunay triangles of the Delaunay complex are empty of any other Voronoi site.
This property still holds for the hyperbolic Delaunay complex which is obtained by a filtration of the ordinary  Euclidean Delaunay triangulation in~\cite{bogdanov2013hyperbolic}.
A {\em hyperbolic ball} in the Poincar\'e conformal disk model or the upper plane model has the shape of a Euclidean ball with {\em displaced center}~\cite{tanuma2011revisiting}. Figure~\ref{fig:hyperbolicemptysphere} displays the Delaunay complex with the empty sphere property in the Poincar\'e and Klein disk models. The centers of these circumscribing spheres are located at the $T$-junctions of the Voronoi diagrams.

\begin{figure}%
\centering

\begin{tabular}{ll}
\fbox{\includegraphics[width=0.4\columnwidth]{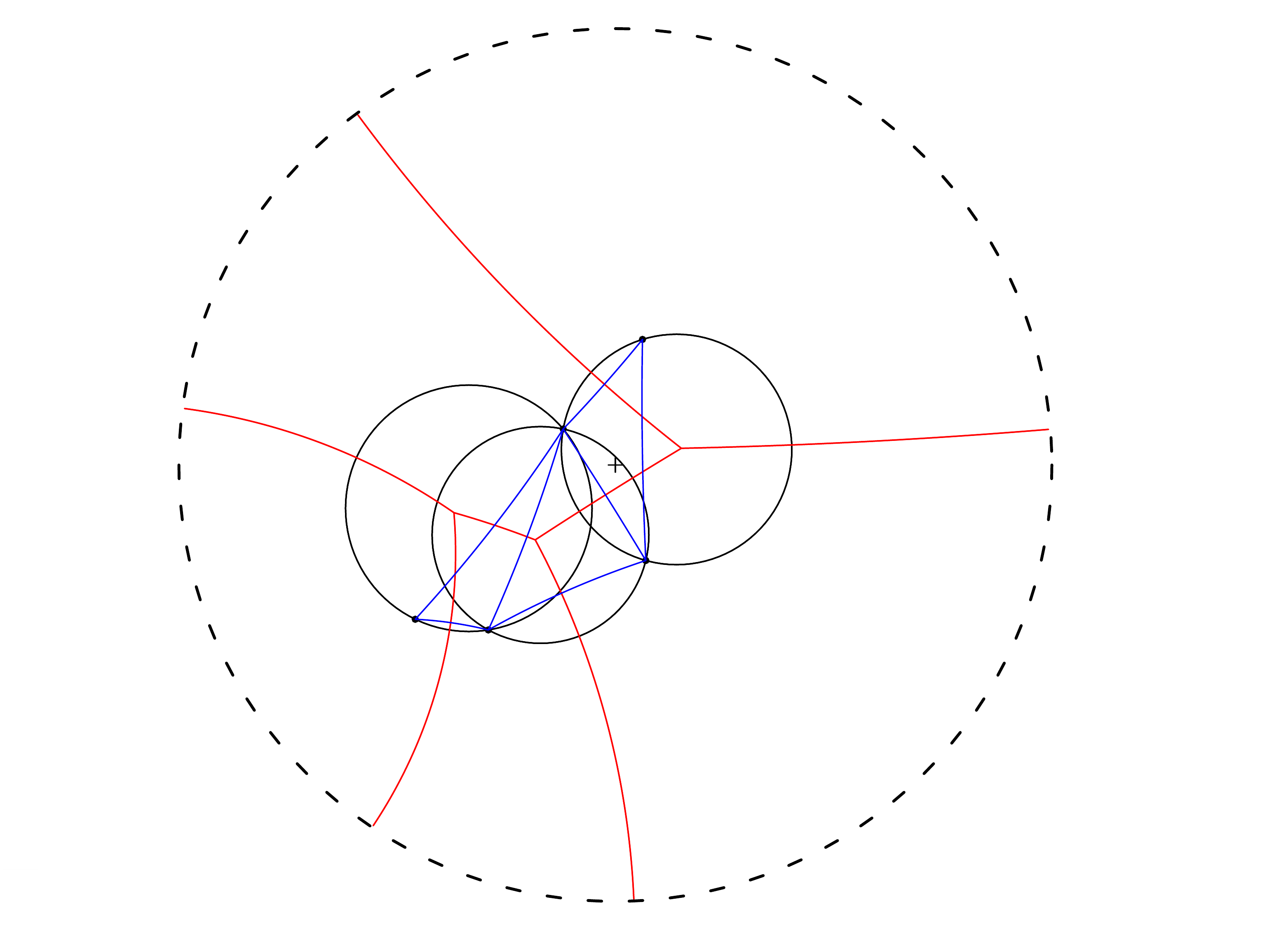}} &
\fbox{\includegraphics[width=0.4\columnwidth]{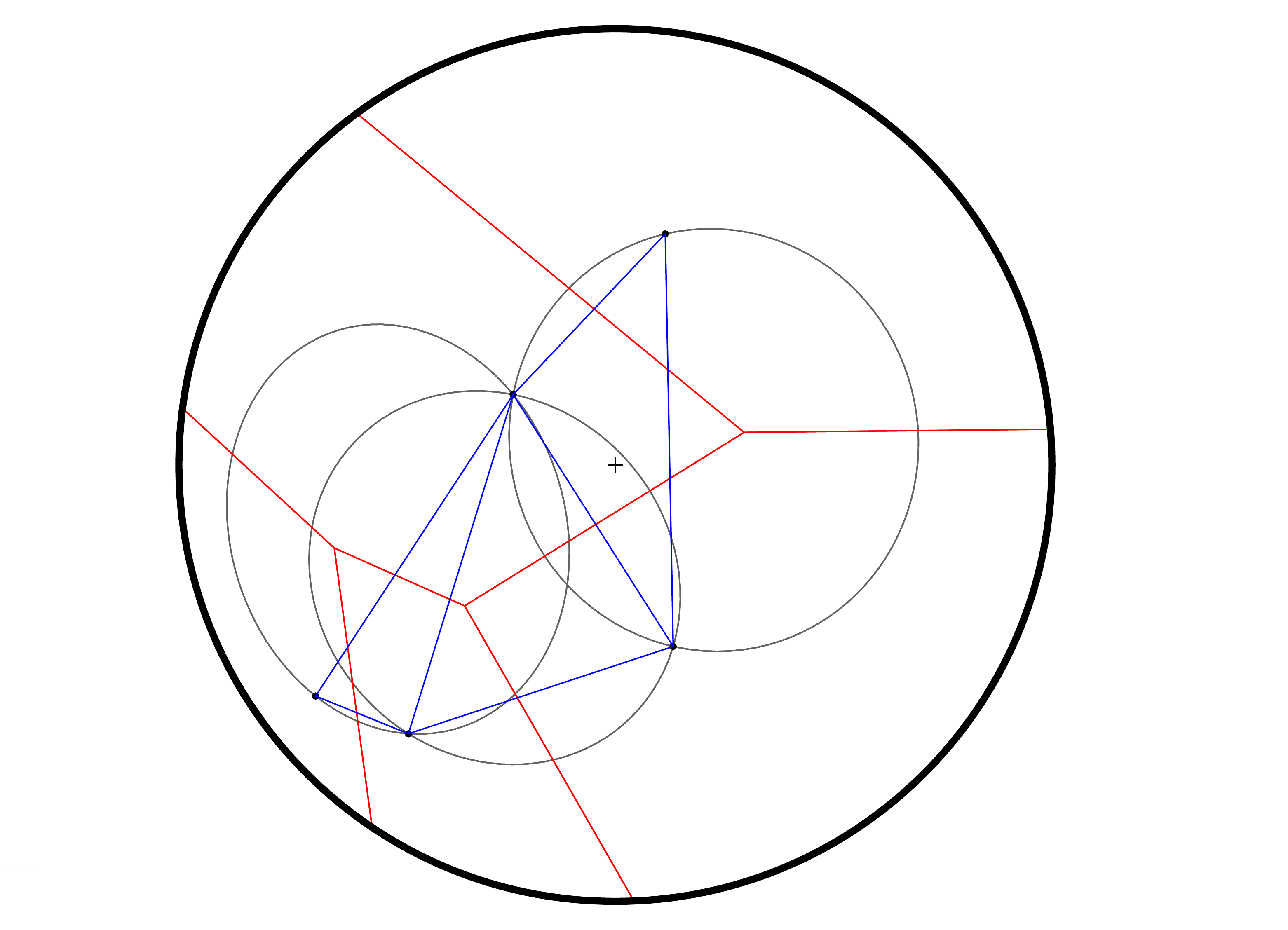}} \\
\fbox{\includegraphics[width=0.4\columnwidth]{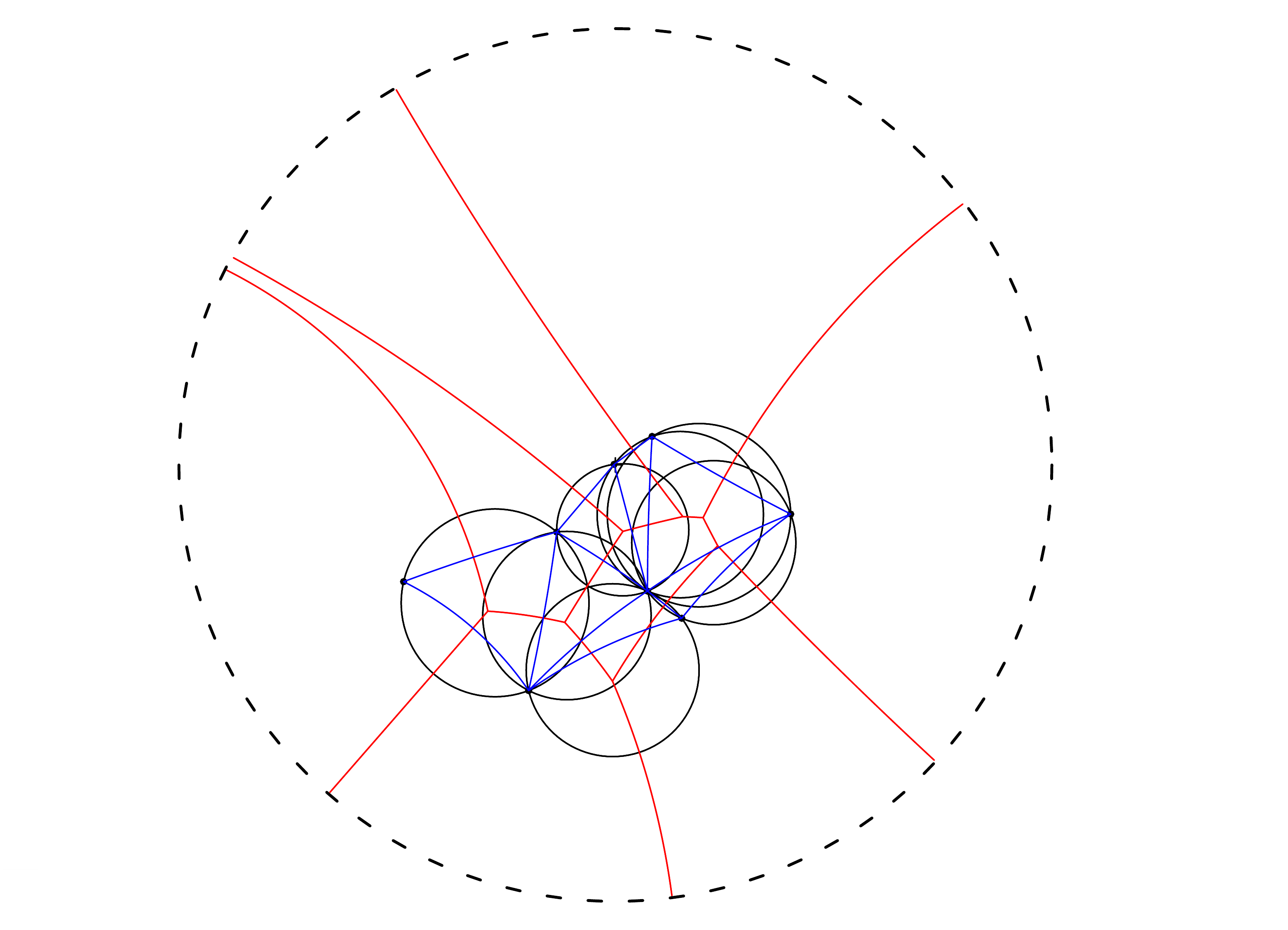}} &
\fbox{\includegraphics[width=0.4\columnwidth]{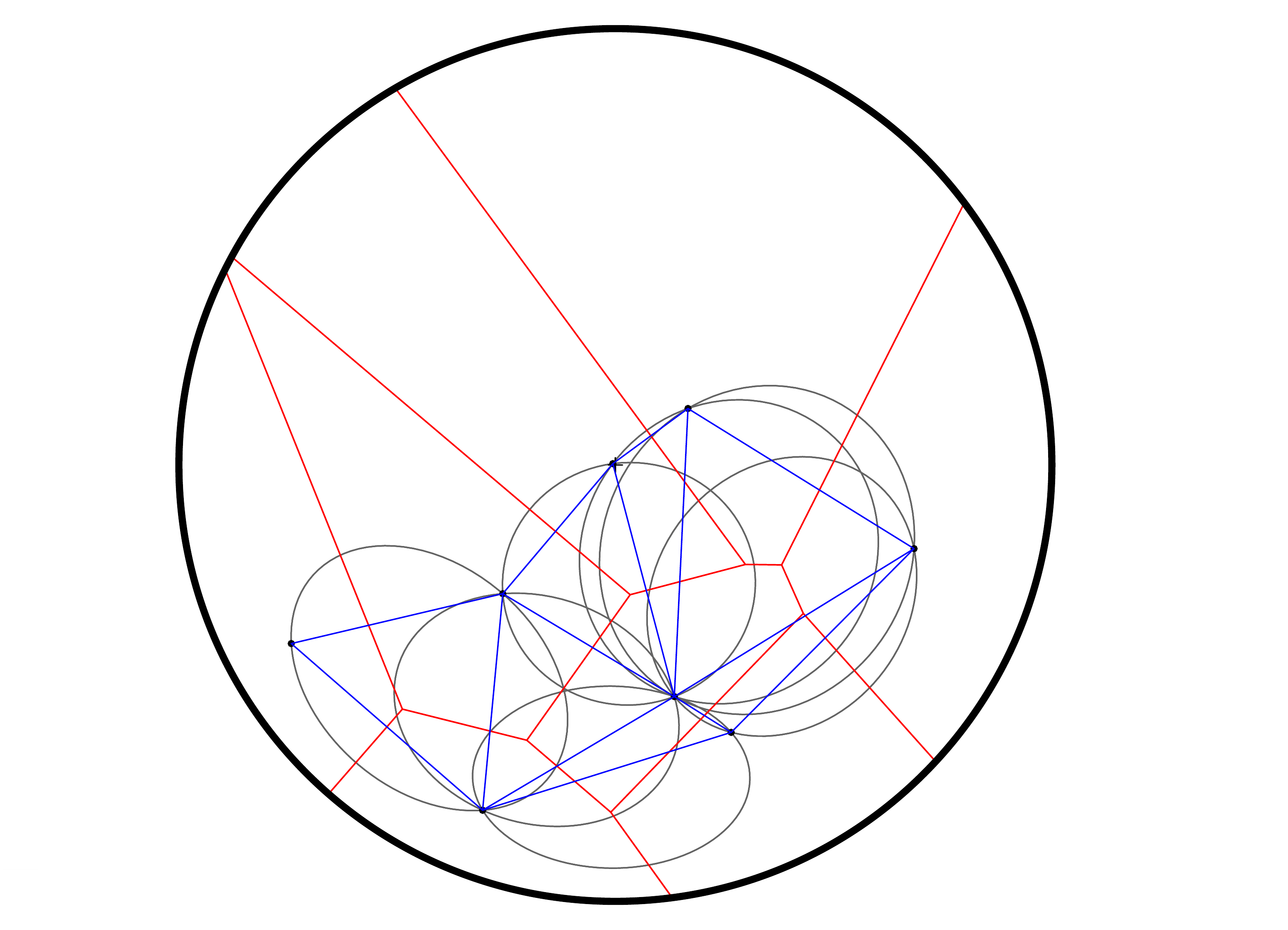}} \\
\fbox{\includegraphics[width=0.4\columnwidth]{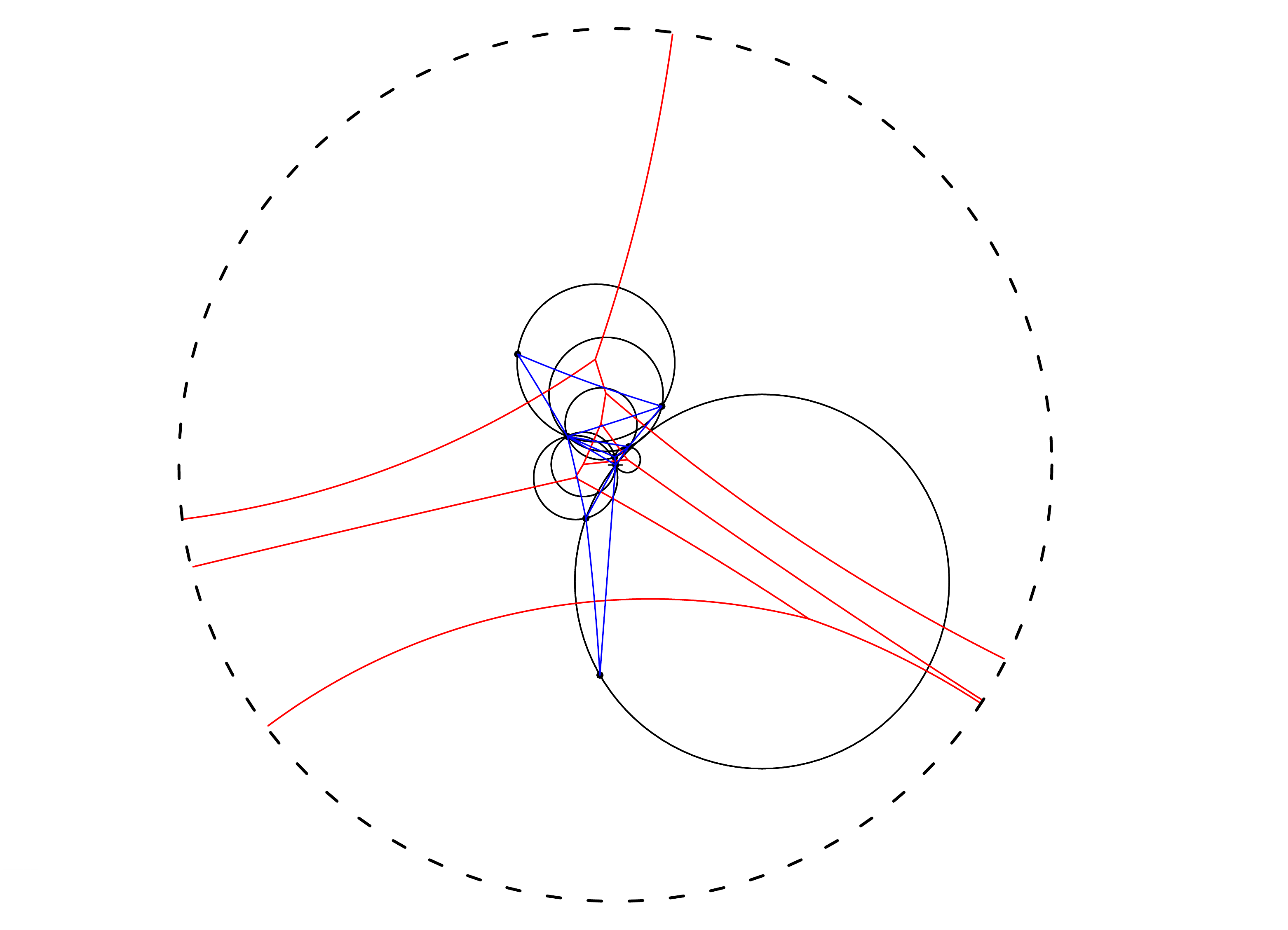}} &
\fbox{\includegraphics[width=0.4\columnwidth]{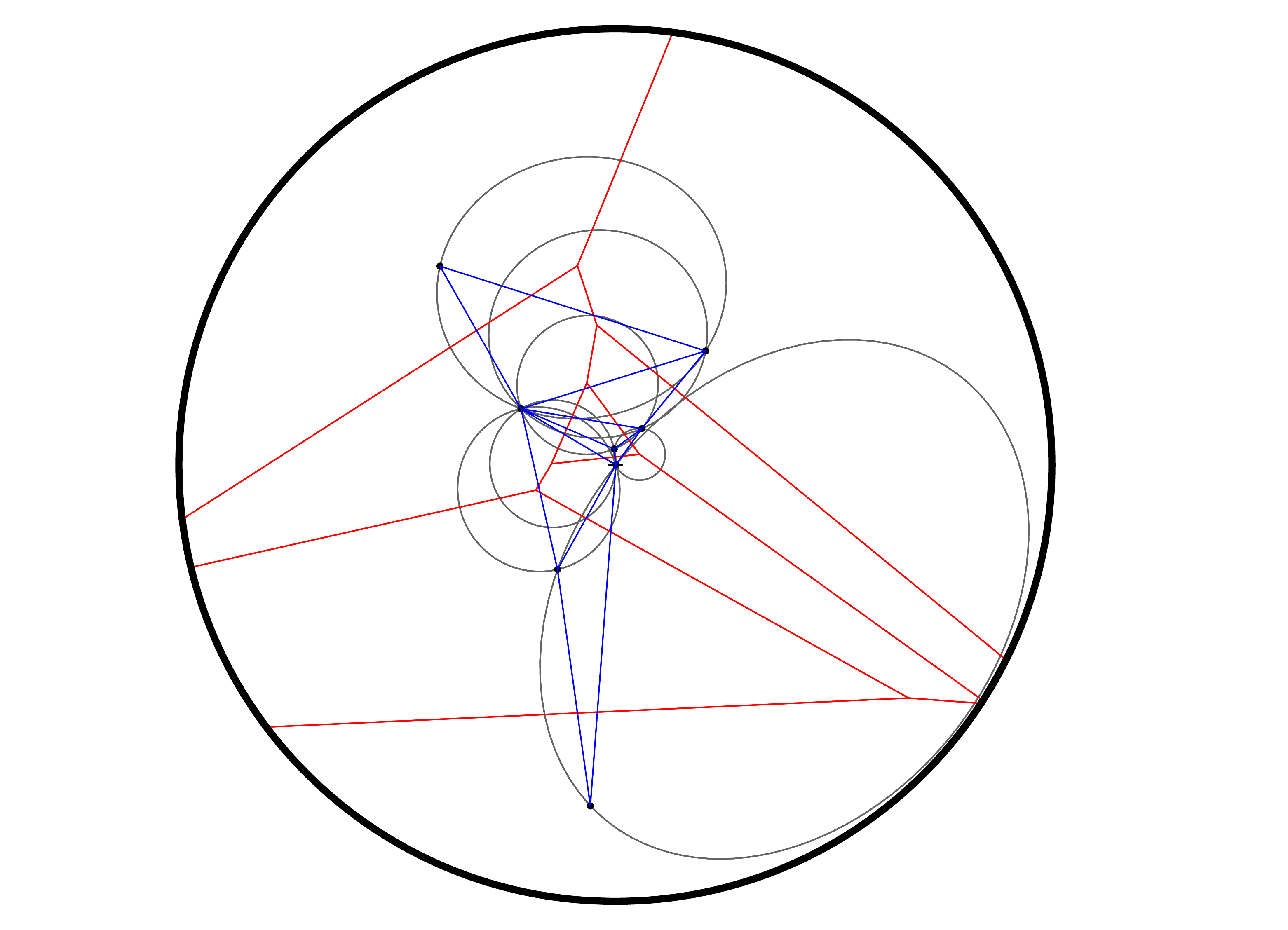}} \\
\end{tabular}

\caption{Delaunay triangles of the hyperbolic Delaunay complex satisfy the empty circumscribing sphere property. The empty sphere centers are located on the Voronoi $T$-junction vertices. The hyperbolic spheres are displayed as ordinary Euclidean sphere (with displaced center) in the Poincar\'e model (left column) and as ellipsoids (with displaced center) in the Klein model (right column). The centers of the empty hyperbolic spheres are located at the Voronoi $T$-junctions.}%

\label{fig:hyperbolicemptysphere}%
\end{figure}

\subsection{The dual Voronoi diagrams on the Cauchy dually flat manifold}\label{sec:dualVor}

The dual Cauchy Voronoi diagrams with respect to the flat divergence $D_\flat$ (and dual reverse flat divergence $D_\flat^*$ which corresponds to a dual Bregman-Tsallis divergence) of \S\ref{sec:dfs} amount to 
calculate 2D dual Bregman Voronoi diagrams~\cite{BVD-2010}.
We get the following dual bisectors:
The primal bisector with respect to the dual flat divergence is:
\begin{eqnarray}
\mathrm{Bi}_{D_\flat}(p_{\lambda_1}:p_{\lambda_2}) &=& \left\{ p_\lambda \ :\  
D_\flat[p_{\lambda_1}:p_{\lambda}]= D_\flat[p_{\lambda_2}:p_{\lambda}] \right\},\\
&=& \left\{ \lambda \ :\  \delta(l_1,s_1;l,s)=\delta(l_2,s_2;l,s)  \right\}.
\end{eqnarray}
Thus this primal bisector with respect to the flat divergence corresponds to the hyperbolic bisector of the 
Fisher-Rao distance/chi square/ KL divergences:
\begin{equation}
\mathrm{Bi}_{D_\flat}(p_{\lambda_1}:p_{\lambda_2})=\mathrm{Bi}_{\rho_\FR}(p_{\lambda_1}:p_{\lambda_2})
=\mathrm{Bi}_{D_\KL}(p_{\lambda_1}:p_{\lambda_2})
=\mathrm{Bi}_{D_{\chi^2}}(p_{\lambda_1}:p_{\lambda_2}).
\end{equation}

The dual bisector with respect to the dual flat divergence (reverse Bregman-Tsallis divergence) is:
\begin{eqnarray}
\mathrm{Bi}_{D_\flat}^*(p_{\lambda_1}:p_{\lambda_2}) &=& \left\{ p_\lambda \ :\  
D_\flat[p_{\lambda}:p_{\lambda_1}]=D_\flat[p_{\lambda}:p_{\lambda_2}] \right\},\\
&=& \left\{ \lambda \ :\   \|\lambda-\lambda_1\| = \| \lambda-\lambda_2 \| \right\}.
\end{eqnarray}
That is, the dual bisector corresponds to an ordinary Euclidean bisector:
\begin{equation}
\mathrm{Bi}_{D_\flat}^*(p_{\lambda_1}:p_{\lambda_2})=\mathrm{Bi}_{\rho_E}(p_{\lambda_1},p_{\lambda_2}).
\end{equation}

Notice that $\mathrm{Bi}_{D_\flat}^*(p_{\lambda_1}:p_{\lambda_2}) =\mathrm{Bi}_{D_\flat^*}(p_{\lambda_1}:p_{\lambda_2})$.

To summarize, one primal bisector coincides with the Fisher-Rao bisector while the dual bisector amounts to the ordinary Euclidean bisector. 

\begin{Theorem}\label{thm:vordfs}
The dual Cauchy Voronoi diagrams with respect to the flat divergence can be calculated efficiently in $\Theta(n\log n)$-time.
\end{Theorem} 

The construction of 2D Bregman Voronoi diagrams is described in~\cite{BVD-2010}.

\subsection{The  Cauchy Voronoi diagrams with respect to $\alpha$-divergences}

The dual bisectors with respect to the $\alpha$-divergences between any two parametric probability densities $p_{\lambda_1}(x)$ 
and $p_{\lambda_2}(x)$ are

\begin{eqnarray}
\mathrm{Bi}_{I_\alpha}(p_{\lambda_1}:p_{\lambda_2}) &=& \left\{ p_\lambda \ :\  
I_\alpha[p_{\lambda_1}:p_{\lambda}]= I_\alpha[p_{\lambda_2}:p_{\lambda}] \right\},\\
&=& \left\{ \lambda \ :\  C_\alpha(p_{\lambda_1};p_{\lambda})=C_\alpha(p_{\lambda_2};p_{\lambda})  \right\},
\end{eqnarray}
and
\begin{eqnarray}
\mathrm{Bi}_{I_\alpha}^*(p_{\lambda_1}:p_{\lambda_2}) &=& \left\{ p_\lambda \ :\  
I_\alpha[p_{\lambda}:p_{\lambda_1}]= I_\alpha[p_{\lambda}:p_{\lambda_2}] \right\},\\
&=& \mathrm{Bi}_{I_{1-\alpha}}(p_{\lambda_1}:p_{\lambda_2}).
\end{eqnarray}

It is an open problem to prove {\em when} the dual $\alpha$-bisectors coincide for the Cauchy family.
We have shown it is the case for the $\chi^2$-divergence and the KL divergence. 
In theory, the Risch semi-algorithm~\cite{risch1970solution} allows one to answer whether a definite integral has a closed-form formula or not. However, the Risch semi-algorithm is only a semi-algorithm as it requires to implement an oracle  to check whether some mathematical expressions are equivalent to zero or not.

\section{Conclusion}\label{sec:Concl}

\begin{table}
\centering

\begin{tabular}{ll}
 Formula & Voronoi \\ \hline
$D_{\chi^2}[p_{l_1,s_1},p_{l_2,s_2}]=\frac{(l_2-l_1)^2+(s_2-s_1)^2}{2s_1s_2}$ & $\Vor_{D_{\chi^2}}$ hyperbolic Voronoi\\
$\rho_\FR[p_{l_1,s_1},p_{l_2,s_2}]= \frac{1}{\sqrt{2}}\arccosh (1+D_{\chi^2}[p_{l_1,s_1},p_{l_2,s_2}])$ & $\Vor_{\rho_\FR}$ hyperbolic Voronoi\\
$D_{\KL}[p_{l_1,s_1},p_{l_2,s_2}]=\log\left(1+\frac{1}{2}D_{\chi^2}[p_{l_1,s_1},p_{l_2,s_2}]\right)$  & $\Vor_{D_{\KL}}$ hyperbolic Voronoi\\
$\rho_{\KL}[p_{l_1,s_1},p_{l_2,s_2}]=\sqrt{D_{\KL}[p_{l_1,s_1},p_{l_2,s_2}]}$ (metric) & $\Vor_{\rho_{\KL}}$ hyperbolic Voronoi\\
$D_{\flat}[p_{l_1,s_1},p_{l_2,s_2}]=2\pi s_2 D_{\chi^2}[p_{l_1,s_1},p_{l_2,s_2}]$ &  Bregman Voronoi:\\
 &  $\Vor_{D_{\flat}}$ hyperbolic Voronoi,  $\Vor_{D_{\flat}}^*$ Euclidean Voronoi.\\ \hline
\end{tabular}

\caption{Summary of the main closed-form formula for the statistical distances between Cauchy densities and their induced Voronoi diagrams.}
\label{tb:dist}
\end{table}

In this paper, we have considered the construction of Voronoi diagrams of finite sets of Cauchy distributions with respect to some common statistical distances. Since statistical distances can potentially be asymmetric, we defined the dual Voronoi diagrams with respect to the forward and reverse/dual statistical distances.
From the viewpoint of information geometry~\cite{IG-2016}, we have reported the construction of two types of geometry  on the Cauchy manifold: 
(1) The {\em invariant $\alpha$-geometry} equipped with the Fisher metric tensor $g_\FR$ and the skewness tensor $T$ from which we can build a family of  pairs of torsion-free affine connections coupled with the metric, and (2) a {\em dually flat geometry} induced by a Bregman generator defined by the free energy $F_q$ of the $q$-Gaussians (instantiated to $q=2$ when dealing with the Cauchy family).
The metric tensor of the latter geometry is called the $q$-Fisher information metric, and is a Riemannian conformal metric of the  Fisher information metric. 
We have shown that the Fisher-Rao distance amount to a scaled hyperbolic distance in the Poincar\'e upper plane model (Proposition~\ref{prop:FRCauchy}), and that all Amari's $\alpha$-geometries~\cite{IG-2016} coincide with the Fisher-Rao geometry since the cubic tensor vanishes, thus yielding a hyperbolic manifold of negative constant scalar curvature $\kappa=-2$ for the Cauchy $\alpha$-geometric manifolds.
We noticed that the Fisher-Rao distance and the KL divergence can be expressed as a strictly increasing function of the chi square divergence.
Then we explained how to conformally flatten the curved Fisher-Rao geometry to obtain a dually flat space where the flat divergence amounts to a canonical Bregman divergence  built from Tsallis' quadratic entropy (Theorem~\ref{thm:BDqgaussian}).
We reported the Hessian metrics of the dual potential functions of the dually flat space, and showed that there are other alternative choices for building Hessian structures~\cite{Matsuzoe-2014}. 
Table~\ref{tb:dist} summarizes the various closed-form formula of statistical dissimilarities obtained for the Cauchy family.
We proved that the square root of the KL divergence between any two Cauchy distributions  is a metric distance (Theorem~\ref{thm:sqrtKL}) in general, and more precisely a Hilbertian metric for the scale Cauchy families (Theorem~\ref{thm:sqrtKLscale}).
It follows that the Cauchy Voronoi diagram for the Fisher-Rao distance coincides with the Voronoi diagram with respect to the  KL divergence or the chi square divergence (Figure~\ref{fig:VDcomparisons}). We showed how to build this hyperbolic Cauchy diagram from an equivalent hyperbolic Voronoi diagram on the corresponding location-scale parameters (see also Appendix~\ref{sec:hpd}). 
Then we proved that the dual hyperbolic Cauchy Delaunay complex is Fisher orthogonal to the Fisher-Rao  hyperbolic Cauchy Voronoi diagram (Theorem~\ref{thm:CVDortho}). The dual Voronoi diagrams with respect to the dual flat divergences can be built from the corresponding dual Bregman-Tsallis divergences with the primal Voronoi diagram coinciding with the hyperbolic Voronoi diagram and the dual diagram coinciding with the ordinary Euclidean Voronoi diagram (Figure~\ref{fig:VDcomparisons}).

\begin{figure}
\centering
\centering
\begin{tabular}{cc}
\includegraphics[bb=0 0 512 512,width=0.45\columnwidth]{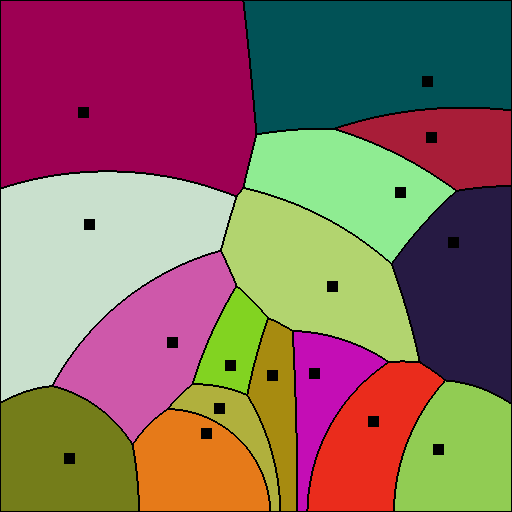} &
\includegraphics[bb=0 0 512 512,width=0.45\columnwidth]{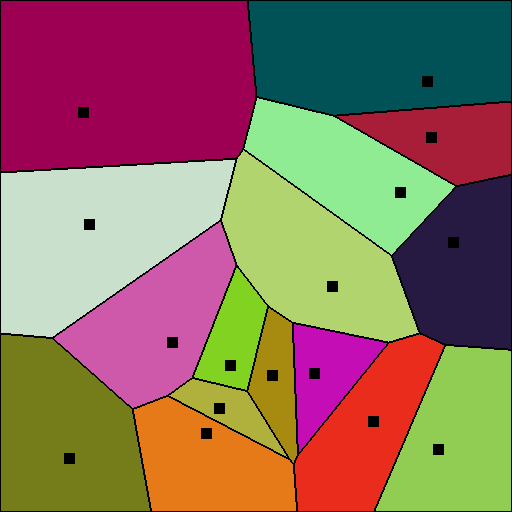} \\
$\Vor_{\rho_\FR}=\Vor_{\rho_\KL}=\Vor_{\rho_\chi^2}=\Vor_{D_\flat}$ & $\Vor_{D_\flat^*}=\Vor_{\rho_E}$.
\end{tabular}

\caption{Voronoi diagrams of a set of Cauchy distributions with respect to the Fisher-Rao (FR) distance $\rho_\FR$, the Kullback-Leibler (KL) divergence $D_\KL$, the $\chi^2$-divergence $D_{\chi^2}$, and the asymmetric Bregman-Tsallis flat divergence $D_\flat$.}
\label{fig:VDcomparisons}

\end{figure}




\bibliographystyle{plain}
\bibliography{KLCauchyBIBV3}

\appendix

\section{Klein hyperbolic Voronoi diagram from a clipped power diagram}\label{sec:hpd}

We concisely recall the efficient construction of the hyperbolic Voronoi diagram in the Klein disk model~\cite{HVD-2010}.
Let $\calP=\{p_1, \ldots, p_n\}$ be a set of $n$ points in the $d$-dimensional open unit ball domain $\bbD=\left\{x\in\bbR^d \ :\ \|x\|_2<1\right\}$, where $\|\cdot\|_2$ denotes the Euclidean $\ell_2$-norm.
The hyperbolic distance between two points $p$ and $q$ is expressed in the Klein model as follows:
\begin{equation}
\rho_{K}(p,q) \eqdef 
\operatorname{arccosh}\left( \frac{1-\langle p,q\rangle}{\sqrt{\left(1-\|p\|^{2}_2\right)\left(1-\|q\|^{2}_2\right)}} \right).
\end{equation}

It follows that the {\em Klein bisector} between any two points in the Klein disk is an hyperplane (affine equation) clipped to $\bbD$:
\begin{equation}\label{eq:bik}
\mathrm{Bi}_{\rho_{K}}({\lambda_1}:{\lambda_2}) = \left\{ \lambda\in\mathbb{D} \ :\ \lambda^\top \left(\sqrt{1-\|\lambda_1\|^{2}_2} \lambda_2-\sqrt{1-\|\lambda_2\|^{2}_2} \lambda_1\right) +\sqrt{1-\|\lambda_2\|^{2}_2}-\sqrt{1-\|\lambda_1\|^{2}_2}=0\right\} . 
\end{equation}

The Klein bisector is a hyperplane (i.e., line in 2D) restricted to the disk   domain $\bbD$.
A Voronoi diagram is said {\em affine}~\cite{boissonnat2006curved} when all  bisectors are hyperplanes.
It is known that {\em affine Voronoi diagrams} can be constructed from equivalent {\em power diagrams}~\cite{boissonnat2006curved}.
Thus the  Klein hyperbolic Voronoi diagram is equivalent to a {\em clipped power diagram}:
\begin{equation}
\Vor_{\rho_{K}}(\calP) = \Vor_{D_\PD}(\calS)\cap \bbD,
\end{equation}
 where 
\begin{equation}
D_\PD(\sigma,x)\eqdef\|x-c\|^2-w,
\end{equation}
 denotes the {\em power ``distance''} between 
 a point $x$ (and more generally a  {\em weighted point}~\cite{clippedPD-2005} when the weight can be negative) 
 to a sphere 
 $\sigma=(c,w)$, 
and $\calS=\{\sigma_1=(c_1,w_1),\ldots, \sigma_n=(c_n,w_n) \}$ is the equivalent set of weighted points.
The power distance is a {\em signed distance} since we have the following property: 
$D_\PD(\sigma,x)<0$ iff $x\in \mathrm{int}(\sigma)$, i.e., the point $x$ falls inside the sphere $\sigma=\{x \ :   \|x-c\|^2_2 = w \}$.
The {\em power bisector} is a hyperplane of equation
\begin{equation}\label{eq:bipd}
\mathrm{Bi}_\PD(\sigma_i,\sigma_j) = \left\{ x\in\bbR^d \ :\ 2x^\top (c_j-c_i)+w_i-w_j=0 \right\}
\end{equation}
Notice that by shifting all weights by a predefined constant $a$, we obtain the same power bisector since $(w_i+a)-(w_j+a)=w_i-w_j$ is  kept invariant.
Thus we may consider without loss of generality that all weights are non-negative, and that the weighted points correspond to spheres with non-negative radius $r_i^2=w_i$.

By identifying Eq.~\ref{eq:bik} with Eq.~\ref{eq:bipd}, we get the following equivalent spheres $\sigma_i=(c_i,w_i)$~\cite{HVD-2010} for the points in the Klein disk:
\begin{eqnarray}
c_i &=& \frac{{p}_{i}}{2 \sqrt{1-\left\|{p}_{i}\right\|^{2}}},\\
w_i &=& \frac{\left\|{p}_{i}\right\|^{2}_2}{4\left(1-\left\|{p}_{i}\right\|^{2}_2\right)}-\frac{1}{\sqrt{1-\left\|{p}_{i}\right\|^{2}_2}}.
\end{eqnarray}
We can then shift all weights by the constant $a=\min_{i\in\{1,\ldots, n\}} w_i$ so that $w_i'=w_i+a\geq 0$.

Thus the Klein hyperbolic Voronoi diagram is a {\em power diagram} clipped to the unit ball $\bbD$~\cite{nielsen1998grouping,clippedPD-2005,clippedVD-2013}.
In computational geometry~\cite{BY-1998}, the power diagram can be calculated from the intersection of $n$ halfspaces by lifting the spheres $\sigma_i$ to  corresponding halfspaces $H_i^+$ of $\bbR^{d+1}$ as follows:  
Let $\calF=\{ (x,z)\in \bbR^{d+1} \ :\ z\geq \sum_{i=1}^d x_i^2\}$ be the epigraph of the paraboloid function, and $\partial\calF$ denotes its boundary.
We lift a point $x\in\bbR^d$ to $\partial\calF$ using the upper arrow operator $x^\uparrow=(x,z=\sum_{i=1}^d x_i^2)$, and we project orthogonally a point $(x,z)$ of the potential function $\calF$ by dropping its last $z$-coordinate so that we have $\downarrow(x^\uparrow)=x$.
Now, when we lift a sphere $\sigma=(c,w)$ to $\calF$, the set of lifted points $\sigma^\uparrow$ all belong to a hyperplane $H_\sigma$, called the {\em polar hyperplane} of equation:
\begin{equation}
H_\sigma: z=2c^\top x-c^\top c+w.
\end{equation}
Let $H_\sigma^+$ denote the upper halfspace with bounding hyperplane $H_\sigma$: $H_\sigma^+: z \geq 2c^\top x-c^\top c+w$.
Then one can show~\cite{BY-1998} that $\Vor_{D_\PD}(\calS)$ is obtained as the vertical projection $\downarrow$ of the intersection of  all these  polar halfspaces $\calH_i$ with $\partial\calF$:
\begin{equation}
\Vor_{D_\PD}(\calS) = \downarrow\left( \left(\cap_{i=1}^n \calH_i^+\right) \cap \partial\calF  \right).
\end{equation}
Transforming back and forth non-vertical $(d+1)$-dimensional hyperplanes to corresponding $d$-dimensional spheres allows one to design  various efficient algorithms, e.g., computing the intersection or the union of spheres~\cite{BY-1998}, useful primitives for molecular chemistry~\cite{VorOkabe-2009}.

Let $H_\bbD^-$ denote the lower halfspace (containing the origin $(x=0,z=0)$) supported by the polar hyperplane associated to the boundary sphere of the disk domain $\bbD$.
Computing the clipped power diagram $\Vor_{D_\PD}(\calS)\cap\bbD$ can be done equivalently as follows:
\begin{eqnarray}
\Vor_{D_\PD}(\calS)\cap\bbD &=& \downarrow\left(\left( \left( \cap_{i=1}^n \calH_i^+  \right)  \cap \partial\calF \right) \cap H_\bbD^-\right),\\
&=& 
\downarrow\left(\left( \left( \cap_{i=1}^n \calH_i^+  \right)   \cap H_\bbD^- \right) \cap \partial\calF \right),\label{eq:intersection}
\end{eqnarray}
using the commutative property of the set intersection.

The advantage of the method of Eq.~\ref{eq:intersection} is that we begin to clip the power diagram using $H_\bbD^-$  before explicitly calculating it.
Indeed, we first compute the {\em intersection polytope} of $n+1$ hyperplanes 
$\calP_K \eqdef \left(\cap_{i=1}^n \calH_i^+\right) \cap H_\bbD^-$.
Then we project down orthogonally  the intersection of $\calP_K$ with $\partial\calF$ to get the clipped power diagram equivalent to the hyperbolic Klein Voronoi diagram:
\begin{equation}
\Vor_{\rho_{K}}(\calP) = \downarrow\left(\calP_K \cap \partial\calF\right).
\end{equation}

By doing so, we potentially reduce the algorithmic complexity by avoiding to compute some of the vertices of  $\calP_\PD \eqdef \left(\cap_{i=1}^n \calH_i^+\right)$  whose orthogonal projection fall outside the domain $\bbD$.

More generally, a Bregman Voronoi diagram~\cite{BVD-2010} can be calculated equivalently as a power diagram (and intersection of $d+1$-dimensional halfspaces) using an arbitrary smooth and strictly convex potential function $F$ instead of the the paraboloid potential function of Euclidean geometry~\cite{HVD-2010}.
The non-empty intersection of halfspaces can in turn be calculated as an equivalent {\em convex hull}~\cite{BY-1998}.
Thus we can compute in practice the hyperbolic Voronoi diagram in the Klein model using the Quickhull algorithm~\cite{barber1996quickhull}.

\section{Symbolic calculations with a computer algebra system}\label{sec:maxima}

We use the open source computer algebra system {\sc Maxima}\footnote{Can be freely downloaded at \url{http://maxima.sourceforge.net/}} to calculate the gradient (partial derivatives) and Hessian of the deformed log-normalizer,  and some definite integrals based on the Cauchy location-scale densities.

\begin{verbatim}
/* Written in Maxima */
assume(s>0);
CauchyStd(x) := (1/(%pi*(x**2+1)));
Cauchy(x,l,s) := (s/(%pi*((x-l)**2+s**2)));
/* check that we get a probability density (=1) */
integrate(Cauchy(x,l,s),x,-inf,inf);
/* calculate the the deformed log-normalizer */
logC(u):=1-(1/u);
logC(Cauchy(x,l,s));
ratsimp(%);
/* calculate partial derivatives of the deformed log-normalizer */
theta(l,s):=[2*%pi*l/s,-%pi/s];
F(theta):=(-%pi**2/theta[2])-(theta[1]**2/(4*theta[2]))-1;
derivative(F(theta),theta[1],1);
derivative(F(theta),theta[2],1);
/* calculated definite integrals */
assume(s1>0);
assume(s2>0);
integrate(Cauchy(x,l2,s2)**2,x,-inf,inf);
integrate(Cauchy(x,l2,s2)**2/Cauchy(x,l1,s1),x,-inf,inf);
\end{verbatim}

We calculate the function $\theta(\eta)$ by solving the following system of equations:
\begin{verbatim}
solve([-t1/(2*t2)=e1, (%pi/t2)**2+ (t1/t2)**2/4=e2],[t1, t2]);
\end{verbatim}

The Hessian metrics of the dual potential functions $F$ and $F^*$ (denoted by $G$ in the code) can be calculated as follows:
\begin{verbatim}
F(theta):=(-%pi**2/theta[2])-(theta[1]**2/(4*theta[2]))-1;
hessian(F(theta),[theta[1], theta[2]]);
G(eta):=1-2*%pi*sqrt(eta[2]-eta[1]**2);
hessian(G(eta),[eta[1], eta[2]]);
\end{verbatim}

The plot of the Fisher-Rao to the square root KL divergence can be plotted using the following commands:
\begin{verbatim}
t(u):=sqrt(log((1/2)+(1/2)*cosh(sqrt(2)*u)));
plot2d(t(u)/u,[u,0,10]);
\end{verbatim}

Symbolic calculations for the $\alpha$-Chernoff coefficient between two Cauchy distributions prove that
 the  $\alpha$-Chernoff coefficient is symmetric for $\alpha=3$ and $\alpha=4$ as exemplified by the {\sc Maxima} code below:

\begin{verbatim}
assume(s1>0);
assume(s2>0);
assume(s>0);
CauchyStd(x) := (1/(%pi*(x**2+1)));
Cauchy(x,l,s) := (s/(%pi*((x-l)**2+s**2)));
/* closed-form */
a: 3;
integrate((Cauchy(x,l2,s2)**a) * (Cauchy(x,l1,s1)**(1-a)),x,-inf,inf);
term1(l1,s1,l2,s2):=ratsimp(%);
integrate((Cauchy(x,l2,s2)**(1-a)) * (Cauchy(x,l1,s1)**(a)),x,-inf,inf);
term2(l1,s1,l2,s2):=ratsimp(%);
/* Is the a-divergence symmetric? */
term1(l1,s1,l2,s2)-term2(l1,s1,l2,s2);
ratsimp(%);
\end{verbatim}

\end{document}